\documentclass[11pt]{article}

\usepackage[letterpaper,margin=1in]{geometry}

\usepackage{amsmath,amssymb,amsfonts}
\usepackage{amsthm}
\usepackage{graphicx}
\usepackage{subcaption}
\usepackage{cite}
\usepackage{dsfont}
\usepackage{url}
\usepackage[hidelinks]{hyperref}

\newtheorem{theorem}{Theorem}
\newtheorem{lemma}{Lemma}

\newtheorem{definition}{Definition}

\newtheorem{corollary}{Corollary}

\title{\LARGE\bf
When Freshness Is Not Enough: Distribution-Aware Age of Information for Networked LQR Control
}

\author{
Abdullah Y. Etcibasi, C. Emre Koksal, and Eylem Ekici
\thanks{Abdullah Y. Etcibasi, C. Emre Koksal, and Eylem Ekici are with the Department of Electrical and Computer Engineering, The Ohio State University, Columbus, OH 43210, USA.}%
}
\date{}

\begin{document}

\maketitle
\thispagestyle{empty}
\pagestyle{empty}

\begin{abstract}
Age of Information (AoI) has become a central metric for the design of wireless update systems, especially in applications where fresh measurements support tracking, estimation, and control. Despite its popularity, the use of mean AoI or peak AoI as a surrogate for closed-loop performance is often motivated by intuition rather than by a control-theoretic derivation. This paper examines whether minimizing the mean AoI is in fact optimal for networked control systems. For scalar linear time-invariant systems with delayed intermittent updates, we show that, under state-independent scheduling policies, the infinite-horizon LQR tracking problem reduces to an optimization over the distribution of inter-scheduling intervals. The resulting objective depends on higher-order statistical moments, and in unstable or correlated regimes on exponential moments, of the inter-scheduling process rather than only on its mean. Consequently, policies with identical mean AoI can induce substantially different tracking costs. We further extend the analysis to disturbances with exponentially decaying autocorrelation and derive equivalent cost formulations that expose the role of the full interval distribution. Finally, we evaluate the theory using real vehicle trajectories from the NGSIM US-101 dataset. The empirical results match the predicted performance trends, demonstrating that mean AoI alone is insufficient for control-oriented network design.
\end{abstract}

\section{Introduction}

Cyber-physical systems (CPSs) increasingly rely on networked sensing, communication, and actuation to close feedback loops over shared and resource-constrained channels. Examples include automated driving, cooperative mobility, remote state estimation, industrial automation, and wireless control of robotic platforms. This setting is naturally modeled within networked control systems (NCSs), where stability and performance must be understood in the presence of communication delays, packet drops, lossy links, and shared-medium scheduling constraints~\cite{hespanha2007survey,zhang2001stability,schenato2007foundations,walsh2001scheduling,park2017wireless}. In these systems, the communication layer determines when measurements are delivered to the controller, while the control layer determines how these measurements are used to regulate the physical process. Since wireless channels are rate-limited, delayed, and often shared by multiple agents, one cannot transmit all measurements at all times. The central design question is therefore not only how to stabilize or track a dynamical system, but also how to allocate limited communication budget so that the resulting closed-loop performance is optimized.

Age of Information (AoI) has emerged as a standard network-layer metric for quantifying information freshness. Introduced for real-time status-update systems in~\cite{kaul2012real} and subsequently developed across queueing and wireless update models~\cite{sun2017update,yates2019age,yates2021age}, AoI measures how old the most recently available information is at the receiver, rather than only the latency of a delivered packet. This distinction has made AoI particularly attractive in status-update systems, where stale measurements may be less useful even if they were delivered reliably. As a result, a large body of work has studied the minimization of average AoI, peak AoI, and related freshness metrics in queues, wireless scheduling systems, random access networks, and energy-constrained communication architectures. In parallel, control-aware communication has investigated event-triggered and self-triggered sampling~\cite{astrom2002comparison,tabuada2007event,heemels2012introduction,heemels2010networked}, AoI- and value-of-information-based scheduling~\cite{ayan2019age,wang2021value,soleymani2021value,aggarwal2023weighted}, and full-loop AoI-oriented CPS designs~\cite{lu2023full}.

For networked control systems, however, the connection between AoI minimization and control performance remains incomplete. The objective of the controller is not freshness itself, but rather regulation, tracking, or a linear quadratic regulator (LQR)-type cost induced by the plant dynamics and disturbances. Existing work often adopts mean AoI or peak AoI as a design criterion for these objectives, implicitly assuming that fresher information necessarily yields better closed-loop behavior. This assumption is natural, but it is not generally sufficient: two scheduling policies may have the same average age while producing very different distributions of inter-scheduling intervals, and those distributional differences can be amplified by unstable dynamics, delayed feedback, or temporally correlated disturbances. Consequently, the relevant question is not whether lower mean AoI is intuitively desirable, but whether the LQR objective actually reduces to a function of mean AoI alone. This paper shows that it does not.

Motivated by this distinction, we do not impose an AoI metric a priori. Instead, we derive the scheduling objective from the LQR tracking cost itself. The main contributions are as follows. First, we reformulate the infinite-horizon LQR tracking cost under state-independent scheduling and show that it depends fundamentally on higher-order moments and, in unstable or correlated regimes, exponential moments of the inter-scheduling interval (equivalently, AoI) distribution. Second, for i.i.d. disturbances, we prove that the optimal state-independent policy is the as-periodic-as-possible scheduler, which concentrates the interval process around its mean. Third, we extend the formulation to temporally correlated disturbances with exponentially decaying autocorrelation and derive the corresponding equivalent optimization problems, revealing how correlation modifies the moment structure of the cost. Fourth, we validate the theoretical claims using a data-driven car-following simulation based on the NGSIM US-101 vehicle trajectory dataset \cite{ngsim_us101_2016}, showing that reducing mean AoI alone is insufficient and that the theoretical cost captures the observed empirical trends.

The remainder of the paper is organized as follows. Section~\ref{sec:Sys_Mdl} introduces the networked control model, the information structures, and the control-scheduling formulation, and shows how the LQR objective reduces to an estimation-error scheduling problem. Section~\ref{sec:AoI_iid_Noise} derives the AoI-equivalent objective under i.i.d. disturbances and identifies the as-periodic-as-possible scheduler as optimal within the state-independent class. Section~\ref{sec:AoI_Correlated_Noise} extends the analysis to temporally correlated disturbances and shows how correlation introduces additional exponential-moment terms. Section~\ref{sec:Real_Sim} validates the theoretical characterization using data-driven car-following simulations based on NGSIM US-101 trajectories.

\section{System Model}
\label{sec:Sys_Mdl}

\begin{figure}[tbp]
  \centering                
  \includegraphics[width=0.7\textwidth]{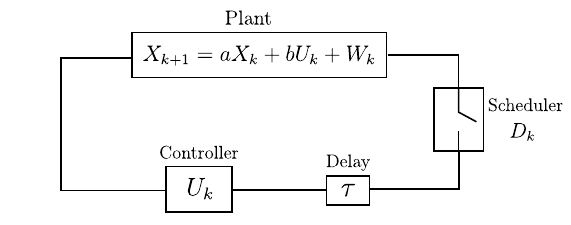}
  \caption{Block diagram of the closed-loop control system}
  \label{fig:Sys_Model}      
\end{figure}

Consider the single-loop, discrete-time, scalar LTI system illustrated in Fig.~\ref{fig:Sys_Model}. 
The plant evolves according to
\begin{equation}
    X_{k+1} = a X_k + b U_k + W_k,
    \label{eqn:system_Model}
\end{equation}
where $X_k$ denotes the system state, $U_k$ the control input, and $W_k$ an additive disturbance. 
We consider two disturbance models for the sequence $\{W_k\}$. In Section~\ref{sec:AoI_iid_Noise}, we assume that $\{W_k\}$ is i.i.d., zero mean, and has finite variance $\sigma_W^2$ (not necessarily Gaussian). In Section~\ref{sec:AoI_Correlated_Noise}, we extend the analysis to temporally correlated disturbances. Controller receives intermittent information through the scheduling mechanism with a fixed \(\tau \geq 1\) delay, so that the observation available to the controller is
\begin{equation}
    Y_{k} = D_{k - \tau}\,X_{k - \tau},
    \label{eqn:Y_k}
\end{equation}
where \(D_{k} \in \{0,1\}\) denotes the scheduling decision at time \(k\). 

Let \(t_n\) denote the \(n\)th scheduling instant, i.e., the \(n\)th time index at which \(D_k=1\). The inter-scheduling interval is then defined as \footnote{Classical AoI is an instantaneous, time-varying quantity. In this work, we evaluate AoI only at sampling instants, immediately before new updates are received. Accordingly, policies with the same mean inter-scheduling interval have the same mean sampled AoI.}
\[
\Delta_n := t_{n+1}-t_n .
\]
Controller's information at time \(k\) is given by
\begin{equation}
    I^\mathcal{C}_{k}
    \;=\;
    \{\,Y_{0},\,U_{0},\,Y_{1},\,U_{1},\,\dots,\,Y_{k-1},U_{\,k-1},\,Y_{\,k}\,\}.
    \label{eqn:Control_info}
\end{equation}
In this paper, we restrict attention to state-independent scheduling policies, as is commonly assumed in wireless network design. Under this restriction, the scheduler does not observe the current plant state and bases its decisions only on the past scheduling and control history. Accordingly, the scheduler information set is given by
\begin{align}
    I^\mathcal{S}_{k}
    \;&=\;
    \{\,D_{0},\,D_{1},\,\dots,\,D_{\,k-1}\,\}.
    \label{eqn:update_info}
\end{align}
Under state-independent scheduling policies, scheduling decisions do not convey additional state-dependent information beyond the transmitted updates. Consequently, the estimation-error dynamics are independent of the control input, and the separation principle applies. See~\cite{etcibasi2026separation} for details.
We define the control and scheduling policies as
\begin{align}
  U_{k} \; &= \; \gamma_{k}\bigl(I^\mathcal{C}_{k}\bigr) \;\in\; \Gamma,\\
  D_{k} \; &= \; f_{k}\bigl(I^\mathcal{S}_{k}\bigr) \;\in\; \Pi,
\end{align}
where the sets \(\Gamma\) and \(\Pi\) contain all admissible control and scheduling decision rules that satisfy the usual measurability and integrability conditions.  
We write  
\(\boldsymbol{\mu}=(\gamma_{0},\gamma_{1},\dots)\in\Gamma\) and  
\(\boldsymbol{\pi}=(f_{0},f_{1},\dots)\in\Pi\)  
for generic control- and scheduling-policy sequences, respectively.  

The main objective is to minimize the constrained expected LQR cost over an infinite horizon:
\begin{align}
    P_1 :\quad 
    & \min_{\boldsymbol{\pi}\in\Pi,\;\boldsymbol{\mu}\in\Gamma}
    \; \lim_{N \to \infty} \frac{1}{N} 
    \mathbb{E}\left[\sum_{k=0}^{N-1} qX_k^2 + rU_k^2\right] 
    \label{eqn:Original_Objective}\\
    \text{s.t.}\quad 
    & \lim_{N \to \infty} \frac{1}{N} \sum_{k=0}^{N-1} D_k \leq p_{max}. 
    \notag
\end{align}
where \(q,\,r,\,p_{max} > 0\).

\begin{theorem} 
\label{theorem:optimal control}
    Consider the system \eqref{eqn:system_Model}--\eqref{eqn:update_info} and the optimization problem \(P_{1}\). Under i.i.d. zero-mean disturbances with finite variance, the optimal controller is given by the linear feedback law 
    \begin{equation}
        U_k = \gamma_k(I^\mathcal{C}_{k}) = -L\,\hat{X}_k,
        \label{eqn:optimal_controller}
    \end{equation}
    where 
    \begin{align*}
        & \hat{X}_k = \mathbb{E}\bigl[X_k\mid I^\mathcal{C}_{k}\bigr], 
        \quad
        L = \frac{a b\,K}{r + b^2 K}, \\
        & K = q + a^2K - \frac{(a b K)^2}{r + b^2 K}.
    \end{align*}
\end{theorem}

\begin{proof}
    The proof follows by dynamic programming under the controller information \(I_k^\mathcal{C}\). See~\cite{etcibasi2026separation} for the full derivation. 
\end{proof}

Let us define the estimation error at the controller as
\begin{equation}
    \mathcal{E}_{k} \;=\; X_{k} \;-\; \hat{X}_{k}.
    \label{eqn:est_error_original}
\end{equation}
Given the optimal controller in~\eqref{eqn:optimal_controller}, the objective function in~\eqref{eqn:Original_Objective} can be reformulated in terms of the estimation error~\(\mathcal{E}_{k}\). The goal is then to determine the optimal scheduling policy that minimizes this reformulated cost.

\begin{theorem} 
\label{theorem:equivalent problem}
  For a scalar LTI system with \(q,r>0\), i.i.d. zero-mean disturbances with finite variance, a fixed delay \(\tau \geq 1\), and using the optimal controller in \eqref{eqn:optimal_controller}, problem \(P_{1}\) can be equivalently written as
  \begin{align}
    P_{2}:\quad 
      &\min_{\boldsymbol{\pi}\; \in \; \Pi}\;  
      \lim_{N \to \infty}
      \frac{1}{N}\,
      \mathbb{E}\left[\sum_{k=0}^{N-1} \mathcal{E}_{k}^{2}\right] 
      \label{eqn:P_2}\\
    \text{s.t.}\quad 
      & \lim_{N \to \infty}
      \frac{1}{N}\sum_{k=0}^{N-1} D_{k} \;\le\; p_{max}. 
      \nonumber
  \end{align}
\end{theorem}

\begin{proof}
    The result follows from the Riccati telescoping argument, see
    \cite{etcibasi2026separation} for details. In particular, it was shown
    that
    \begin{align}
        \frac{1}{N} \mathbb{E} \Bigg[
        q X_N^2 + & \sum_{k=0}^{N-1} \Big( q X_k^2 + r U_k^2 \Big)
        \Bigg]
        =
        \frac{K_0}{N} \mathbb{E} [X_0^2] \notag \\
        & +
        \frac{1}{N} \mathbb{E}\left[
        \sum_{k=0}^{N-1}
        L^2 \left( r + b^2 K \right)
        \left( X_k - \mathbb{E}[X_k \mid I^\mathcal{C}_{k}] \right)^2
        \right]
        + K \sigma_W^2 .
        \label{eqn:eqv_prob_last_eqn}
    \end{align}
\end{proof}

Under the state-independent information structure considered here, the control law is fixed by the conditional state estimate, while the scheduler affects performance only through the temporal pattern of received updates and the resulting estimation-error process. This observation allows us to study scheduling through the distribution of the inter-scheduling intervals. In the next section, we use this formulation to examine whether classical freshness metrics, in particular mean AoI and peak AoI, are sufficient to characterize the optimal scheduling policy for the LQR tracking problem.

\section{Control-Aware Freshness under i.i.d. Disturbances}
\label{sec:AoI_iid_Noise}

In this section, we derive the freshness objective induced directly by the LQR tracking criterion. Under state-independent scheduling, the resulting objective depends on the inter-scheduling interval distribution through higher-order terms, not only through its mean. This characterization allows us to identify the optimal rate-limited scheduler and construct an example showing that a lower mean AoI can still lead to a higher tracking cost.

Before giving the main theorem for this section, we first characterize the estimation-error dynamics between two consecutive received updates. For \(1 \le m \le \Delta_n\), we have
\begin{align*}
	X_{t_n+\tau+m} & = a^{\tau+m}X_{t_n}
	+\sum_{j=0}^{\tau+m-1} a^{\tau+m-1-j}bU_{t_n+j} +\sum_{j=0}^{\tau+m-1} a^{\tau+m-1-j}W_{t_n+j}.
\end{align*}
Under state-independent scheduling policies and i.i.d. zero-mean disturbances,
\begin{align*}
	\mathbb{E}[W_{t_n+j}\mid I^{\mathcal C}_{t_n+\tau+m}]
	=
	0,
	\qquad 0 \le j \le \tau+m,
\end{align*}
which implies
\begin{align}
	\mathcal{E}_{t_n+\tau+m}
	=
	\sum_{j=0}^{\tau+m-1}
	a^{\tau+m-1-j}W_{t_n+j}. \label{eqn:est_error_dynamics}
\end{align}
The following theorem is the main result of this section. It shows that the LQR-induced freshness objective depends not only on mean AoI, but also on higher-order statistics of the inter-scheduling interval distribution.

\begin{theorem} \label{theorem:equivalent_problem_state_indep_iid}
    Consider the system~\eqref{eqn:system_Model}--\eqref{eqn:update_info} with $q,r>0$, a fixed delay $\tau \ge 1$, and i.i.d. zero-mean disturbances $W_k$ satisfying $\mathbb{E}[W_k^2]=\sigma_W^2<\infty$. Under state-independent scheduling policies that induce a stationary and ergodic inter-scheduling interval process $\{\Delta_n\}_{n\ge0}$ with finite-valued support, the tracking problem $P_1$ can be equivalently written as follows.
    
    For $|a| \neq 1$,
    \begin{align}
        P_3: \;& \min_{\Delta_n} \quad 
        \frac{\mathbb{E}\!\left[a^{2\Delta_n}\right]-1}
        {\mathbb{E}\!\left[\Delta_n\right]}\\
        & \; \text{s.t.}\quad 
        \frac{1}{\mathbb{E}\!\left[\Delta_n\right]} \le p_{max}. 
        \nonumber
    \end{align}
    
    For $|a|=1$,
    \begin{align}
        P_4: \;& \min_{\Delta_n} \quad 
        \frac{\mathbb{E}\!\left[\Delta_n^2\right]}
        {\mathbb{E}\!\left[\Delta_n\right]}\\
        & \; \text{s.t.}\quad 
        \frac{1}{\mathbb{E}\!\left[\Delta_n\right]} \le p_{max}. 
        \nonumber
    \end{align}
\end{theorem}
\vspace{.5em}
\begin{proof}
Let \(M_N\) denote the total number of scheduling instants up to time \(N\), i.e.,
\begin{align}
    M_N \triangleq \sum_{k=0}^{N-1} D_k .
\end{align}
Assume that the long-term communication rate exists and is given by
\begin{align}
    p_c \triangleq \lim_{N\to\infty}\frac{M_N}{N} \le p_{max}.
\end{align}
Let $t_0$ denote the first scheduling instant. Then, the total cost can be decomposed into an initial transient term, the contributions of the completed scheduling intervals, and a final residual term as
\begin{align}
    \sum_{k=0}^{N-1}\mathcal{E}_k^2
    =
    \sum_{k=0}^{t_0+\tau}\mathcal{E}_k^2
    +
    \sum_{n=0}^{M_N-1}
    \sum_{k=t_n+\tau+1}^{t_{n+1}+\tau}
    \mathcal{E}_k^2
    +
    \sum_{k=t_{M_N}+\tau+1}^{N-1}\mathcal{E}_k^2 .
\end{align}
Dividing both sides by $N$ and taking the limit as $N\to\infty$ yields
\begin{align}
    & \lim_{N\to\infty}
    \frac{1}{N}
    \sum_{k=0}^{N-1}\mathcal{E}_k^2
     =
    \lim_{N\to\infty}
    \frac{1}{N}
    \Bigg(
    \sum_{k=0}^{t_0+\tau}\mathcal{E}_k^2  +
    \sum_{n=0}^{M_N-1}
    \sum_{k=t_n+\tau+1}^{t_{n+1}+\tau}
    \mathcal{E}_k^2
    +
    \sum_{k=t_{M_N}+\tau+1}^{N-1}\mathcal{E}_k^2
    \Bigg).
\end{align}

Under state-independent scheduling policies that induce a stationary and ergodic inter-scheduling interval process $\{\Delta_n\}_{n\ge0}$ with finite-valued support, and under finite-valued disturbances with finite variance, the cycle cost
\begin{align}
    H_n
    \triangleq
    \sum_{k=t_n+\tau+1}^{t_{n+1}+\tau}
    \mathcal{E}_k^2
\end{align}
has finite expectation. Consequently, the initial transient and residual terms become asymptotically negligible, yielding
\begin{align}
    \lim_{N\to\infty}
    \frac{1}{N}
    \sum_{k=0}^{N-1}\mathcal{E}_k^2
    =
    \lim_{N\to\infty}
    \frac{1}{N}
    \sum_{n=0}^{M_N-1} H_n .
\end{align}
Using
\begin{align}
    \lim_{N\to\infty}\frac{M_N}{N}=p_c ,
\end{align}
we obtain
\begin{align}
    \lim_{N\to\infty}
    \frac{1}{N}
    \sum_{k=0}^{N-1}\mathcal{E}_k^2
    =
    p_c
    \lim_{N\to\infty}
    \frac{1}{M_N}
    \sum_{n=0}^{M_N-1} H_n .
\end{align}
Moreover, under the same assumptions, the induced process $\{H_n\}_{n\ge0}$ is stationary and ergodic. Therefore, by the ergodic theorem,
\begin{align}
    \lim_{N\to\infty}
    \frac{1}{M_N}
    \sum_{n=0}^{M_N-1} H_n
    =
    \mathbb{E}[H_n] .
\end{align}
Hence,
\begin{align}
    \lim_{N\to\infty}
    \frac{1}{N}
    \sum_{k=0}^{N-1}\mathcal{E}_k^2
    =
    p_c
    \mathbb{E}
    \left[
    \sum_{k=t_n+\tau+1}^{t_{n+1}+\tau}
    \mathcal{E}_k^2
    \right].
\end{align}

Next, observe that the completed scheduling intervals cover the horizon up to an asymptotically negligible boundary term. Therefore,
\begin{align}
    \sum_{n=0}^{M_N-1}\Delta_n
    =
    N + o(N),
    \qquad \text{as } N\to\infty .
\end{align}
Dividing both sides by $M_N$ gives
\begin{align}
    \frac{1}{M_N}
    \sum_{n=0}^{M_N-1}\Delta_n
    =
    \frac{N}{M_N}
    +
    o\!\left(\frac{N}{M_N}\right).
\end{align}
Taking the limit and using
\begin{align}
    \lim_{N\to\infty}\frac{M_N}{N}=p_c ,
\end{align}
yields
\begin{align}
    \lim_{N\to\infty}
    \frac{1}{M_N}
    \sum_{n=0}^{M_N-1}\Delta_n
    =
    \frac{1}{p_c}.
\end{align}
Since $\{\Delta_n\}_{n\ge0}$ is stationary and ergodic, another application of the ergodic theorem gives
\begin{align}
    \mathbb{E}[\Delta_n]
    =
    \lim_{N\to\infty}
    \frac{1}{M_N}
    \sum_{n=0}^{M_N-1}\Delta_n
    =
    \frac{1}{p_c}.
\end{align}
Therefore,
\begin{align}
    p_c
    =
    \frac{1}{\mathbb{E}[\Delta_n]} .
\end{align}
Substituting this relation into the objective yields the equivalent optimization problem
\begin{align}
    \min_{\Delta_n}
    \quad
    \frac{1}{\mathbb{E}[\Delta_n]}
    \,
    \mathbb{E}
    \left[
    \sum_{m=1}^{\Delta_n}
    \mathcal{E}_{t_n+\tau+m}^2
    \right]
\end{align}
subject to
\begin{align}
    \frac{1}{\mathbb{E}[\Delta_n]}
    \le p_{max} .
\end{align}
Hence, using \eqref{eqn:est_error_dynamics}, for $|a|\neq 1$ we have
\begin{align}
    &\mathbb{E}
    \left[
    \sum_{m=1}^{\Delta_n}
    \mathcal{E}_{t_n+\tau+m}^2
    \right]
    \nonumber\\
    &\overset{(a)}{=}
    \sigma_w^2
    \sum_{s=1}^{\infty}
    \Pr(\Delta_n=s)
    \sum_{m=1}^{s}
    \frac{1-a^{2(\tau+m)}}{1-a^2}
    \nonumber\\
    &=
    \sigma_w^2
    \left(
    \frac{\mathbb{E}[\Delta_n]}{1-a^2}
    +
    \frac{a^{2\tau+2}}{(1-a^2)^2}
    \left(
    \mathbb{E}[a^{2\Delta_n}]-1
    \right)
    \right).
\end{align}
For $|a|=1$, we obtain
\begin{align}
    &\mathbb{E}
    \left[
    \sum_{m=1}^{\Delta_n}
    \mathcal{E}_{t_n+\tau+m}^2
    \right]
    \nonumber\\
    & \overset{(a)}{=}
    \sum_{s=1}^{\infty}
    \Pr(\Delta_n=s)
    \sum_{m=1}^{s}
    (\tau+m)\sigma_w^2
    \nonumber\\
    &=
    \sum_{s=1}^{\infty}
    \Pr(\Delta_n=s)
    \left(
    \tau s
    +
    \frac{s}{2}
    +
    \frac{s^2}{2}
    \right)
    \sigma_w^2
    \nonumber\\
    &=
    \left(
    \left(\tau+\frac{1}{2}\right)\mathbb{E}[\Delta_n]
    +
    \frac{1}{2}\mathbb{E}[\Delta_n^2]
    \right)
    \sigma_w^2 .
\end{align}
where step $(a)$ for both $|a| \neq 1$ and $|a|=1$ follows because the scheduling policy is state independent and the disturbance process is i.i.d. zero-mean. Therefore, conditioning on the scheduling interval $\Delta_n=s$ does not change the second-order statistics of the disturbance sequence, and we have
\begin{align}
    \mathbb{E}\!\left[\mathcal{E}_{t_n+\tau+m}^2 \mid \Delta_n=s\right]
    =
    \sum_{j=0}^{\tau+m-1} a^{2(\tau+m-1-j)}\sigma_w^2 .
\end{align}
Hence, the equivalent optimization problem becomes the following.

For $|a|\neq 1$,
\begin{align}
    \min_{\Delta_n} & 
    \bigg\{
    \frac{a^{2\tau+2}}{(1-a^2)^2}
    \frac{\mathbb{E}[a^{2\Delta_n}]}
    {\mathbb{E}[\Delta_n]}
    -
    \frac{a^{2\tau+2}}{(1-a^2)^2}
    \frac{1}{\mathbb{E}[\Delta_n]}  +
    \frac{1}{1-a^2}
    \bigg\}
    \sigma_w^2
\end{align}
subject to
\begin{align}
    \frac{1}{\mathbb{E}[\Delta_n]}
    \le p_{max} .
\end{align}

For $|a|=1$,
\begin{align}
    \min_{\Delta_n}
    \quad
    \left(
    \frac{1}{2}
    \frac{\mathbb{E}[\Delta_n^2]}
    {\mathbb{E}[\Delta_n]}
    +
    \tau+\frac{1}{2}
    \right)
    \sigma_w^2 \label{eqn:eqv_cost_iid_case}
\end{align}
subject to
\begin{align}
    \frac{1}{\mathbb{E}[\Delta_n]}
    \le p_{max} .
\end{align}

Ignoring additive and positive multiplicative constants yields the equivalent formulations stated in the theorem.
\end{proof}

Theorem~\ref{theorem:equivalent_problem_state_indep_iid} provides the key link between freshness and control performance. It shows that, under state-independent scheduling, the infinite-horizon LQR tracking problem does not reduce to minimizing mean AoI. Instead, the control-induced freshness objective depends on the full inter-scheduling interval distribution: for $|a|\neq 1$, through the exponential moment $\mathbb{E}[a^{2\Delta_n}]$, and for $|a|=1$, through the second moment $\mathbb{E}[\Delta_n^2]$. Consequently, the theorem converts the qualitative statement that ``freshness matters'' into an explicit optimization criterion determined by the plant dynamics.

We now investigate the optimal AoI distribution for the rate-limited, always-accessible channel in problems $P_6$ and $P_7$.

\begin{definition}[As-periodic-as-possible scheduler]
	\label{def:as_periodic_as_possible}
	For a target rate \(p\in(0,1]\), let
	\(
	d \triangleq \left\lfloor 1/p \right\rfloor,
	\;
	\theta \triangleq 1/p-d .
	\)
	A scheduler is called as periodic as possible with rate \(p\) if
	\[
	\mathbb{P}(\Delta_n=d)=1-\theta,
	\;
	\mathbb{P}(\Delta_n=d+1)=\theta .
	\]
\end{definition}
If \(1/p\) is an integer, this reduces to the periodic scheduler \(\Delta_n=1/p\).

\begin{corollary}
	\label{corollary:optimal_as_periodic_as_possible}
	For the rate-limited, always-accessible channel, the optimal solution of \(P_6\) and \(P_7\) is the as-periodic-as-possible scheduler with rate \(p_{\max}\). Consequently, the rate constraint is active.
\end{corollary}

\begin{proof}
    Let
    \[
        \mu \triangleq \mathbb{E}[\Delta_n],
        \qquad
        \mu_0 \triangleq \frac{1}{p_{\max}} .
    \]
    Since the communication-rate constraint is
    \[
        \frac{1}{\mathbb{E}[\Delta_n]}\le p_{\max},
    \]
    feasibility is equivalent to \(\mu\ge \mu_0\). We first fix an arbitrary
    feasible mean \(\mu\) and optimize over all positive integer-valued
    distributions with this mean.
    Write
    \[
        m=\lfloor \mu \rfloor,
        \qquad
        \theta=\mu-m,
        \qquad
        \theta\in[0,1).
    \]
    For \(P_6\), where \(|a|\neq1\), let \(c=a^2\). Then \(c>0\), \(c\neq1\),
    and the relevant numerator is \(\mathbb{E}[c^{\Delta_n}]\). The forward
    second difference of \(f(d)=c^d\) is
    \[
        f(d+2)-2f(d+1)+f(d)
        =
        c^d(c-1)^2
        >
        0.
    \]
    Thus \(c^d\) is discretely convex. Equivalently, by Jensen's inequality,
    if the prescribed mean \(\mu\) is an integer, then
    \[
        \mathbb{E}[c^{\Delta_n}]
        \ge
        c^\mu,
    \]
    with equality if \(\Delta_n=\mu\) almost surely, i.e., under the periodic
    scheduler. If \(\mu\) is not an integer, the integer-valued analogue is
    obtained by placing all probability mass on the two adjacent integers
    around \(\mu\). Hence, among all integer-valued \(\Delta_n\) with mean
    \(\mu\), the minimum of \(\mathbb{E}[c^{\Delta_n}]\) is attained by
    \[
        \mathbb{P}(\Delta_n=m)=1-\theta,
        \qquad
        \mathbb{P}(\Delta_n=m+1)=\theta .
    \]
    Equivalently,
    \[
        \mathbb{E}[c^{\Delta_n}]
        \ge
        c^m(1-\theta)+c^{m+1}\theta,
    \]
    with equality under the distribution above.

    The same fixed-mean conclusion holds for \(P_7\). Indeed, when
    \(|a|=1\), the numerator is \(\mathbb{E}[\Delta_n^2]\), and \(d^2\) is
    discretely convex. Hence, for fixed \(\mu\), the minimum of
    \(\mathbb{E}[\Delta_n^2]\) is also attained by the same adjacent
    two-point distribution. Therefore, for either \(P_6\) or \(P_7\), once the
    mean is fixed, the optimal distribution is as periodic as possible.

    It remains to optimize the mean. For \(P_6\), under the adjacent
    two-point distribution and for \(\mu\in[m,m+1)\),
    \[
        V(\mu)
        \triangleq
        \mathbb{E}[c^{\Delta_n}]
        =
        c^m(m+1-\mu)+c^{m+1}(\mu-m)
        =
        c^m\bigl[1+(\mu-m)(c-1)\bigr].
    \]
    Thus the best achievable objective value at mean \(\mu\) is
    \[
        G_m(\mu)
        \triangleq
        \frac{V(\mu)-1}{\mu}.
    \]
    Since \(m\) is fixed on the interval \([m,m+1)\),
    \[
        G_m'(\mu)
        =
        \frac{\mu V'(\mu)-\bigl(V(\mu)-1\bigr)}{\mu^2}
        =
        \frac{1+c^m\bigl[m(c-1)-1\bigr]}{\mu^2}.
    \]
    Define
    \[
        H_m(c)
        \triangleq
        1+c^m\bigl[m(c-1)-1\bigr].
    \]
    Then \(H_m(1)=0\), and
    \[
        H_m'(c)
        =
        m(m+1)c^{m-1}(c-1).
    \]
    Hence \(H_m'(c)<0\) for \(0<c<1\) and \(H_m'(c)>0\) for \(c>1\). Since
    \(H_m(1)=0\), it follows that
    \[
        H_m(c)>0,
        \qquad c>0,\quad c\neq1.
    \]
    Therefore \(G_m'(\mu)>0\). The best achievable value for \(P_6\) is
    increasing in the mean \(\mu\).

    For \(P_7\), again using the adjacent two-point distribution,
    \[
        V(\mu)
        \triangleq
        \mathbb{E}[\Delta_n^2]
        =
        m^2(m+1-\mu)+(m+1)^2(\mu-m)
        =
        \mu(2m+1)-m(m+1).
    \]
    The best achievable objective value is
    \[
        G_m(\mu)
        \triangleq
        \frac{V(\mu)}{\mu},
    \]
    and therefore
    \[
        G_m'(\mu)
        =
        \frac{\mu(2m+1)-V(\mu)}{\mu^2}
        =
        \frac{m(m+1)}{\mu^2}
        \ge 0.
    \]
    Thus the best achievable value for \(P_7\) is also increasing in
    \(\mu\).

    Since the optimal fixed-mean objective is increasing in \(\mu\) for both
    \(P_6\) and \(P_7\), the smallest feasible mean is optimal:
    \[
        \mu^\star=\mu_0=\frac{1}{p_{\max}}.
    \]
    Hence
    \[
        \mathbb{E}[\Delta_n^\star]=\frac{1}{p_{\max}},
        \qquad
        \frac{1}{\mathbb{E}[\Delta_n^\star]}=p_{\max},
    \]
    so the communication-rate constraint is active.
    Substituting \(\mu^\star\) into the fixed-mean optimizer, with
    \[
        d=\lfloor \mu^\star \rfloor,
        \qquad
        \theta=\mu^\star-d,
    \]
    gives
    \[
        \mathbb{P}(\Delta_n^\star=d)=1-\theta,
        \qquad
        \mathbb{P}(\Delta_n^\star=d+1)=\theta .
    \]
    This is precisely the as-periodic-as-possible scheduler with rate
    \(p_{\max}\). If \(\mu^\star\) is an integer, then \(\theta=0\), and the
    policy reduces to the deterministic periodic scheduler.
\end{proof}

As shown above, the inequality constraint is active under the optimal policy.
However, this monotonic behavior does not hold for every inter-scheduling
interval distribution. In particular, if an arbitrary inter-scheduling interval
distribution is shifted to the right, thereby decreasing the communication
rate, the tracking cost is not guaranteed to increase. Equivalently, there
exist inter-scheduling interval distributions for which decreasing the
communication rate can improve performance. This occurs because the mean
inter-scheduling interval appears in the denominator of the objective functions
in $P_6$ and $P_7$. We illustrate this point as follows.

Consider the objective in $P_7$,
\[
    J(\Delta)=\frac{\mathbb{E}[\Delta^2]}{\mathbb{E}[\Delta]}.
\]
Let $\widetilde{\Delta}=\Delta+s$, where $s>0$. Then
\[
    J(\widetilde{\Delta})
    =
    \frac{\mathbb{E}[(\Delta+s)^2]}{\mathbb{E}[\Delta+s]}
    =
    \frac{\mathbb{E}[\Delta^2]+2s\mathbb{E}[\Delta]+s^2}
    {\mathbb{E}[\Delta]+s}.
\]
Writing
\[
    \mu=\mathbb{E}[\Delta],
    \qquad
    m_2=\mathbb{E}[\Delta^2],
\]
we obtain
\[
    J(\widetilde{\Delta})-J(\Delta)
    =
    \frac{s\left(2\mu^2+s\mu-m_2\right)}
    {\mu(\mu+s)}.
\]
Therefore,
\[
    J(\widetilde{\Delta})<J(\Delta)
    \quad \Longleftrightarrow \quad
    m_2>2\mu^2+s\mu.
\]
Equivalently,
\[
    J(\widetilde{\Delta})>J(\Delta)
    \quad \Longleftrightarrow \quad
    m_2<\mu(s+2\mu).
\]
Thus, for the \(P_7\) objective, shifting the distribution to the right
increases the cost only when the second moment is not too large. If the
original inter-scheduling interval distribution has a sufficiently large second
moment, namely \(m_2>\mu(s+2\mu)\), the same shift decreases the objective
value even though the communication rate is reduced.

Now consider the objective in $P_6$,
\[
    J(\Delta)
    =
    \frac{\mathbb{E}[a^{2\Delta}]-1}{\mathbb{E}[\Delta]}.
\]
Let
\[
    c\triangleq a^2,
    \qquad c>0,\quad c\neq 1,
\]
so that
\[
    J(\Delta)=\frac{\mathbb{E}[c^\Delta]-1}{\mathbb{E}[\Delta]}.
\]
Let $\widetilde{\Delta}=\Delta+s$, where $s>0$. Then
\[
    J(\widetilde{\Delta})
    =
    \frac{\mathbb{E}[c^{\Delta+s}]-1}
    {\mathbb{E}[\Delta+s]}
    =
    \frac{c^s\mathbb{E}[c^\Delta]-1}
    {\mathbb{E}[\Delta]+s}.
\]
Writing
\[
    \mu=\mathbb{E}[\Delta],
    \qquad
    M_c=\mathbb{E}[c^\Delta],
\]
we obtain
\[
    J(\widetilde{\Delta})-J(\Delta)
    =
    \frac{\mu(c^sM_c-1)-(\mu+s)(M_c-1)}
    {\mu(\mu+s)}.
\]
Equivalently,
\[
    J(\widetilde{\Delta})-J(\Delta)
    =
    \frac{M_c\bigl(\mu(c^s-1)-s\bigr)+s}
    {\mu(\mu+s)}.
\]
Therefore,
\[
    M_c\bigl(\mu(c^s-1)-s\bigr)+s<0.
\]
Equivalently, writing
\[
    \eta_s \triangleq s+\mu-\mu c^s=s-\mu(c^s-1),
\]
we have
\[
    J(\widetilde{\Delta})>J(\Delta)
    \quad \Longleftrightarrow \quad
    s>\eta_s M_c .
\]
This gives two cases. If
\[
    \eta_s\le 0
    \quad \Longleftrightarrow \quad
    s+\mu\le \mu c^s,
\]
then \(s>\eta_s M_c\) holds automatically because \(s>0\) and \(M_c>0\). In this
case, shifting the distribution to the right necessarily increases the
objective. If
\[
    \eta_s>0
    \quad \Longleftrightarrow \quad
    s+\mu>\mu c^s,
\]
then the sign depends on the exponential moment:
\[
    J(\widetilde{\Delta})>J(\Delta)
    \quad \Longleftrightarrow \quad
    M_c<\frac{s}{\eta_s},
\]
while
\[
    J(\widetilde{\Delta})<J(\Delta)
    \quad \Longleftrightarrow \quad
    M_c>\frac{s}{\eta_s}.
\]
Thus, when \(\eta_s>0\) and \(\mathbb{E}[c^\Delta]\) is sufficiently large,
shifting the inter-scheduling interval distribution to the right can reduce the
objective value, even though the communication rate is reduced.

As we observed, shifting a distribution to the right lowers the communication rate, but it can still reduce the tracking cost because the objective depends on the entire interval distribution, not only on its mean. The following example illustrates this phenomenon.

\paragraph*{\textbf{Interesting observation:}}
Consider a two-point randomized
inter-scheduling interval distribution
\[
    \Delta =
    \begin{cases}
        \delta_s, & \text{with probability } 1-p_\ell,\\
        \delta_\ell, & \text{with probability } p_\ell,
    \end{cases}
\]
where \(\delta_s,\delta_\ell\in\mathbb{Z}_{>0}\),
\(\delta_s<\delta_\ell\), and \(p_\ell\in[0,1]\). We also define the shifted
distribution
\[
    \widetilde{\Delta}=\Delta+s,
    \qquad s\in\mathbb{Z}_{>0},
\]
so that
\[
    \mathbb{P}(\widetilde{\Delta}=\delta_s+s)=1-p_\ell,
    \qquad
    \mathbb{P}(\widetilde{\Delta}=\delta_\ell+s)=p_\ell .
\]

\begin{figure}[t]
    \centering
    \begin{subfigure}{0.48\linewidth}
        \centering
        \includegraphics[width=0.7\linewidth]{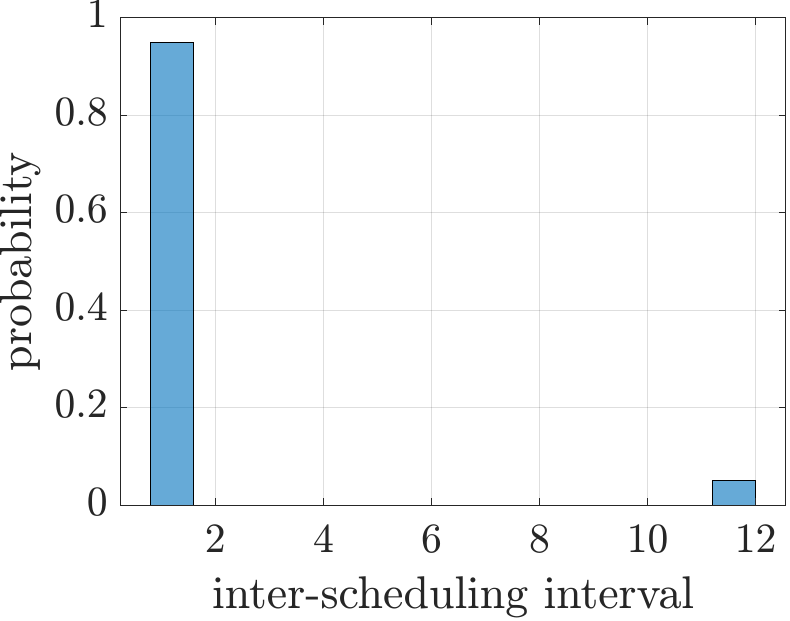}
        \caption{Original two-point distribution.}
        \label{fig:iut_distribution_two_point_iid_original}
    \end{subfigure}
    \hfill
    \begin{subfigure}{0.48\linewidth}
        \centering
        \includegraphics[width=0.7\linewidth]{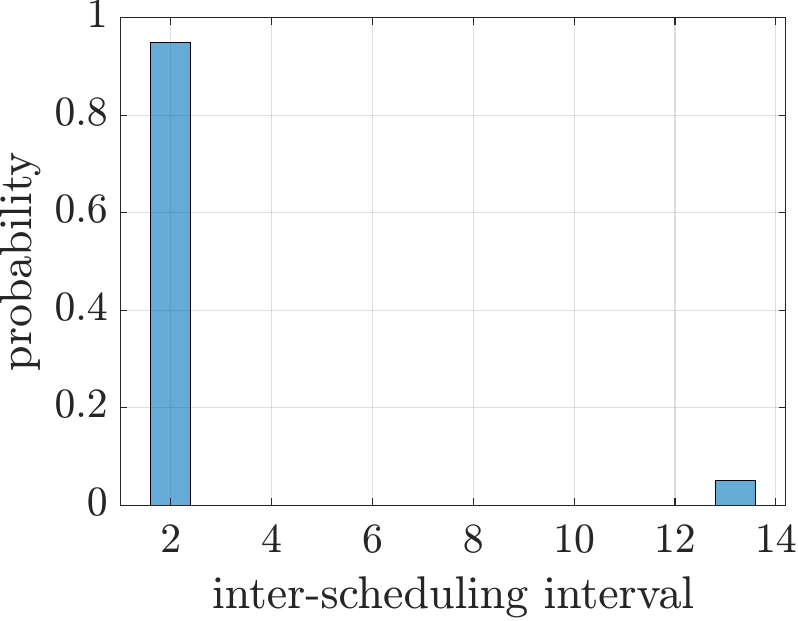}
        \caption{Shifted two-point distribution.}
        \label{fig:iut_distribution_two_point_iid_shifted}
    \end{subfigure}
    \caption{Original and shifted two-point inter-scheduling interval
    distributions for \(\delta_s=1\), \(\delta_\ell=12\),
    \(p_\ell=0.05\), and \(s=1\). The shifted scheduler increases each
    inter-scheduling interval by one time step, thereby reducing the
    communication rate.}
    \label{fig:iut_distribution_two_point_iid}
\end{figure}
Now consider the specific choice illustrated in
Fig.~\ref{fig:iut_distribution_two_point_iid},
\[
    \delta_s=1,
    \qquad
    \delta_\ell=12,
    \qquad
    p_\ell=0.05,
    \qquad
    s=1.
\]
Then
\[
    \mathbb{P}(\Delta=1)=0.95,
    \qquad
    \mathbb{P}(\Delta=12)=0.05,
\]
whereas
\[
    \mathbb{P}(\widetilde{\Delta}=2)=0.95,
    \qquad
    \mathbb{P}(\widetilde{\Delta}=13)=0.05.
\]
The corresponding means are
\[
    \mathbb{E}[\Delta]
    =
    0.95(1)+0.05(12)
    =
    1.55,
\]
and
\[
    \mathbb{E}[\widetilde{\Delta}]
    =
    0.95(2)+0.05(13)
    =
    2.55.
\]
Therefore,
\[
    \frac{1}{\mathbb{E}[\widetilde{\Delta}]}
    <
    \frac{1}{\mathbb{E}[\Delta]},
\]
so the shifted distribution has a smaller communication rate.

For the objective in \(P_7\),
\[
    \mathbb{E}[\Delta^2]
    =
    0.95(1)^2+0.05(12)^2
    =
    8.15,
\]
and hence
\[
    J(\Delta)
    =
    \frac{\mathbb{E}[\Delta^2]}{\mathbb{E}[\Delta]}
    =
    \frac{8.15}{1.55}
    \approx 5.26.
\]
For the shifted distribution,
\[
    \mathbb{E}[\widetilde{\Delta}^2]
    =
    0.95(2)^2+0.05(13)^2
    =
    12.25,
\]
so
\[
    J(\widetilde{\Delta})
    =
    \frac{\mathbb{E}[\widetilde{\Delta}^2]}
    {\mathbb{E}[\widetilde{\Delta}]}
    =
    \frac{12.25}{2.55}
    \approx 4.80.
\]
Thus,
\[
    J(\widetilde{\Delta})<J(\Delta).
\]
This example shows that shifting the inter-scheduling interval distribution to
the right decreases the communication rate, but can also decrease the
equivalent objective.

This example also admits a queueing-theoretic interpretation. Suppose the channel is usually fast, but occasionally incurs a large service time: the effective update time is \(1\) with probability \(0.95\) and \(12\) with probability \(0.05\). One might expect that the best strategy is to schedule updates as frequently as possible. The calculation above shows that this intuition can fail: in this example, scheduling once every two time steps yields better tracking performance than scheduling at every time step. This is consistent with the ``update-or-wait'' phenomenon in~\cite{sun2017update}, where zero-wait policies are shown to be suboptimal for certain AoI penalty functions and service-time distributions. The key distinction is that~\cite{sun2017update} optimizes freshness directly, whereas here the freshness objective emerges from the LQR tracking cost.

Before closing the i.i.d. case analysis, we clarify the connection with the
classical time-average AoI metric. As noted in the system model, classical AoI
is an instantaneous, time-varying process, whereas our analysis uses the AoI
sampled at the controller observation times. Under the same stationary
inter-scheduling interval model and constant delay \(\tau\), the classical
time-average AoI is
\[
    \frac{1}{2}
    \frac{\mathbb{E}[\Delta_n^2]}{\mathbb{E}[\Delta_n]}
    +
    \tau+\frac{1}{2}.
\]
Thus, up to an affine scaling and shift, the classical mean-AoI objective
coincides with \(P_4\), the special case of our LQR-induced objective when
\(|a|=1\). This equivalence is fragile: it holds only at the marginal open-loop
gain \(|a|=1\). If \(a\) is perturbed even slightly away from this value, the
equivalent objective becomes \(P_3\), which depends on the exponential moment
\(\mathbb{E}[a^{2\Delta_n}]\) rather than only on the second moment. Therefore,
even near \(a=1\), minimizing classical mean AoI need not minimize the LQR
tracking cost. In Fig.~\ref{fig:iid_two_point_a_near_one}, we illustrate this with an example where the LQR cost increases although the mean AoI decreases.

\section{Control-Aware Freshness under Correlated Disturbances}
\label{sec:AoI_Correlated_Noise}
The preceding section isolates the effect of the inter-scheduling interval distribution when the disturbance sequence is i.i.d. Physical systems, however, generally experience temporally correlated disturbances. Vehicle velocities, actuator loads, wind disturbances, and human-driven traffic trajectories all exhibit temporal persistence. To account for this structure, we now extend the analysis to correlated disturbances and show that the same distributional message remains: mean AoI is not sufficient, and the relevant cost depends on exponential moment terms determined jointly by the plant dynamics and the correlation decay.

Assume that the disturbance process $\{W_k\}_{k\ge0}$ is wide-sense stationary and zero mean, with exponentially decaying autocorrelation given by
\begin{equation}
    \mathbb{E}[W_iW_j] = R_W(i-j) = e^{-\beta |i-j|}\sigma_W^2 , \label{eqn:noise_ACF}
\end{equation}
where $\beta \ge 0$ controls the correlation strength. This exponential form describes a disturbance process with exponentially fading memory and a single dominant correlation timescale. Disturbances at nearby time instants are strongly correlated, while their dependence decays exponentially with temporal separation, corresponding to short-range dependence. Such exponential velocity correlations are not merely theoretical assumptions but have been observed experimentally in physical systems, including the Lagrangian motion of tracer particles in fully developed turbulent flows~\cite{mordant2001measurement}. Exponentially correlated acceleration models are also widely used to represent persistent vehicle and target maneuvers~\cite{singer1970estimating}. For the NGSIM US-101 velocity data considered in Section~\ref{sec:Real_Sim}, the parameter $\beta$ is selected by minimizing the mean-squared error between the empirical average ACF and~\eqref{eqn:noise_ACF}. The resulting MSE of $3.5\times10^{-3}$ demonstrates excellent agreement between the exponential model and the measured temporal correlation.

In the correlated disturbances setting, the optimal controller characterization in Theorem~\ref{theorem:optimal control} and the equivalent formulation in Theorem~\ref{theorem:equivalent problem} do not generally hold, because the independence arguments used in the i.i.d. case fail. We therefore return to the original infinite-horizon LQR problem \(P_1\). Assuming state-independent scheduling, a maintained separation structure, and a fixed controller, we focus on the scheduler design. To isolate the effect of scheduling, we consider the small control penalty regime \(r\to0\).

Note that, even under a fixed controller assumption, the control cost term $rU_k^2$ cannot generally be removed directly from the objective. The reason is that the controller input is typically a function of the most recent received state update $X_{t_n}$. Since the scheduling decisions determine the distribution of the scheduling instants and the disturbance process is temporally correlated, the future states $X_{t_n+m}$ become statistically dependent on the previously received state $X_{t_n}$. Consequently, the control inputs $U_k$ also become dependent on the scheduling decisions through the induced state evolution. Therefore, removing the control term from the objective is, in general, not theoretically justified, even when the controller itself is fixed. Nevertheless, in the regime $r\to0$, the dominant contribution to the objective is the state tracking term. Hence, the problem we consider can be written as
\begin{align}
    P_{7} :\quad 
    & \min_{D_k} \; 
    \lim_{N \to \infty} \frac{q}{N} 
    \mathbb{E}\left[\sum_{k=0}^{N-1} X_k^2\right] \\
    \text{s.t.}\quad 
    & \lim_{N \to \infty} \frac{1}{N} 
    \sum_{k=0}^{N-1} D_k 
    \leq p_{max}. 
    \notag
\end{align}

Now, similarly to the i.i.d. disturbance case, under state-independent scheduling policies that induce stationary and ergodic inter-scheduling intervals $\{\Delta_n\}_{n\ge0}$ with finite-valued support, and under finite-valued disturbances, the infinite-horizon average cost can be rewritten in terms of the typical scheduling period. Let $t_0$ denote the first scheduling instant and let $M_N$ denote the number of completed scheduling intervals within the horizon $[0,N)$. Then,
\begin{align}
    \sum_{k=0}^{N-1} X_k^2
    & =
    \sum_{k=0}^{t_0+\tau} X_k^2
    +
    \sum_{n=0}^{M_N-1}
    \sum_{k=t_n+\tau+1}^{t_{n+1}+\tau} X_k^2  +
    \sum_{k=t_{M_N}+\tau+1}^{N-1} X_k^2 .
\end{align}
Dividing both sides by $N$ and taking the limit as $N\to\infty$ yields
\begin{align}
    &\lim_{N\to\infty}
    \frac{1}{N}
    \sum_{k=0}^{N-1} X_k^2
    \nonumber\\
    &=
    \lim_{N\to\infty}
    \frac{1}{N}
    \Bigg(
    \sum_{k=0}^{t_0+\tau} X_k^2
    +
    \sum_{n=0}^{M_N-1}
    \sum_{k=t_n+\tau+1}^{t_{n+1}+\tau} X_k^2
    +
    \sum_{k=t_{M_N}+\tau+1}^{N-1} X_k^2
    \Bigg) \notag \\
    &\overset{(a)}{=}
    \lim_{N\to\infty}
    \frac{1}{N}
    \sum_{n=0}^{M_N-1}
    \sum_{k=t_n+\tau+1}^{t_{n+1}+\tau} X_k^2
    \nonumber\\
    &=
    \lim_{N\to\infty}
    \frac{M_N}{N}
    \frac{1}{M_N}
    \sum_{n=0}^{M_N-1}
    \sum_{k=t_n+\tau+1}^{t_{n+1}+\tau} X_k^2
    \nonumber\\
    &\overset{(b)}{=}
    p_c
    \mathbb{E}
    \left[
    \sum_{k=t_n+\tau+1}^{t_{n+1}+\tau} X_k^2
    \right].
\end{align}
In step $(a)$, the initial transient and residual boundary terms vanish asymptotically. This requires
\begin{align}
    \mathbb{E}[R_n] < \infty ,
\end{align}
where
\begin{align}
    R_n
    \triangleq
    \sum_{k=t_n+\tau+1}^{t_{n+1}+\tau} X_k^2 .
\end{align}
Under stationary and ergodic inter-scheduling intervals $\{\Delta_n\}_{n\ge0}$ with finite-valued support, and finite-valued disturbances, the process $\{R_n\}_{n\ge0}$ has finite expectation, and the boundary terms become asymptotically negligible.

Step $(b)$ follows from
\begin{align}
    \frac{M_N}{N}
    \to p_c,
    \qquad \text{as } N\to\infty ,
\end{align}
where $p_c$ denotes the long-term communication rate.

Finally, in step $(c)$, we apply the ergodic theorem, which requires the processes
\begin{align}
    \{R_n\}_{n\ge0}
    \quad \text{and} \quad
    \{\Delta_n\}_{n\ge0}
\end{align}
to be stationary and ergodic.

Moreover, note that the communication-rate constraint can be written as
\begin{align}
    \sum_{k=0}^{N-1} D_k
    =
    M_N .
\end{align}
Hence,
\begin{align}
    p_c
    =
    \frac{1}{\mathbb{E}[\Delta_n]} .
\end{align}
Therefore, the problem can be equivalently rewritten as
\begin{align}
    \min_{\Delta_n}
    \quad
    \frac{q}{\mathbb{E}[\Delta_n]}
    \mathbb{E}
    \left[
    \sum_{m=1}^{\Delta_n}
    X_{t_n+\tau+m}^2
    \right] \label{eqn:X_cost_typical_period}
\end{align}
subject to
\begin{align}
    \frac{1}{\mathbb{E}[\Delta_n]}
    \le p_{max} .
\end{align}

Now, let us consider the impulsive controller defined below. 
\begin{align}
    U_{t_n+\tau} &= -\frac{a}{b}\left(a^\tau X_{t_n}+\sum_{j=0}^{\tau-1} a^{\tau-1-j} bU_{t_n+j}\right), \label{eqn:impulsive_controller}\\
    U_{t_n+\tau+m} &= 0, \qquad 0 \le m \le t_{k+1}-t_n .
\end{align}
Using the linear system dynamics, the system state at time $t_n+\tau$ can be written as
\begin{equation}
    X_{t_n+\tau}
    =
    a^\tau X_{t_n}
    +
    \sum_{j=0}^{\tau-1} a^{\tau-1-j} bU_{t_n+j}
    +
    \sum_{j=0}^{\tau-1} a^{\tau-1-j} W_{t_n+j}.
\end{equation}
Substituting the impulsive controller~\eqref{eqn:impulsive_controller} into the system dynamics yields
\begin{align}
    X_{t_n+\tau+1}
    &= a\left(a^\tau X_{t_n}+\sum_{j=0}^{\tau-1} a^{\tau-1-j} bU_{t_n+j}+\sum_{j=0}^{\tau-1} a^{\tau-1-j} W_{t_n+j}\right) \nonumber\\
    &\quad - b\left(\frac{a}{b}\left(a^\tau X_{t_n}+\sum_{j=0}^{\tau-1} a^{\tau-1-j} bU_{t_n+j}\right)\right)+W_{t_n+\tau} \nonumber\\
    &= \sum_{j=0}^{\tau} a^{\tau-j} W_{t_n+j}.
\end{align}

Next, we work on the objective in \eqref{eqn:X_cost_typical_period} for different cases of the open loop gain $a$ and the correlation coefficient $\beta$. First, let's investigate the case $|a|=1$ and very correlated disturbances $\beta \to 0$. This case is quite important in the sense of the practical simulation we do in Section \ref{sec:Real_Sim} because in that simulation we use the real vehicle speed data collected on a highway and we show that the autocorrelation of these mean normalized speed data is actually coming from a exponentially correlated speed process with $\beta=0.007558$, which is nearly zero. Hence, the objective in \eqref{eqn:X_cost_typical_period} for $|a|=1$ becomes
\begin{align}
    &\mathbb{E}\!\left[\sum_{m=1}^{\Delta_n} X_{t_n+\tau+m}^2\right]
    \nonumber\\
    &=
    \sigma_W^2
    \sum_{s=1}^{\infty}\Pr(\Delta_n=s)
    \sum_{m=1}^{s}
    \left(
    \tau+m
    +
    2\sum_{d=1}^{\tau+m-1}(\tau+m-d)e^{-\beta d}
    \right)
    \nonumber\\
    &=
    \sigma_W^2
    \sum_{s=1}^{\infty}\Pr(\Delta_n=s)
    \sum_{m=1}^{s}
    \left(
    \left(1+\frac{2e^{-\beta}}{1-e^{-\beta}}\right)(\tau+m)
    -
    \frac{2e^{-\beta}}{(1-e^{-\beta})^2}
    +
    \frac{2e^{-\beta}}{(1-e^{-\beta})^2}e^{-\beta(\tau+m)}
    \right)
    \nonumber\\
    &=
    \Bigg[
    \left(
    \left(\tau+\frac{1}{2}\right)
    \left(1+\frac{2e^{-\beta}}{1-e^{-\beta}}\right)
    -
    \frac{2e^{-\beta}}{(1-e^{-\beta})^2}
    \right)
    \mathbb{E}[\Delta_n]
    \nonumber\\
    &\qquad
    +
    \frac{1}{2}
    \left(1+\frac{2e^{-\beta}}{1-e^{-\beta}}\right)
    \mathbb{E}[\Delta_n^2]
    +
    \frac{2e^{-\beta(\tau+2)}}{(1-e^{-\beta})^3}
    \left(
    1-\mathbb{E}[e^{-\beta\Delta_n}]
    \right)
    \Bigg]\sigma_W^2 .
    \label{eqn:corr_noise_a_1}
\end{align}
The first equality follows from the state-independent scheduling policy and the exponentially correlated covariance model. In particular, conditioning on \(\Delta_n=s\) does not alter the second-order disturbance statistics, so the variance of the partial sum is obtained by summing all covariance terms \(\mathbb{E}[W_{t_n+i}W_{t_n+j}]=e^{-\beta|i-j|}\sigma_W^2\).
Now, as $\beta\to0$, the objective becomes
\begin{align}
    \lim_{\beta\to0}\mathbb{E}\!\left[\sum_{m=1}^{\Delta_n}X_{t_n+\tau+m}^2\right]
    &=\sum_{s=1}^{\infty}\Pr(\Delta_n=s)\sum_{m=1}^{s}\sum_{j=0}^{\tau+m-1}\sum_{i=0}^{\tau+m-1}\sigma_w^2 \nonumber\\
    &=\sum_{s=1}^{\infty}\Pr(\Delta_n=s)\sum_{m=1}^{s}(\tau+m)^2\sigma_w^2 \nonumber\\
    &=\left(\frac{1}{3}\mathbb{E}[\Delta_n^3]+\left(\tau+\frac{1}{2}\right)\mathbb{E}[\Delta_n^2]+\left(\tau^2+\tau+\frac{1}{6}\right)\mathbb{E}[\Delta_n]\right)\sigma_w^2 .
\end{align}
Hence, for the case $|a|=1$ and $\beta\to0$, the equivalent problem becomes
\begin{align}
    P_{8}: \; & \min_{\Delta_n}\quad
    q\left(
    \frac{1}{3}\frac{\mathbb{E}[\Delta_n^3]}{\mathbb{E}[\Delta_n]}
    +
    \left(\tau+\frac{1}{2}\right)
    \frac{\mathbb{E}[\Delta_n^2]}{\mathbb{E}[\Delta_n]}
    +
    \tau^2+\tau+\frac{1}{6}
    \right)\sigma_w^2 \label{eqn:eqv_cost_corr_noise_a_1_beta_0}\\
    & \textit{s.t.}\quad \frac{1}{\mathbb{E}[\Delta_n]} \le p_{max} . \notag
\end{align}
Let
\[
    \rho \triangleq e^{-\beta}.
\]
Using the expanded form in~\eqref{eqn:corr_noise_a_1}, for $\beta>0$ the equivalent problem for $|a| = 1$ can be written as
\begin{align}
    P_{9}: \; & \min_{\Delta_n}\quad
    q\left(
    F_1\frac{\mathbb{E}[\Delta_n^2]}{\mathbb{E}[\Delta_n]}
    +
    F_2\frac{1-\mathbb{E}[e^{-\beta\Delta_n}]}{\mathbb{E}[\Delta_n]}
    +
    F_3
    \right)\sigma_w^2 \notag\\
    & \textit{s.t.}\quad \frac{1}{\mathbb{E}[\Delta_n]} \le p_{max} . \notag
\end{align}
where
\begin{align*}
F_1 &\triangleq \frac{1}{2}+\frac{\rho}{1-\rho},\\
F_2 &\triangleq \frac{2\rho^{\tau+2}}{(1-\rho)^3},\\
F_3 &\triangleq \tau+\frac{1}{2}+(2\tau+1)\frac{\rho}{1-\rho}-\frac{2\rho}{(1-\rho)^2}.
\end{align*}
Finally, consider the general case \( |a|\neq 1 \) and \( \beta>0 \). The objective in \eqref{eqn:X_cost_typical_period} becomes
\begin{align}
&\mathbb{E}\!\left[\sum_{m=1}^{\Delta_n}X_{t_n+\tau+m}^2\right]
\nonumber\\
&=
\sigma_W^2
\sum_{s=1}^{\infty}\Pr(\Delta_n=s)
\sum_{m=1}^{s}
\Bigg(
\sum_{i=0}^{\tau+m-1}a^{2(\tau+m-1-i)}
+
2\sum_{d=1}^{\tau+m-1}
e^{-\beta d}
\sum_{k=0}^{\tau+m-1-d}
a^{2(\tau+m-1)-2k-d}
\Bigg)
\nonumber\\
&=
\frac{\sigma_W^2}{1-a^2}
\Bigg[
\left(
1+\frac{2ae^{-\beta}}{1-ae^{-\beta}}
\right)\mathbb{E}[\Delta_n]
+
\Bigg(
-\frac{a^{2\tau+2}}{1-a^2}
-\frac{2(ae^{-\beta})^{\tau+1}}{(1-ae^{-\beta})^2}
-\frac{2a^{2\tau+1}e^{-\beta}}
{(1-a^{-1}e^{-\beta})(1-a^2)}
\nonumber\\
&\qquad\qquad
+
\frac{2a^{\tau+1}e^{-\beta(\tau+1)}}
{(1-a^{-1}e^{-\beta})(1-ae^{-\beta})}
\Bigg)
+
\left(
\frac{a^{2\tau+2}}{1-a^2}
+
\frac{2a^{2\tau+1}e^{-\beta}}
{(1-a^{-1}e^{-\beta})(1-a^2)}
\right)
\mathbb{E}[a^{2\Delta_n}]
\nonumber\\
&\quad
+
\left(
\frac{2(ae^{-\beta})^{\tau+1}}{(1-ae^{-\beta})^2}
-
\frac{2a^{\tau+1}e^{-\beta(\tau+1)}}
{(1-a^{-1}e^{-\beta})(1-ae^{-\beta})}
\right)
\mathbb{E}[(ae^{-\beta})^{\Delta_n}]
\Bigg].
\end{align}
The first equality follows from the state-independent scheduling policy and the exponentially correlated covariance model
\(\mathbb{E}[W_{t_n+i}W_{t_n+j}]=e^{-\beta|i-j|}\sigma_W^2\). The final expression is obtained by evaluating the resulting geometric sums.
Let
\[
    \xi \triangleq \frac{\rho}{a},
    \qquad
    \eta \triangleq a\rho .
\]
Hence, the equivalent problem for $|a| \neq 1$ becomes
\begin{align}
P_{10} :\;& \min_{\Delta_n}\quad
\frac{q}{1-a^2}
\left(
H_1
+
H_2\frac{\mathbb{E}[a^{2\Delta_n}]}{\mathbb{E}[\Delta_n]}
+
H_3\frac{\mathbb{E}[\eta^{\Delta_n}]}{\mathbb{E}[\Delta_n]}
+
H_4\frac{1}{\mathbb{E}[\Delta_n]}
\right)\sigma_w^2 \notag \\
& \; \textit{s.t.} \quad\frac{1}{\mathbb{E}[\Delta_n]} \le p_{max} , \notag
\end{align}
where
\begin{align*}
    H_1 &\triangleq
    \frac{1+\eta}{1-\eta},\\
    H_2 &\triangleq
    \frac{a^{2\tau}\bigl(a^2(1-\xi)+2\eta\bigr)}
    {(1-\xi)(1-a^2)},\\
    H_3 &\triangleq
    -\frac{2\rho^2\eta^\tau(1-a^2)}
    {(1-\xi)(1-\eta)^2},\\
    H_4 &\triangleq
    \frac{2\eta^{\tau+1}(\xi-\eta)}
    {(1-\xi)(1-\eta)^2}
    -
    \frac{a^{2\tau}(a^2+\eta)}
    {(1-\xi)(1-a^2)}.
\end{align*}

We now characterize the optimal inter-scheduling distribution for the
correlated-noise equivalent problems, following the same fixed-mean argument
used in Section~\ref{sec:AoI_iid_Noise}. The key step is the following
discrete Jensen-type lemma.

\begin{lemma}
\label{lem:discrete_jensen_apap}
Let \(f:\mathbb{Z}_{+}\to\mathbb{R}\) be discretely convex, i.e.,
\[
    f(d+1)-2f(d)+f(d-1)\ge 0,
    \qquad d\ge 2.
\]
Then, among all positive integer-valued random variables \(\Delta\) with
fixed mean \(\mathbb{E}[\Delta]=\mu\), the minimum of
\(\mathbb{E}[f(\Delta)]\) is achieved by the as-periodic-as-possible
distribution.
\end{lemma}

\begin{proof}
Let \(p_d=\mathbb{P}(\Delta=d)\). Suppose that a feasible distribution places
positive probability at two points \(i\) and \(j\), with \(j-i\ge 2\). Move an
amount \(\varepsilon>0\) of probability mass from \(i\) to \(i+1\), and the
same amount from \(j\) to \(j-1\). This operation preserves both the total
probability and the mean, since
\[
    \varepsilon(i+1-i)+\varepsilon(j-1-j)=0.
\]
The corresponding change in \(\mathbb{E}[f(\Delta)]\) is
\[
    \varepsilon\big(f(i+1)-f(i)\big)
    +
    \varepsilon\big(f(j-1)-f(j)\big).
\]
Since \(f\) is discretely convex, the forward differences
\[
    \nabla f(d)\triangleq f(d+1)-f(d)
\]
are nondecreasing in \(d\). Therefore, because \(i<j-1\),
\[
    f(i+1)-f(i)
    =
    \nabla f(i)
    \le
    \nabla f(j-1)
    =
    f(j)-f(j-1).
\]
Hence,
\[
    \big(f(i+1)-f(i)\big)
    +
    \big(f(j-1)-f(j)\big)
    \le 0.
\]
Thus, this smoothing operation cannot increase the objective. Repeating the
operation removes all gaps larger than one in the support. Therefore, a
minimizing distribution must be supported on two adjacent integers. Matching
the fixed mean \(\mu\) gives the support points
\(\ell=\lfloor\mu\rfloor\) and \(\ell+1\), with probabilities
\(1-\theta\) and \(\theta=\mu-\ell\), respectively.
\end{proof}

For \(P_{8}\), fixing
\[
    \mu=\mathbb{E}[\Delta_n]
\]
fixes the denominator in the objective. Hence, the relevant kernel is
\[
    f_{8}(d)
    =
    \frac{1}{3}d^3+\left(\tau+\frac{1}{2}\right)d^2 .
\]
Its discrete second difference is
\[
    f_{8}(d+1)-2f_{8}(d)+f_{8}(d-1)
    =
    2d+2\tau+1
    >
    0.
\]
Therefore, \(f_{8}\) is discretely convex, and
Lemma~\ref{lem:discrete_jensen_apap} implies that the optimal distribution
for any fixed mean is as-periodic-as-possible.

For \(P_{9}\), recall that \(\rho=e^{-\beta}\in(0,1)\). The coefficients
\(F_1,F_2,F_3\) are defined above. Since \(F_3\) is independent of
\(\Delta_n\), it does not affect the fixed-mean minimization.
For fixed \(\mu=\mathbb{E}[\Delta_n]\), the relevant kernel is
\[
    f_{9}(d)
    =
    F_1 d^2+F_2(1-\rho^d).
\]
Its discrete second difference is
\[
\begin{aligned}
    f_{9}(d+1)-2f_{9}(d)+f_{9}(d-1)
    &=
    2F_1-F_2\rho^{d-1}(1-\rho)^2 .
\end{aligned}
\]
Using
\[
    F_1
    =
    \frac{1}{2}+\frac{\rho}{1-\rho}
    =
    \frac{1+\rho}{2(1-\rho)}
\]
and
\[
    F_2
    =
    \frac{2\rho^{\tau+2}}{(1-\rho)^3},
\]
we obtain
\[
\begin{aligned}
    f_{9}(d+1)-2f_{9}(d)+f_{9}(d-1)
    &\ge
    \frac{1+\rho}{1-\rho}
    -
    \frac{2\rho^{\tau+2}}{1-\rho} \\
    &=
    \frac{1+\rho-2\rho^{\tau+2}}{1-\rho}
    \ge 0,
\end{aligned}
\]
because \(0<\rho<1\) and \(\tau\ge 0\). Thus, \(f_{9}\) is discretely
convex, and the optimal fixed-mean distribution is again
as-periodic-as-possible.

For \(P_{10}\), the fixed-mean reduction gives the kernel
\[
    f_{10}(d)
    =
    \frac{1}{1-a^2}
    \left(
    H_2a^{2d}+H_3\eta^d
    \right),
\]
since \(H_1\) and \(H_4/\mathbb{E}[\Delta_n]\) are fixed when
\(\mu=\mathbb{E}[\Delta_n]\) is fixed. Its discrete second difference is
\[
\begin{aligned}
    f_{10}(d+1)-2f_{10}(d)+f_{10}(d-1)
    &=
    \frac{1}{1-a^2}
    \left(
    H_2a^{2d-2}(a^2-1)^2
    +
    H_3\eta^{d-1}(\eta-1)^2
    \right)\\
    &=
    \frac{
    a^{2\tau+2d-2}\bigl(a^2(1-\xi)+2\eta\bigr)
    -
    2\rho^2\eta^{\tau+d-1}
    }{1-\xi}.
\end{aligned}
\]
Using \(\xi=\rho/a\) and \(\eta=a\rho\), and assuming \(a>0\), this can be
rewritten as
\[
\begin{aligned}
    f_{10}(d+1)-2f_{10}(d)+f_{10}(d-1)
    &=
    \frac{
    a^{\tau+d-1}\rho^{\tau+d+1}
    \left[
    \left(\frac{a}{\rho}\right)^{\tau+d}
    \left(\frac{a}{\rho}+1\right)
    -2
    \right]
    }{
    1-\frac{\rho}{a}
    } .
\end{aligned}
\]
Let \(r=a/\rho\) and \(n=\tau+d\). Since \(d\ge2\) and \(\tau\ge0\), we have
\(n\ge2\). The function \(r^n(r+1)-2\) is strictly increasing for \(r>0\) and
vanishes at \(r=1\). Therefore, its sign agrees with the sign of \(r-1\),
which is also the sign of \(1-\rho/a\). Hence,
\[
    f_{10}(d+1)-2f_{10}(d)+f_{10}(d-1)\ge0,
\]
for \(a>0\) and \(a\neq\rho\). Thus, under this parameter regime,
\(f_{10}\) is discretely convex and the optimal fixed-mean distribution for
\(P_{10}\) is also as-periodic-as-possible.

It remains to determine whether the communication-rate constraint is active.
For this purpose, we use the following scalar reduction.

\begin{lemma}
\label{lem:apap_mean_monotonicity}
Let \(g:\mathbb{Z}_{+}\to\mathbb{R}\), and let \(\Delta_\mu\) denote the
as-periodic-as-possible random variable with mean \(\mu\). For
\(\mu\in[n,n+1]\), where \(n=\lfloor\mu\rfloor\),
\[
    \mathbb{P}(\Delta_\mu=n)=n+1-\mu,
    \qquad
    \mathbb{P}(\Delta_\mu=n+1)=\mu-n .
\]
Define
\[
    V_g(\mu)
    \triangleq
    \mathbb{E}[g(\Delta_\mu)]
    =
    (n+1-\mu)g(n)+(\mu-n)g(n+1),
\]
and
\[
    \Phi_g(\mu)
    \triangleq
    \frac{V_g(\mu)}{\mu}.
\]
Then, on each interval \((n,n+1)\),
\[
    \Phi_g'(\mu)
    =
    \frac{
    n\bigl(g(n+1)-g(n)\bigr)-g(n)
    }{\mu^2}
    =
    \frac{
    ng(n+1)-(n+1)g(n)
    }{\mu^2}.
\]
Consequently, \(\Phi_g\) is nondecreasing on every feasible interval if and
only if
\[
    ng(n+1)\ge(n+1)g(n)
\]
on those intervals. In that case, the minimum over \(\mu\ge\mu_0\) is attained
at \(\mu=\mu_0\).
\end{lemma}

\begin{proof}
For \(\mu\in(n,n+1)\),
\[
    V_g'(\mu)=g(n+1)-g(n).
\]
Therefore,
\[
    \Phi_g'(\mu)
    =
    \frac{\mu V_g'(\mu)-V_g(\mu)}{\mu^2}.
\]
Substituting the affine expression for \(V_g(\mu)\) gives
\[
    \mu\bigl(g(n+1)-g(n)\bigr)
    -
    \bigl((n+1-\mu)g(n)+(\mu-n)g(n+1)\bigr)
    =
    ng(n+1)-(n+1)g(n),
\]
which proves the result.
\end{proof}

We now apply Lemma~\ref{lem:apap_mean_monotonicity} to the three
correlated-noise equivalent costs. Let
\[
    \mu_0=\frac{1}{p_{\max}}.
\]
The communication-rate constraint satisfies
\[
    \frac{1}{\mathbb{E}[\Delta_n]}\le p_{\max}
    \quad \Longleftrightarrow \quad
    \mu\ge \mu_0.
\]

For \(P_{8}\), the relevant kernel is
\[
    g_8(d)
    =
    \frac{1}{3}d^3+\left(\tau+\frac{1}{2}\right)d^2 .
\]
Thus,
\[
\begin{aligned}
    ng_8(n+1)-(n+1)g_8(n)
    &=
    n(n+1)
    \left(
    \frac{2n+1}{3}+\tau+\frac{1}{2}
    \right)\\
    &=
    \frac{n(n+1)(4n+6\tau+5)}{6}
    >
    0 .
\end{aligned}
\]
Hence, the reduced \(P_{8}\) cost is increasing in \(\mu\), and the
communication-rate constraint is active.

For \(P_{9}\), the relevant kernel is
\[
    g_9(d)
    =
    F_1d^2+F_2(1-\rho^d).
\]
Lemma~\ref{lem:apap_mean_monotonicity} gives
\[
\begin{aligned}
    ng_9(n+1)-(n+1)g_9(n)
    &=
    F_1n(n+1)
    +
    F_2\left(\rho^n(1+n(1-\rho))-1\right).
\end{aligned}
\]
The second term is nonpositive. However, using
\[
    1-\rho^n(1+n(1-\rho))
    \le
    \frac{n(n+1)}{2}(1-\rho)^2,
\]
we obtain
\[
\begin{aligned}
    ng_9(n+1)-(n+1)g_9(n)
    &\ge
    F_1n(n+1)
    -
    \frac{F_2n(n+1)}{2}(1-\rho)^2\\
    &=
    \frac{n(n+1)}{1-\rho}
    \left(
    \frac{1+\rho}{2}
    -
    \rho^{\tau+2}
    \right)
    \ge 0,
\end{aligned}
\]
because \(0<\rho<1\) and \(\tau\ge0\). Therefore, the reduced \(P_{9}\) cost
is also nondecreasing in \(\mu\), and the communication-rate constraint is
active.

For \(P_{10}\), the corresponding kernel is
\[
    g_{10}(d)
    =
    \frac{1}{1-a^2}
    \left(
    H_2a^{2d}+H_3\eta^d
    +
    H_4
    \right).
\]
In this case,
\[
\begin{aligned}
    ng_{10}(n+1)-(n+1)g_{10}(n)
    &=
    \frac{1}{1-a^2}
    \Bigl[
    H_2a^{2n}\bigl(n(a^2-1)-1\bigr)\\
    &\qquad\qquad
    +
    H_3\eta^n\bigl(n(\eta-1)-1\bigr)
    -
    H_4
    \Bigr].
\end{aligned}
\]
Thus, the communication-rate constraint is active for \(P_{10}\) whenever
\[
    H_2a^{2n}\bigl(n(a^2-1)-1\bigr)
    +
    H_3\eta^n\bigl(n(\eta-1)-1\bigr)
    -
    H_4
    \quad
    \text{has the same sign as }1-a^2
\]
for every integer \(n\) whose interval \([n,n+1]\) intersects the feasible
set \([\mu_0,\infty)\). Unlike \(P_{8}\) and \(P_{9}\), this monotonicity is
parameter dependent; hence, for the general \(P_{10}\) cost, activity of the
communication-rate constraint should be verified through the above condition.

Similar to the i.i.d. disturbance case, the fact that the optimal constrained
solution uses the largest allowed communication rate does not imply that every
increase in communication rate lowers the LQR cost. For arbitrary
inter-scheduling interval distributions, increasing the communication rate can
also change higher-order moments and exponential moments in unfavorable ways.
Consequently, one can find schedulers for which a higher communication rate
does not reduce, and may even increase, the equivalent LQR cost.

\begin{figure}[t]
    \centering
    \begin{subfigure}{0.48\linewidth}
        \centering
        \includegraphics[width=0.7\linewidth]{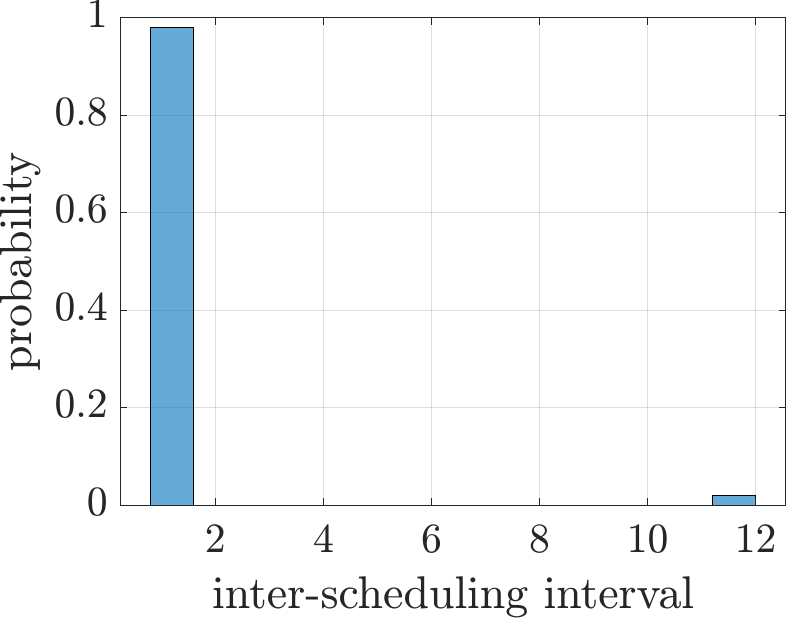}
        \caption{Original two-point distribution.}
        \label{fig:iut_distribution_two_point_correlated_original}
    \end{subfigure}
    \hfill
    \begin{subfigure}{0.48\linewidth}
        \centering
        \includegraphics[width=0.7\linewidth]{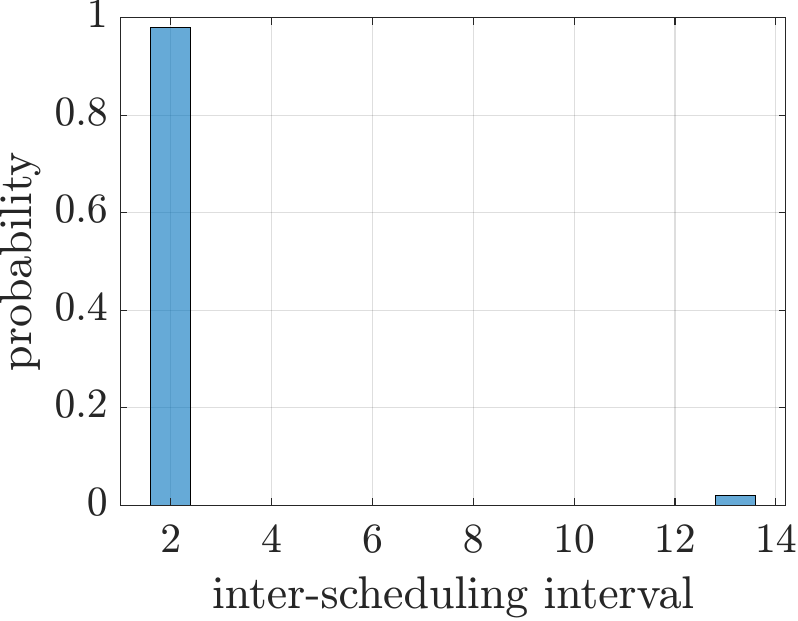}
        \caption{Shifted two-point distribution.}
        \label{fig:iut_distribution_two_point_correlated_shifted}
    \end{subfigure}
    \caption{Original and shifted two-point inter-scheduling interval
    distributions used in the correlated-noise simulations for
    \(\delta_s=1\), \(\delta_\ell=12\), \(p_\ell=0.02\), and \(s=1\).
    The shifted scheduler increases each inter-scheduling interval by one
    time step, thereby reducing the communication rate.}
    \label{fig:iut_distribution_two_point_correlated}
\end{figure}

To make this observation concrete, we revisit the two-point randomized
inter-scheduling interval distribution introduced in
Section~\ref{sec:AoI_iid_Noise}. Consider the specific choice illustrated in
Fig.~\ref{fig:iut_distribution_two_point_correlated},
\[
    \delta_s=1,
    \qquad
    \delta_\ell=12,
    \qquad
    p_\ell=0.02,
    \qquad
    s=1.
\]
Then
\[
    \mathbb{P}(\Delta=1)=0.98,
    \qquad
    \mathbb{P}(\Delta=12)=0.02,
\]
whereas
\[
    \mathbb{P}(\widetilde{\Delta}=2)=0.98,
    \qquad
    \mathbb{P}(\widetilde{\Delta}=13)=0.02.
\]
The corresponding means are
\[
    \mathbb{E}[\Delta]
    =
    0.98(1)+0.02(12)
    =
    1.22,
\]
and
\[
    \mathbb{E}[\widetilde{\Delta}]
    =
    0.98(2)+0.02(13)
    =
    2.22.
\]
Therefore,
\[
    \frac{1}{\mathbb{E}[\widetilde{\Delta}]}
    <
    \frac{1}{\mathbb{E}[\Delta]},
\]
so the shifted distribution has a smaller communication rate.

We now evaluate the correlated-noise equivalent objective in \(P_{8}\).
For \(\tau=1\), and ignoring the common positive scaling factor
\(q\sigma_w^2\), the relevant cost term is
\[
    J_{8}(\Delta)
    =
    \frac{1}{3}\frac{\mathbb{E}[\Delta^3]}{\mathbb{E}[\Delta]}
    +
    \frac{3}{2}\frac{\mathbb{E}[\Delta^2]}{\mathbb{E}[\Delta]}
    +
    \frac{13}{6}.
\]
For the original distribution,
\[
    \mathbb{E}[\Delta^2]
    =
    0.98(1)^2+0.02(12)^2
    =
    3.86,
\]
and
\[
    \mathbb{E}[\Delta^3]
    =
    0.98(1)^3+0.02(12)^3
    =
    35.54.
\]
Therefore,
\[
\begin{aligned}
    J_{8}(\Delta)
    &=
    \frac{1}{3}\frac{35.54}{1.22}
    +
    \frac{3}{2}\frac{3.86}{1.22}
    +
    \frac{13}{6} \\
    &\approx 16.62 .
\end{aligned}
\]
For the shifted distribution,
\[
    \mathbb{E}[\widetilde{\Delta}^2]
    =
    0.98(2)^2+0.02(13)^2
    =
    7.30,
\]
and
\[
    \mathbb{E}[\widetilde{\Delta}^3]
    =
    0.98(2)^3+0.02(13)^3
    =
    51.78.
\]
Thus,
\[
\begin{aligned}
    J_{8}(\widetilde{\Delta})
    &=
    \frac{1}{3}\frac{51.78}{2.22}
    +
    \frac{3}{2}\frac{7.30}{2.22}
    +
    \frac{13}{6} \\
    &\approx 14.87 .
\end{aligned}
\]
Hence,
\[
    J_{8}(\widetilde{\Delta})<J_{8}(\Delta).
\]
This example shows that, under the correlated-noise equivalent cost in
\(P_{8}\), shifting the inter-scheduling interval distribution to the right
decreases the communication rate and also decreases the equivalent objective.

For completeness, we also evaluate the finite-\(\beta\) cost in \(P_{9}\)
for the same two distributions. Again ignoring the common positive scaling
factor \(q\sigma_w^2\), define
\[
    J_{9}(\Delta)
    =
    F_1\frac{\mathbb{E}[\Delta^2]}{\mathbb{E}[\Delta]}
    +
    F_2
    \frac{1-\mathbb{E}[\rho^\Delta]}{\mathbb{E}[\Delta]}
    +
    F_3 .
\]
For the original and shifted distributions,
\[
    \mathbb{E}[\rho^\Delta]
    =
    0.98\rho+0.02\rho^{12},
    \qquad
    \mathbb{E}[\rho^{\widetilde{\Delta}}]
    =
    0.98\rho^2+0.02\rho^{13}.
\]
Therefore,
\[
\begin{aligned}
    J_{9}(\Delta)
    &=
    F_1\frac{3.86}{1.22}
    +
    F_2
    \frac{1-0.98\rho-0.02\rho^{12}}{1.22}
    +
    F_3,\\
    J_{9}(\widetilde{\Delta})
    &=
    F_1\frac{7.30}{2.22}
    +
    F_2
    \frac{1-0.98\rho^2-0.02\rho^{13}}{2.22}
    +
    F_3 .
\end{aligned}
\]
Hence, the shifted finite-\(\beta\) cost is smaller precisely when
\[
    F_1\left(\frac{7.30}{2.22}-\frac{3.86}{1.22}\right)
    +
    F_2\left(
    \frac{1-0.98\rho^2-0.02\rho^{13}}{2.22}
    -
    \frac{1-0.98\rho-0.02\rho^{12}}{1.22}
    \right)
    <0.
\]
For the fitted value \(\beta=0.007558\), this gives
\[
    J_{9}(\widetilde{\Delta})-J_{9}(\Delta)
    \approx
    -1.68,
\]
so the shifted distribution also has a smaller \(P_{9}\) equivalent cost in
the highly correlated regime used in the simulations.

Similar to the discussion in Section~\ref{sec:AoI_iid_Noise}, this example
also admits a useful channel interpretation. Suppose updates can be attempted
frequently, but the effective update process occasionally experiences a large
delay. Specifically, with probability \(p_\ell=0.02\), an update cycle lasts
\(\delta_\ell=12\) time steps instead of \(\delta_s=1\) time step. It may seem
natural to use the channel as frequently as possible. The shifted scheduler
tests the opposite behavior by waiting one additional time step between
scheduling attempts, thereby reducing the communication rate. The calculation
above shows that this lower-rate strategy can yield a smaller correlated-noise
equivalent cost. We return to this example in the simulation section and
compare the resulting empirical LQR costs under correlated US-101 velocity
disturbances.

\section{Data-Driven Validation with NGSIM US-101 Trajectories}
\label{sec:Real_Sim}
In this section, we evaluate the theoretical analysis for state-independent scheduling policies developed in Section~\ref{sec:AoI_Correlated_Noise} using practical car-following simulations based on the Next Generation Simulation (NGSIM) dataset. The NGSIM dataset is one of the most comprehensive and widely used datasets for studying vehicle dynamics and microscopic traffic behavior~\cite{thiemann2008estimating}. It was initiated by the U.S. Department of Transportation to support the development and validation of traffic simulation models by providing high-resolution, real-world vehicle trajectory data~\cite{ngsim_us101_2016}. This setting is particularly well suited to our objective because vehicle trajectories exhibit temporally correlated fluctuations, allowing us to test whether the LQR cost follows the higher-order moment and exponential-moment structure predicted by Section~\ref{sec:AoI_Correlated_Noise}, rather than being determined by mean AoI alone.

The dataset consists of multiple traffic scenarios collected from different locations, including US 101 (Los Angeles), I-80 (Emeryville), Lankershim Boulevard (Los Angeles), and Peachtree Street (Atlanta). Each dataset captures detailed interactions between vehicles under varying traffic conditions such as free flow, congestion buildup, and fully congested regimes. In what follows, we use these real trajectories to construct a data-driven car-following scenario and evaluate whether the empirical closed-loop performance is governed by the full inter-scheduling-time, or AoI, distribution, as predicted by the theory.

\begin{figure}[!ht]
    \centering
    \includegraphics[width=0.7\linewidth]{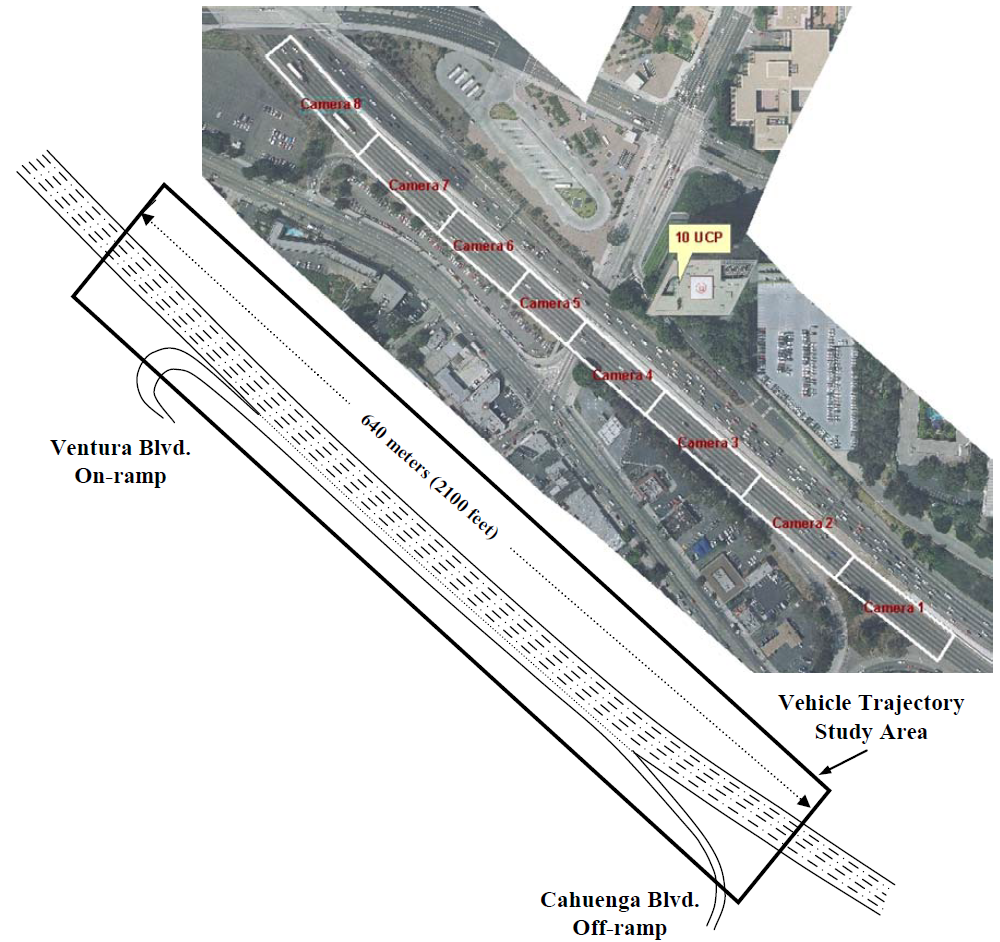}
    \caption{NGSIM US 101 study area showing the freeway segment, camera coverage, and ramp locations.}
    \label{fig:us101_map}
\end{figure}

NGSIM dataset includes all recorded vehicle measurements across all scenarios and time durations without subsampling, resulting in nearly 12 million data entries. This exhaustive coverage makes it particularly suitable for data-driven modeling and statistical analysis, as it preserves the full temporal and spatial structure of traffic dynamics.

Among the available NGSIM scenarios and locations, we focus on the highway setting, specifically the US 101 dataset, which provides a representative example of multi-lane freeway traffic dynamics. A schematic illustration of the study area is provided in Fig.~\ref{fig:us101_map}. In particular, the US 101 dataset contains vehicle trajectory data collected over a 640-meter section of a freeway using eight synchronized cameras mounted on a high-rise building. The trajectories were extracted using specialized computer vision tools, providing precise vehicle positions, velocities, and accelerations at a high temporal resolution of 0.1 seconds~\cite{ngsim_us101_2016}. The dataset spans 45 minutes, divided into three 15-minute intervals, capturing the transition from uncongested to congested traffic conditions~\cite{ngsim_us101_metadata}.

\begin{figure}[!ht]
    \centering
    \resizebox{0.8\linewidth}{!}{%
    \begin{minipage}{\linewidth}
        \centering
        
        \begin{subfigure}{0.48\linewidth}
            \centering
            \includegraphics[width=0.9\linewidth]{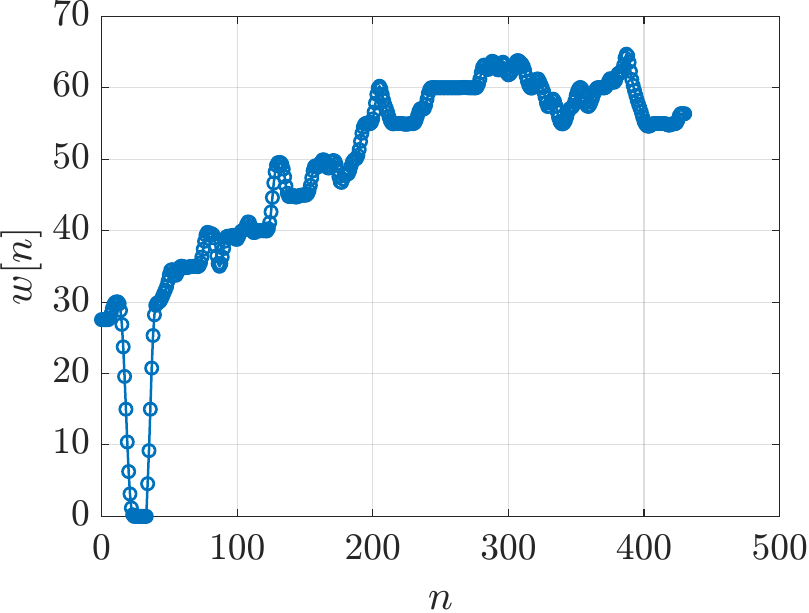}
            \caption{Vehicle 681.}
            \label{fig:track38}
        \end{subfigure}
        \hfill
        \begin{subfigure}{0.48\linewidth}
            \centering
            \includegraphics[width=0.9\linewidth]{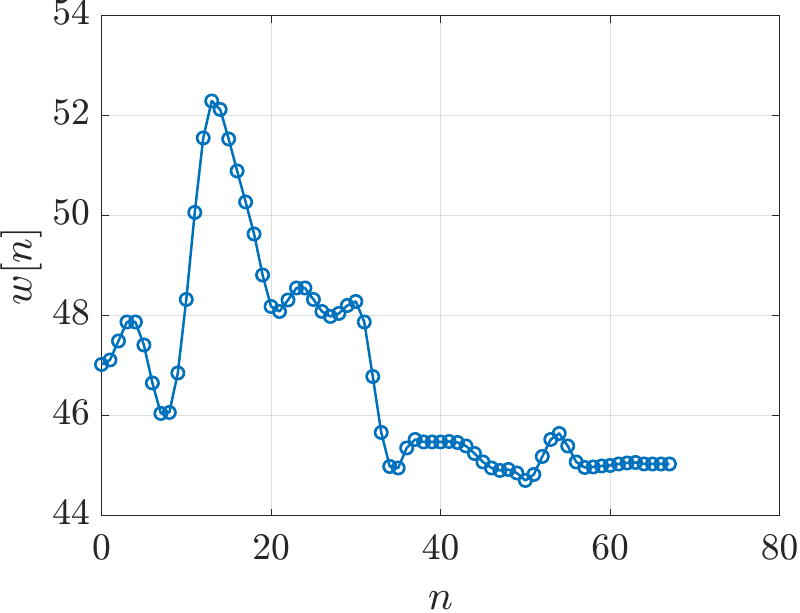}
            \caption{Vehicle 4031.}
            \label{fig:track102}
        \end{subfigure}

        \vspace{0.3cm}

        \begin{subfigure}{0.48\linewidth}
            \centering
            \includegraphics[width=0.9\linewidth]{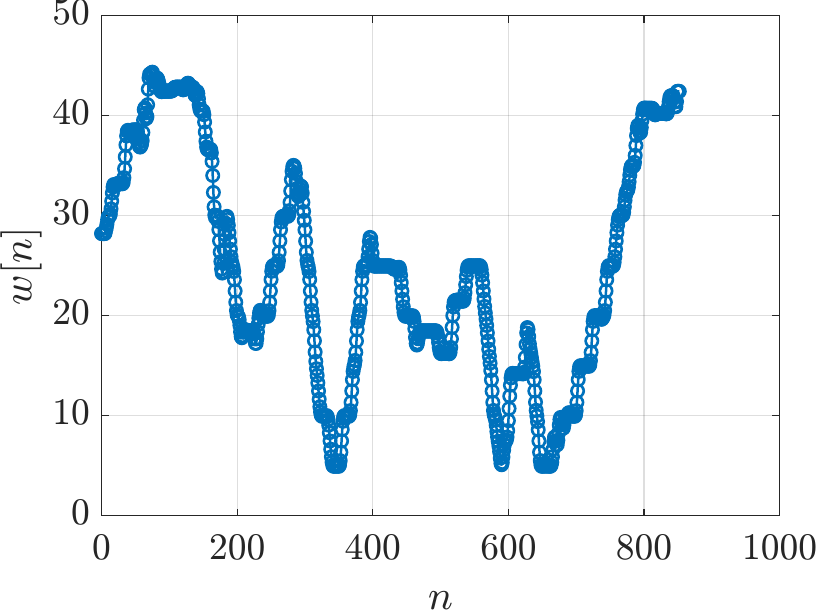}
            \caption{Vehicle 4731.}
            \label{fig:track68}
        \end{subfigure}
        \hfill
        \begin{subfigure}{0.48\linewidth}
            \centering
            \includegraphics[width=0.9\linewidth]{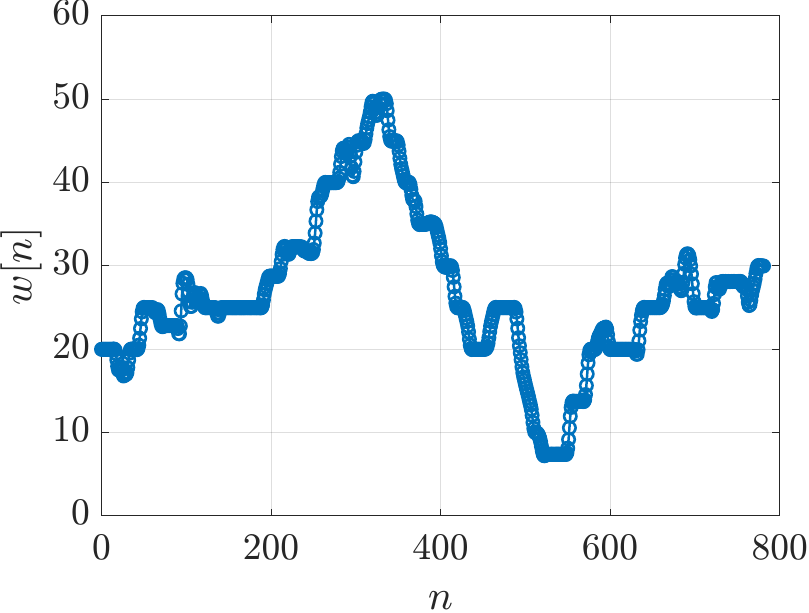}
            \caption{Vehicle 4473.}
            \label{fig:track43}
        \end{subfigure}
    \end{minipage}
    }
    \caption{Example velocity/disturbance sequences extracted from the NGSIM US-101 dataset.}
    \label{fig:noise_sequences}
\end{figure}

To obtain a setting that closely resembles a standard multi-lane highway scenario, we restrict our analysis to the first four mainline lanes. This choice avoids the auxiliary lanes associated with on-ramps and off-ramps present in the dataset, which introduce additional merging and diverging behaviors that are outside the scope of our study. Furthermore, we focus exclusively on passenger vehicles (autos), excluding motorcycles and trucks to ensure a more homogeneous traffic composition. As a result, our analysis is conducted on a subset consisting of approximately 3,900 vehicles.

In our setup, we leverage the high-resolution trajectory data in a data-driven manner by directly using individual vehicle trajectories to model the exogenous input of the system. Representative examples of the extracted velocity/disturbance sequences are shown in Fig.~\ref{fig:noise_sequences}. In particular, each trajectory from the NGSIM US-101 dataset is treated as a realization of the leader vehicle dynamics. This enables the incorporation of realistic traffic dynamics into the system through the disturbance process while maintaining full control over the follower dynamics.

Several observations regarding Fig.~\ref{fig:noise_sequences} are worth mentioning. First, although our analysis is restricted to 3900 vehicle trajectories, the original dataset contains nearly 6500 vehicles. Therefore, the vehicle IDs appearing in the figures may exceed 3900. Second, the durations of the velocity profiles vary significantly across vehicles. This is because the recorded road segment has a fixed length; consequently, vehicles traveling at higher speeds remain in the dataset for shorter durations, whereas slower traffic results in longer trajectories. To better illustrate this effect, the number of contributing trajectories at each time index $n$ is shown in Fig.~\ref{fig:num_contributing_tracks}. As can be observed, the number of available trajectories decreases significantly for $n > 1000$. As a result, some abrupt variations may appear in the running average LQR cost curves around this region due to the reduced number of contributing samples.

\begin{figure}[t]
    \centering
    \includegraphics[width=0.6\linewidth]{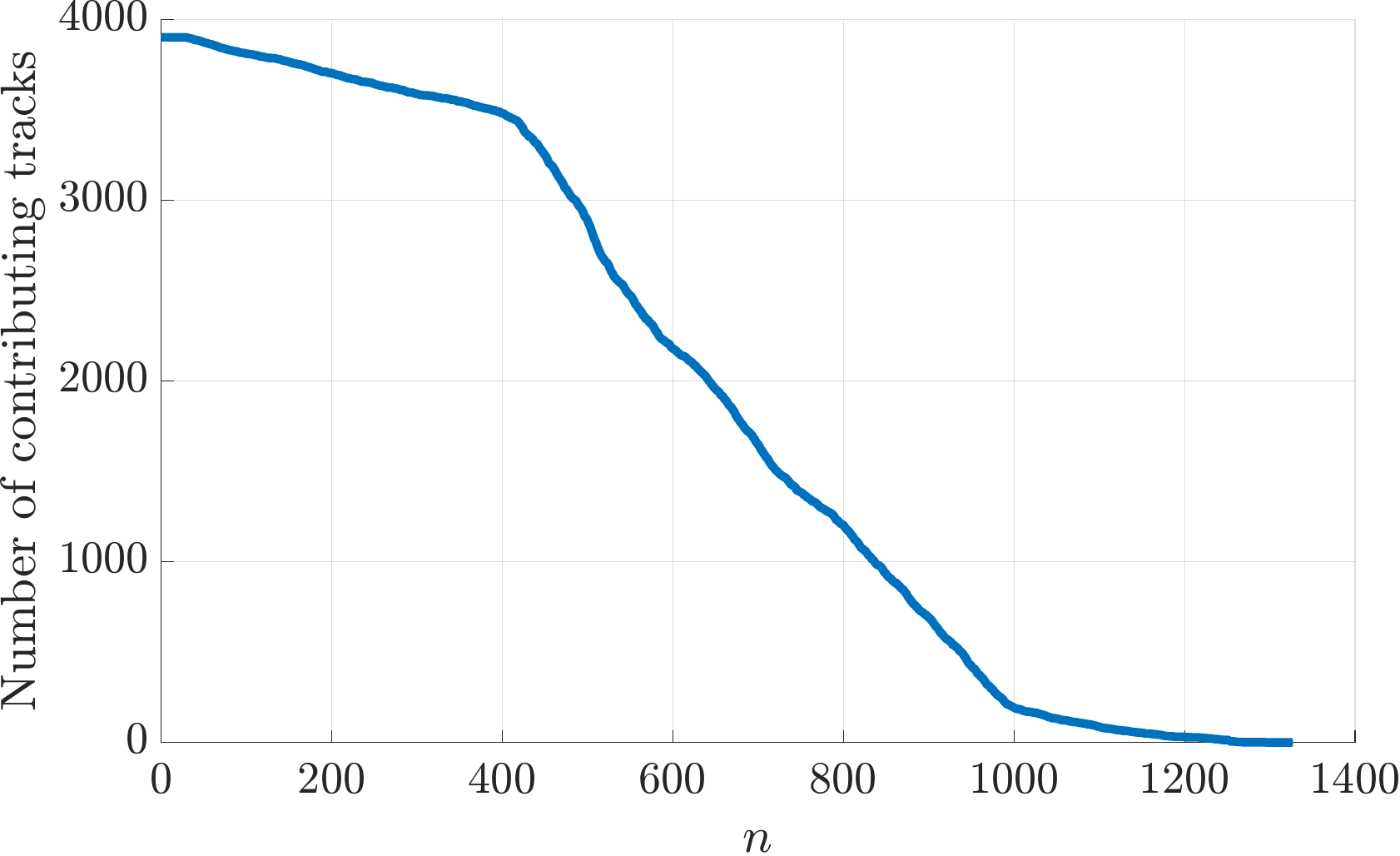}
    \caption{Number of contributing tracks for each averaged time index $n$.}
    \label{fig:num_contributing_tracks}
\end{figure}

The simulation setup is depicted in Fig.~\ref{fig:Car_Following}, which illustrates a discrete-time car-following scenario with a leader vehicle and a follower vehicle. The leader velocity profile is obtained from the NGSIM US-101 dataset. The objective is to maintain a desired inter-vehicle distance, while the follower does not have direct access to the relative distance measurement. Instead, the follower only receives information transmitted by the leader vehicle through a scheduling mechanism. Specifically, at each time step $n$, the leader decides whether to transmit its state according to the scheduling variable $d[n] \in \{0,1\}$. The communication channel is assumed to be rate-limited, error-free, and subject to a fixed delay. Based on the received updates, the follower determines its velocity, which serves as the control input of the system and is denoted by $u[n]$. Therefore, this setup directly corresponds to the system model introduced in Section~\ref{sec:Sys_Mdl}.

\begin{figure}[t]
    \centering
    \includegraphics[width=0.8\linewidth]{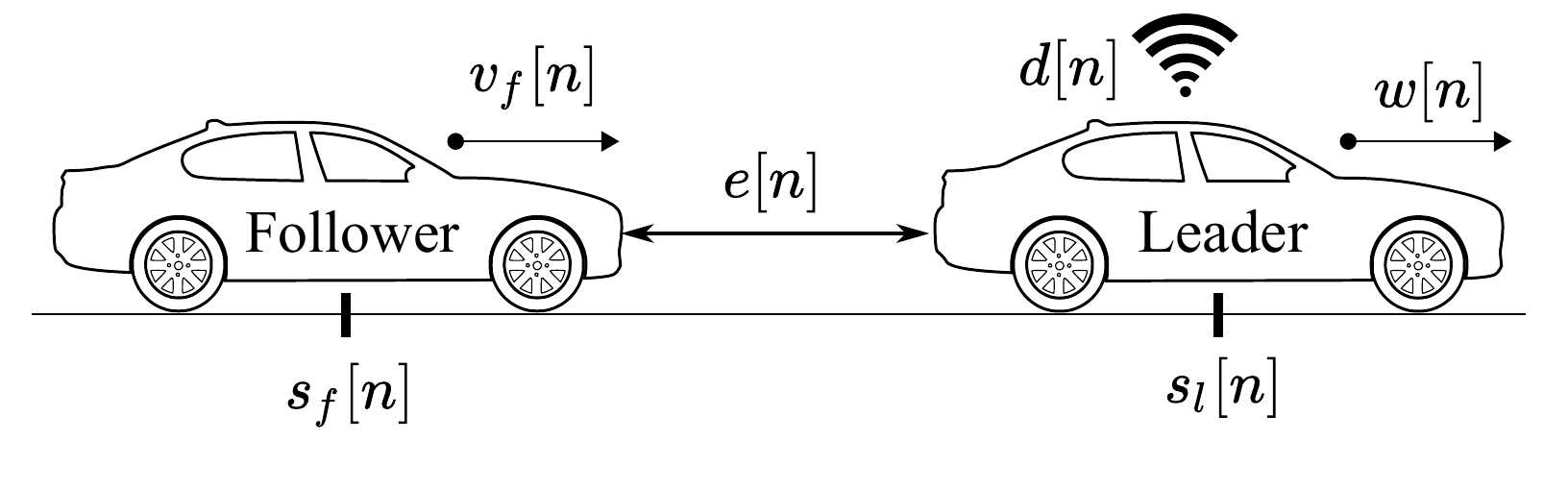}
    \caption{Discrete-time car-following scenario consisting of a leader and a follower vehicle. The leader determines the transmission decisions $d[n]$, while the follower regulates its velocity based on the received updates in order to maintain the desired inter-vehicle distance $e[n]$.}
    \label{fig:Car_Following}
\end{figure}

Let $s_\ell[n]$ and $s_f[n]$ denote the longitudinal positions of the leader and follower vehicles at time step $n$, respectively, where the system is sampled with period $T$. The vehicle dynamics are given by
\begin{align}
    s_\ell[n+1] &= s_\ell[n] + T w[n], \\
    s_f[n+1] &= s_f[n] + T v_f[n],
\end{align}
where $w[n]$ denotes the leader velocity and $v_f[n]$ represents the follower velocity, which acts as the control input of the system. Define the inter-vehicle spacing
\begin{equation}
    e[n] = s_\ell[n] - s_f[n].
\end{equation}
Then, the spacing dynamics evolve as
\begin{equation}
    e[n+1] = e[n] - T v_f[n] + T w[n].
\end{equation}
Assuming $\mathbb{E}[w[n]] = \mu'$, we decompose
\begin{equation}
    w[n] = \tilde{w}[n] + \mu',
\end{equation}
where $\tilde{w}[n]$ is zero-mean. Substituting yields
\begin{equation}
    e[n+1] = e[n] - T (v_f[n] - \mu') + T \tilde{w}[n].
\end{equation}

To regulate the spacing around a desired value $e^*$, define the state
\begin{equation}
    x[n] = e[n] - e^*.
\end{equation}
Then, the system becomes
\begin{equation}
    x[n+1] = x[n] - T (v_f[n] - \mu') + T \tilde{w}[n].
\end{equation}
Defining the control input
\begin{equation}
    u[n] = v_f[n] - \mu',
\end{equation}
we obtain the linear system
\begin{equation}
    x[n+1] = a x[n] + b u[n] + w'[n],
\end{equation}
where
\begin{equation}
    a = 1, \quad b = -T, \quad w'[n] = T \tilde{w}[n].
\end{equation}
Here, \(T=0.1\) s is the sampling period of the NGSIM US-101 trajectory data.

Note that the resulting system closely matches the form of the system in~\eqref{eqn:system_Model}. The key difference, however, lies in the modeling of the disturbance $w[n]$. In our setup, $w[n]$ is directly obtained from individual vehicle trajectories in the US 101 dataset, rather than being synthetically generated. As a result, $w[n]$ is non-zero mean.

To address this, we preprocess the disturbance sequence to obtain a zero-mean representation. Specifically, we consider two approaches. In the \emph{offline} approach, the mean of the entire trajectory is computed and subtracted from $w[n]$. In contrast, the \emph{online} approach estimates the mean in a causal manner using a sliding window over past samples, allowing the method to adapt to time-varying statistics. In both cases, the resulting disturbance $w'[n]$ is approximately zero-mean, enabling consistency with our system model formulation.

\begin{figure}[t]
    \centering
    \resizebox{0.85\linewidth}{!}{%
    \begin{minipage}{\linewidth}
        \centering
        \begin{subfigure}{0.48\linewidth}
            \centering
            \includegraphics[width=\linewidth]{noise_sequence_orig_track38_veh681.pdf}
            \caption{Original velocity/disturbance sequence.}
            \label{fig:orig_track38}
        \end{subfigure}
        \hfill
        \begin{subfigure}{0.48\linewidth}
            \centering
            \includegraphics[width=\linewidth]{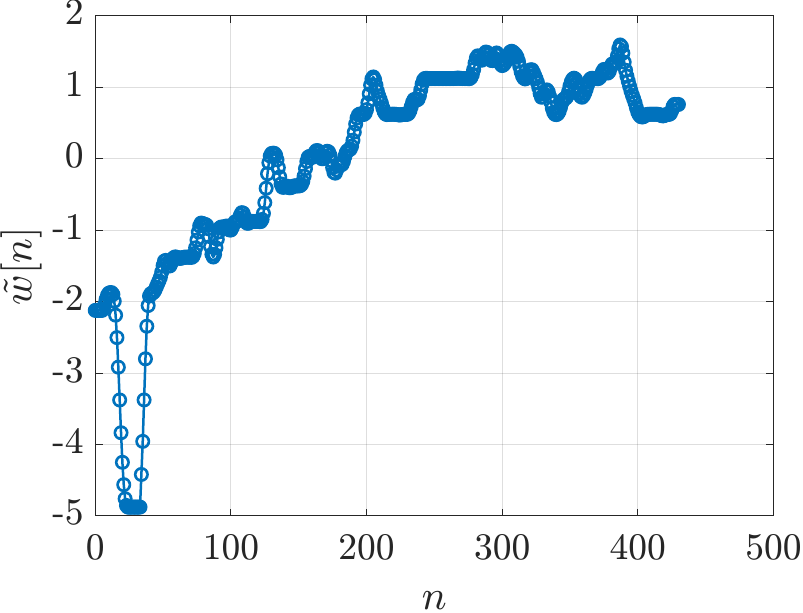}
            \caption{Corresponding effective disturbance sequence.}
            \label{fig:eff_track38}
        \end{subfigure}
    \end{minipage}
    }
    \caption{Comparison between the original and effective disturbance sequences for Vehicle 681 from the NGSIM US-101 dataset.}
    \label{fig:track38_comparison}
\end{figure}

More specifically, the effective disturbance sequence is constructed by removing the estimated mean of the measured velocity-derived disturbance and scaling it by the sampling/control gain factor. The resulting effective disturbance can be written as
\begin{equation}
    w'[n] = T \bigl( w[n] - \hat{\mu}_w[n] \bigr),
    \label{eqn:mean_normalization}
\end{equation}
where $\hat{\mu}_w[n]$ denotes the estimated mean of the disturbance sequence. For the offline case, $\hat{\mu}_w[n]$ corresponds to the sample mean computed using the entire trajectory, whereas in the online case it is estimated causally using the selected moving/windowed averaging method. An example illustrating the transformation from the original disturbance sequence to the corresponding effective disturbance sequence is shown in Fig.~\ref{fig:track38_comparison}.

We next describe the scheduling policies used in the simulations. To compare policies at the same communication rate, we use a common target rate \(p_{\mathrm{target}}\). Thus, all policies have the same mean AoI but induce different inter-scheduling interval distributions. Specifically, we consider three classes of state-independent schedulers: the as-periodic-as-possible policy, random scheduling, and Erlang-distributed inter-scheduling times with shape parameters \(k=2,4,8\). 

The random (Bernoulli) scheduling policy sets \(\{D_k\}\) as an i.i.d.\ Bernoulli sequence with
\begin{equation}
\mathbb{P}(D_k=1)=p_{max},\qquad \mathbb{P}(D_k=0)=1-p_{max}.
\end{equation}
This induces a geometric distribution on the inter-scheduling times. The periodic policy is defined in Corollary~\ref{corollary:optimal_as_periodic_as_possible}. It is the ``as-periodic-as-possible'' policy that satisfies the rate constraint for any \(0 < p_{\mathrm{target}} < 1\) by distributing scheduling instants as uniformly as possible over time. The Erlang inter-scheduling times are generated according to
\[
f_{\Delta_n}(x)=\frac{\lambda^k x^{k-1} e^{-\lambda x}}{(k-1)!}, \quad x \geq 0,
\]
where \(k\) is the shape parameter and \(\lambda\) is selected so that \(\mathbb{E}[\Delta_n]=1/p_{\mathrm{target}}\). In the discrete-time implementation, the generated inter-scheduling times are quantized to integer scheduling intervals while preserving the target mean rate. Thus, for a fixed communication rate \(p_{\mathrm{target}}\), all policies have the same mean inter-scheduling time but differ in the higher-order properties of their AoI distributions. An example of the resulting distributions for \(p_{\mathrm{target}}=0.15\) is shown in Fig.~\ref{fig:iut_histograms_offline}.

As shown in Fig.~\ref{fig:iut_histograms_offline}, the as-periodic-as-possible policy has only two possible inter-scheduling times, chosen so that the mean interval equals \(1/0.15\). The random scheduling policy produces a geometric inter-scheduling distribution, which has the largest variability among the policies considered. The Erlang-distributed policies interpolate between these two extremes: for smaller \(k\), the distribution is more dispersed and closer to an exponential shape, whereas for larger \(k\), it becomes increasingly concentrated around its mean.

\begin{figure}[!ht]
    \centering
    \includegraphics[width=\linewidth]{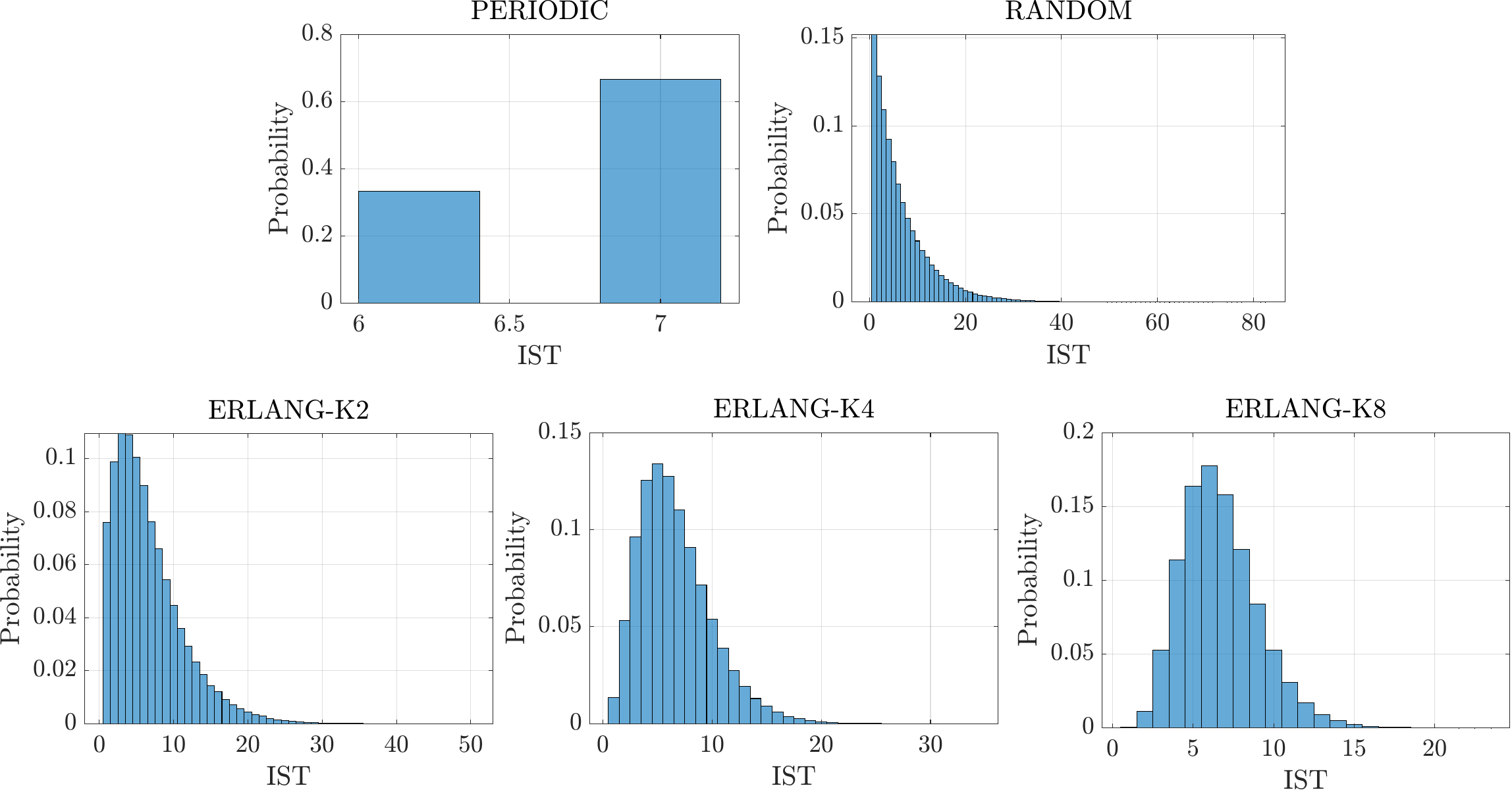}
    \caption{Inter-scheduling-time (IST) distributions for different scheduling policies with the communication rate $p_{\mathrm{target}}=0.15$.}
    \label{fig:iut_histograms_offline}
\end{figure}

On the controller side, we consider three controller policies. The first two are commonly used benchmark controllers in the networked control literature: the zero-order-hold (ZOH) controller and the impulsive (IMP) controller. The ZOH controller applies the linear feedback law~\eqref{eqn:optimal_controller} when a new observation is received and holds the control input constant until the next observation. Specifically, for scheduling instants \(\{t_n\}\),
\begin{equation}
    U^{\mathrm{ZOH}}_k =
    \begin{cases}
        -L \hat{X}_k, & \text{if } k = t_n + \tau \text{ for some } n, \\
        U_{k-1}, & \text{otherwise}.
    \end{cases}
    \label{eqn:ZOH_controller}
\end{equation}
In contrast, the impulsive controller concentrates the entire control action at the observation instant and applies no control between scheduling times:
\begin{equation}
    U^{\mathrm{IMP}}_k =
    \begin{cases}
        -\dfrac{a}{b}\,\hat{X}_k, & \text{if } k = t_n + \tau \text{ for some } n, \\
        0, & \text{otherwise}.
    \end{cases}
    \label{eqn:imp_controller}
\end{equation}
where \(\hat{X}_k\) denotes the conditional state estimate appearing in~\eqref{eqn:opt_controller_sym}. Since the impulsive controller exerts no input during the inter-scheduling interval, the feedback gain is chosen as \(-a/b\) rather than \(-L\). This gain corresponds to the limiting case of the LQR controller as \(r \to 0\), where control effort is penalized negligibly and the nominal state can be driven to zero in a single step. As a third benchmark, we consider the symmetric controller, defined for all \(m \ge \tau\) by
\begin{equation}
    U_{t_n+m}
    = -L\,\hat{X}_{t_n+m}
    = -L\,(a-bL)^{m-\tau}
    \left(
    a^{\tau} X_{t_n}
    +
    \sum_{j=0}^{\tau-1} a^{\tau-1-j} b\,U_{t_n+j}
    \right).
    \label{eqn:opt_controller_sym}
\end{equation}

In our first analysis, we validate the AoI-equivalent cost in
\eqref{eqn:eqv_cost_iid_case}, which was derived for i.i.d. disturbances.
However, to compare the LQR cost with this equivalent cost, we must also
include the MSE-equivalent term obtained in \eqref{eqn:eqv_prob_last_eqn}.
Therefore, the complete equivalent cost is given by
\begin{equation}
    \left(
    L^2(r+b^2K)\left(
    \frac{1}{2}
    \frac{\mathbb{E}[\Delta_n^2]}
    {\mathbb{E}[\Delta_n]}
    +
    \tau+\frac{1}{2}
    \right) + K
    \right) \sigma_W^2 .
    \label{eqn:eqv_cost_iid_complete}
\end{equation}
To evaluate the AoI-equivalent cost in the i.i.d. setting, we use the
empirical distribution obtained from all velocity profiles, shown in
Fig.~\ref{fig:empirical_W_dist}. We then draw independent samples from this
empirical distribution and compute the running average LQR cost. Figure~\ref{fig:theory_vs_sim_iid} compares the theoretical cost in
\eqref{eqn:eqv_cost_iid_complete} with the steady-state i.i.d.-sampling
simulation results as a function of the target communication rate
$p_{\mathrm{target}}$. As shown, the theoretical expression closely matches
the simulation results.

\begin{figure}[!ht]
    \centering
    \includegraphics[width=0.75\linewidth]{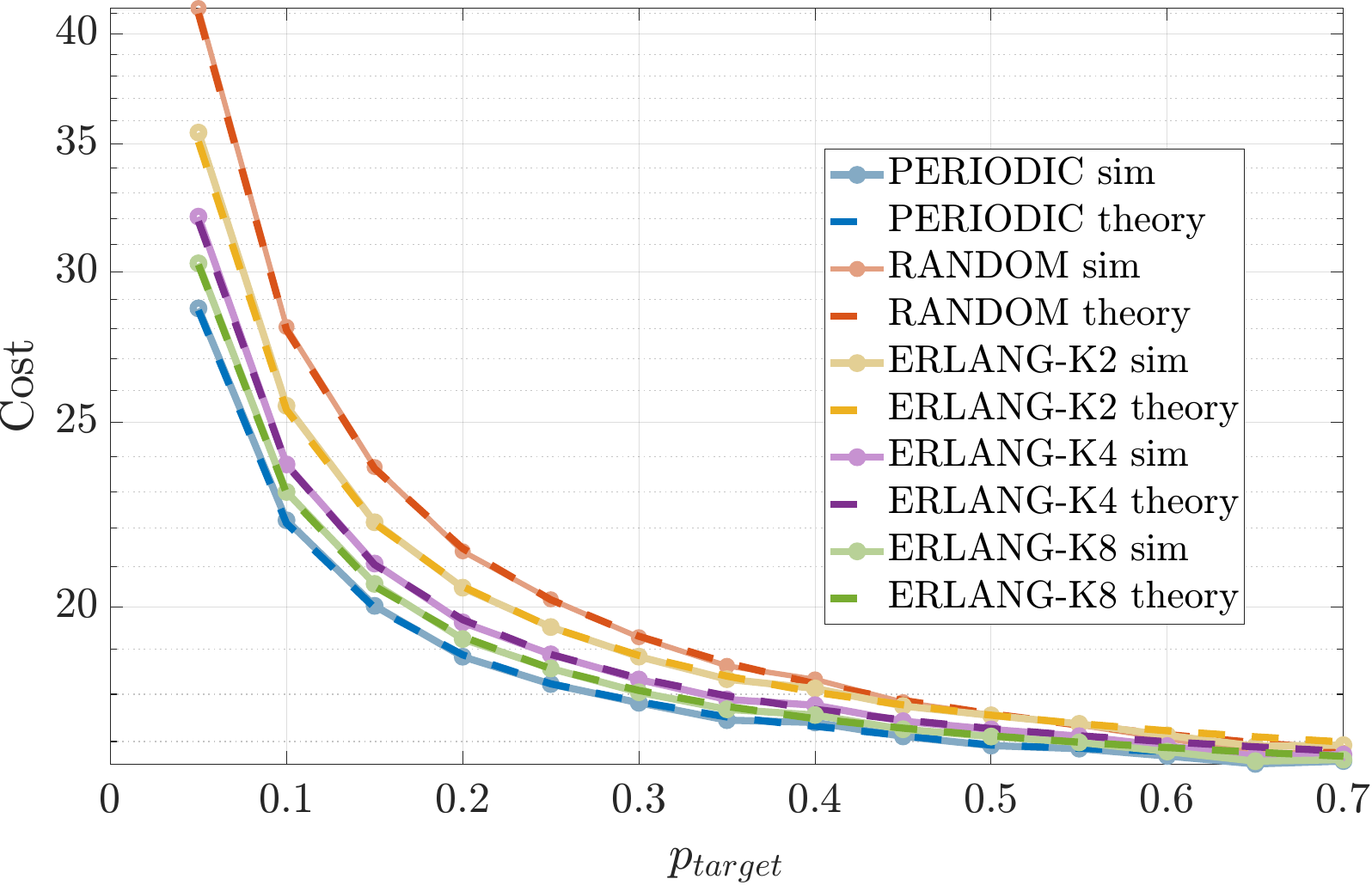}
    \caption{Comparison between the theoretical AoI-equivalent cost in
    \eqref{eqn:eqv_cost_iid_complete} and the simulated running average LQR
    cost under i.i.d. sampling from the empirical disturbance distribution. The
    close agreement confirms the validity of the equivalent-cost expression
    for the i.i.d. disturbance case.}
    \label{fig:theory_vs_sim_iid}
\end{figure}

Now consider the example introduced in Section~\ref{sec:AoI_iid_Noise},
where the two-point randomized scheduler has the inter-scheduling interval
distribution shown in Fig.~\ref{fig:iut_distribution_two_point_iid}. In Fig.~\ref{fig:Running_Avg_IID_Two_Point_1_12_0_05}, we compare the running-average LQR cost of this scheduler with that of its shifted version, where $\bar{p}$ denotes the empirically achieved communication rate.
As shown, the theoretical values closely match the simulation results. More
importantly, increasing the communication rate from approximately \(0.392\) to
\(0.645\) degrades the tracking performance.

\begin{figure}[!ht]
    \centering
    \includegraphics[width=0.65\linewidth]{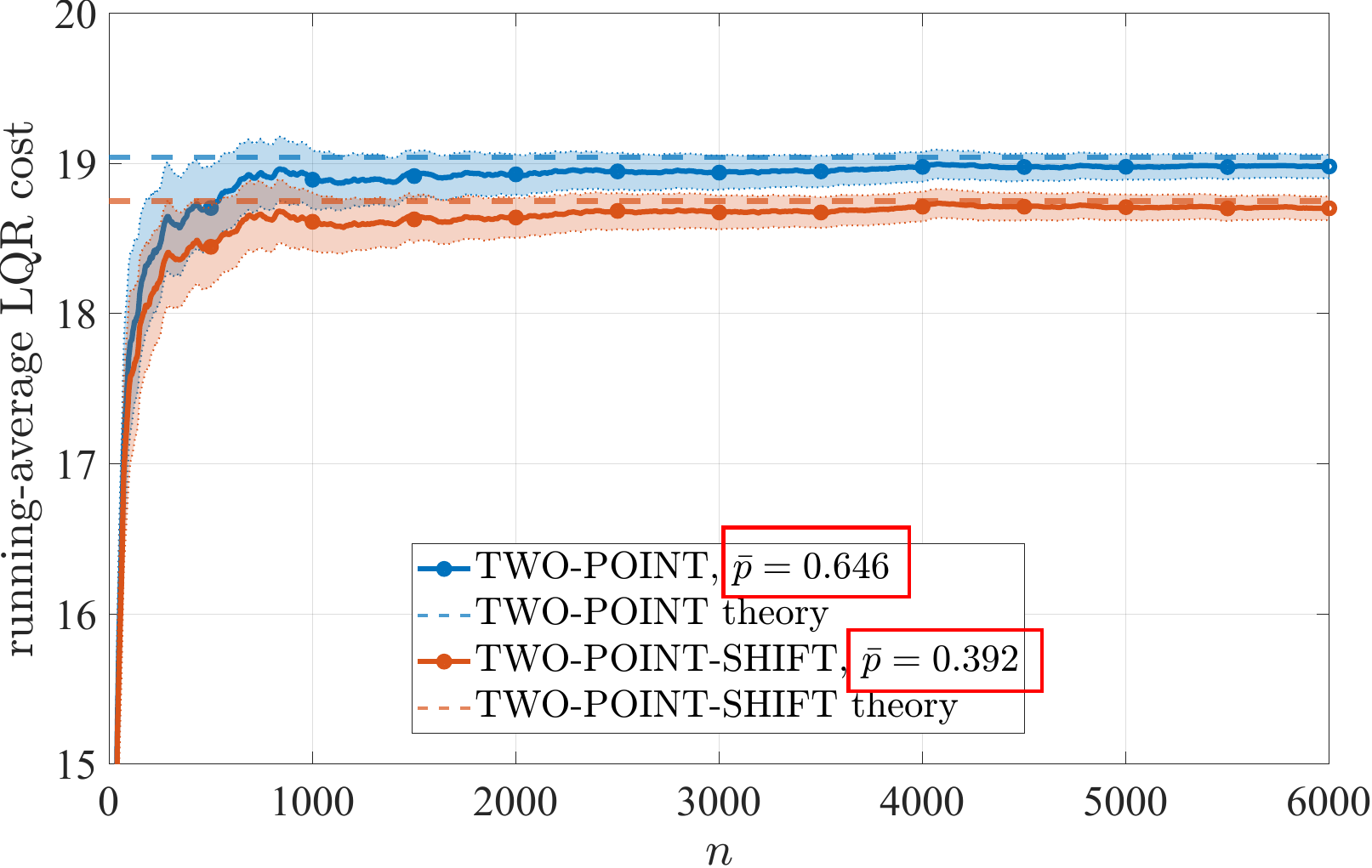}
    \caption{Running-average LQR cost for the original and shifted two-point randomized schedulers in the i.i.d. disturbance case, with \(\delta_s=1\), \(\delta_\ell=12\), \(p_\ell=0.05\), \(s=1\), and \(\tau=1\). The shaded regions represent the \(95\%\) confidence intervals.}
    \label{fig:Running_Avg_IID_Two_Point_1_12_0_05}
\end{figure}

The effect becomes even more pronounced in the second example. Consider the
two-point randomized scheduler shown in
Fig.~\ref{fig:iut_distribution_two_point_1_50}, with \(\delta_s=1\),
\(\delta_\ell=50\), \(p_\ell=0.02\), and its shifted version with \(s=10\).
The original scheduler has mean inter-scheduling interval
\[
    \mathbb{E}[\Delta]=0.98(1)+0.02(50)=1.98,
\]
corresponding to a communication rate of approximately \(0.506\). The shifted
scheduler has
\[
    \mathbb{E}[\widetilde{\Delta}]
    =
    0.98(11)+0.02(60)
    =
    11.98,
\]
corresponding to a communication rate of approximately \(0.083\). As shown in
Fig.~\ref{fig:Running_Avg_IID_Two_Point_1_50_0_02}, increasing the
communication rate by roughly a factor of six leads to worse tracking
performance.

\begin{figure}[!ht]
    \centering
    \begin{subfigure}{0.48\linewidth}
        \centering
        \includegraphics[width=0.7\linewidth]{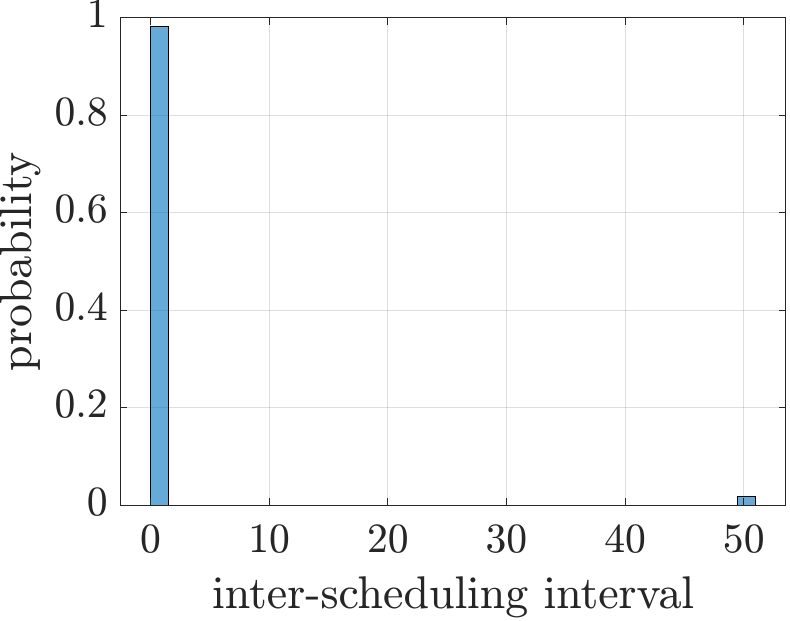}
        \caption{Original two-point distribution.}
        \label{fig:iut_distribution_two_point_1_50_original}
    \end{subfigure}
    \hfill
    \begin{subfigure}{0.48\linewidth}
        \centering
        \includegraphics[width=0.7\linewidth]{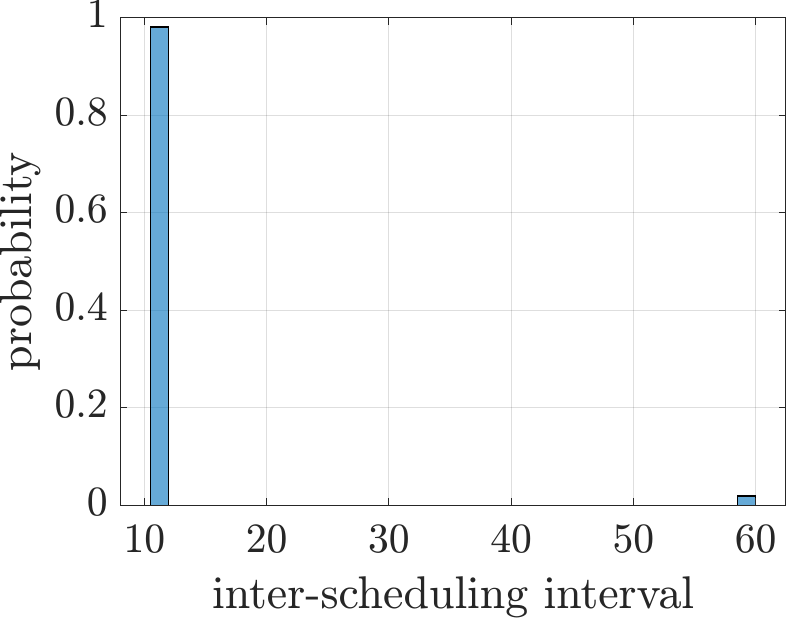}
        \caption{Shifted two-point distribution.}
        \label{fig:iut_distribution_two_point_1_50_shifted}
    \end{subfigure}
    \caption{Original and shifted two-point inter-scheduling interval
    distributions for \(\delta_s=1\), \(\delta_\ell=50\),
    \(p_\ell=0.02\), and \(s=10\). The shifted scheduler increases each
    inter-scheduling interval by ten time steps, thereby reducing the
    communication rate.}
    \label{fig:iut_distribution_two_point_1_50}
\end{figure}

This example has the same channel interpretation discussed earlier. Suppose
the channel is usually fast, but with probability \(2\%\) the effective update
cycle lasts \(50\) time steps instead of one time step. Although it may seem
natural to transmit as frequently as possible, the simulation shows the
opposite behavior: waiting longer between transmission attempts can reduce the
running-average LQR cost.

As discussed in Section~\ref{sec:AoI_iid_Noise}, our analysis considers sampled AoI. However, even if one considers the classical instantaneous AoI, the equivalent LQR cost is not determined solely by mean AoI when the open-loop gain is perturbed slightly away from \(|a|=1\), as shown in Fig.~\ref{fig:iid_two_point_a_near_one}. As can be seen, even for the small perturbation \(a=0.993\), the shifted scheduler has a higher mean AoI but a lower LQR cost.

\begin{figure}[!ht]
    \centering
    \includegraphics[width=0.65\linewidth]{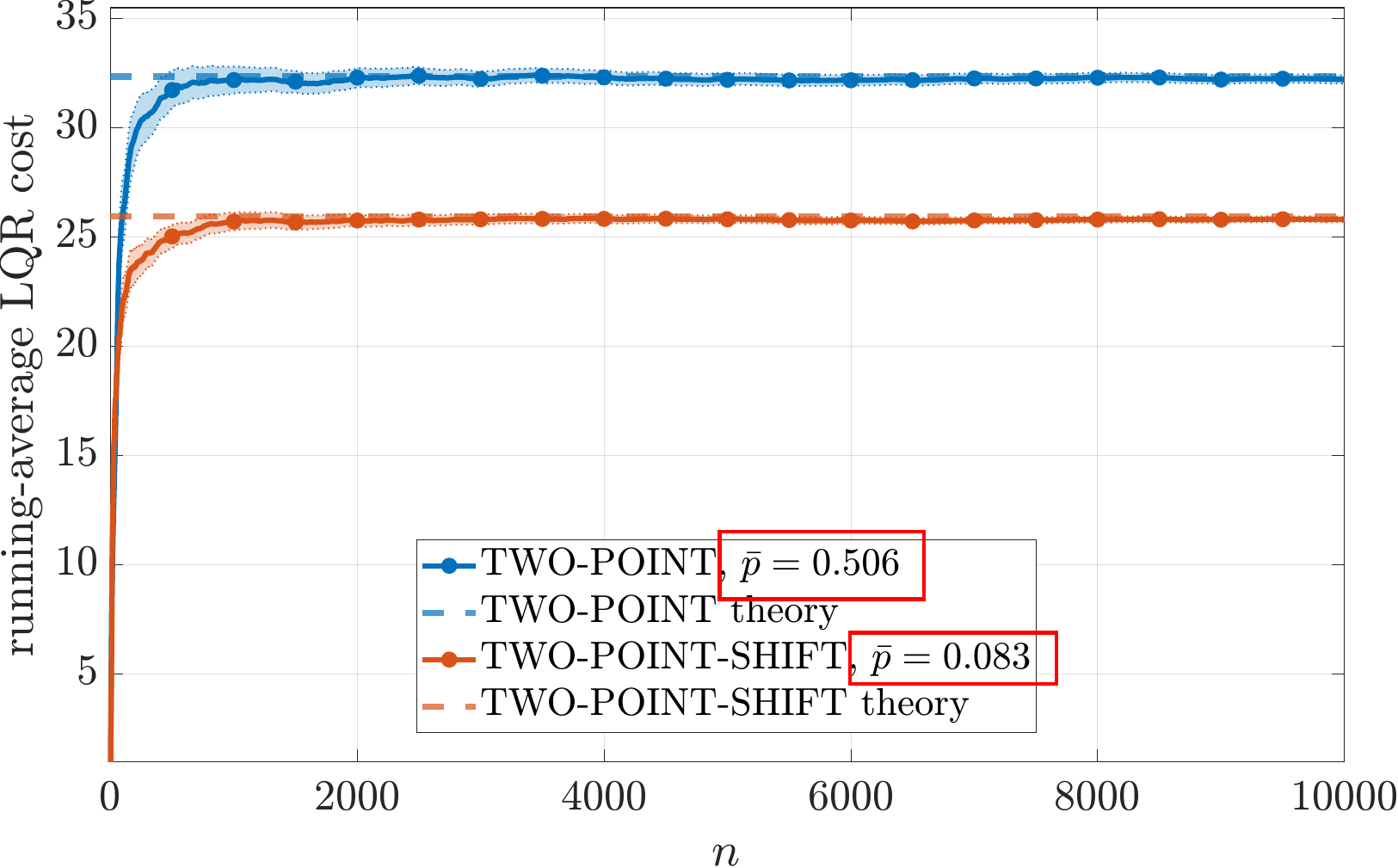}
    \caption{Running-average LQR cost for the original and shifted two-point
    randomized schedulers in the i.i.d. disturbance case, with
    \(\delta_s=1\), \(\delta_\ell=50\), \(p_\ell=0.02\), \(s=10\), and
    \(\tau=1\). The shaded regions represent the \(95\%\) confidence intervals.}
    \label{fig:Running_Avg_IID_Two_Point_1_50_0_02}
\end{figure}

\begin{figure}[!ht]
	\centering
	\includegraphics[width=0.65\linewidth]{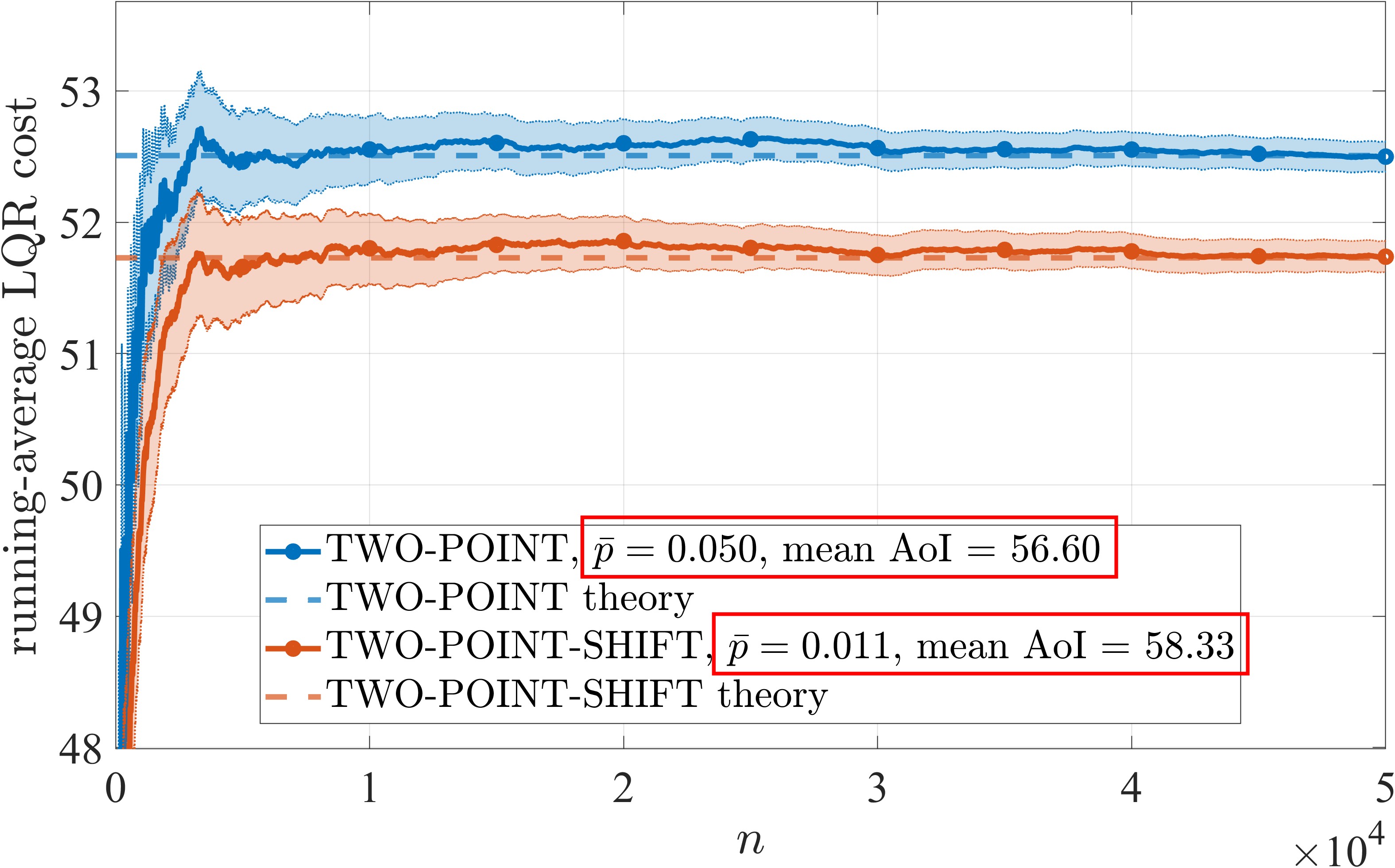}
	\caption{Running-average LQR cost for two i.i.d. two-point schedulers near the marginally stable case. Although the shifted scheduler has a larger classical mean AoI, it achieves a smaller LQR tracking cost, illustrating that the equivalence between mean AoI and the LQR-induced objective is fragile away from \(|a|=1\).}
	\label{fig:iid_two_point_a_near_one}
\end{figure}

We now turn to the correlated-disturbance case, where the simulations are
performed directly using the velocity profiles from the US-101 dataset.
Unlike the iid case, these profiles preserve the temporal correlations
present in the real traffic data, as expected in car-following traffic
dynamics. Therefore, the analysis in this section follows the
correlated-noise framework developed in
Section~\ref{sec:AoI_Correlated_Noise}. To connect the data-driven
disturbances with this framework, we model the autocorrelation function of
the disturbance process using the exponential form in
\eqref{eqn:noise_ACF}, where $\beta \ge 0$ controls the correlation strength.

We next estimate \(\beta\) from the empirical velocity data and evaluate the quality of the exponential autocorrelation function (ACF) approximation. As a preliminary observation, the mean-normalized velocity process \(w[n]\) has an approximately Gaussian marginal distribution, as shown in Fig.~\ref{fig:empirical_W_dist}. The autocorrelation analysis, however, requires additional care. In the simulations, each vehicle trajectory is treated as one sample path of the disturbance process \(w[n]\). Consequently, the average of the sample-path ACFs need not coincide with the ACF of the underlying stochastic process, especially when the trajectories are finite and heterogeneous. This effect is illustrated in Fig.~\ref{fig:exp_corr_fit}. The black curve shows the target exponential ACF, while the remaining curves show the sample-path ACFs of 10 synthetic trajectories generated from the same correlated process. The average of these sample-path ACFs can differ noticeably from the target ACF. Therefore, to estimate \(\beta\), we use a reverse fitting procedure. For each candidate value of \(\beta\), we generate \(10^4\) synthetic sample paths with the empirical marginal distribution shown in Fig.~\ref{fig:empirical_W_dist} and the exponential ACF model above. We then compute the average sample-path ACF and compare it with the average ACF obtained from the 3900 empirical vehicle-speed profiles. The resulting MSE as a function of \(\beta\) is shown in Fig.~\ref{fig:beta_mse}. The best-fitting value is \(\beta=0.007558\), indicating that the velocity process is highly correlated over time. The corresponding target exponential ACF, synthetic average ACF, and empirical average ACF are shown in Fig.~\ref{fig:beta_best_fit}.

\begin{figure}[!t]
    \centering
    \resizebox{0.9\linewidth}{!}{%
    \begin{minipage}{\linewidth}
        \centering
        
        \begin{subfigure}[t]{0.48\linewidth}
            \centering
            \includegraphics[height=0.24\textheight,keepaspectratio]{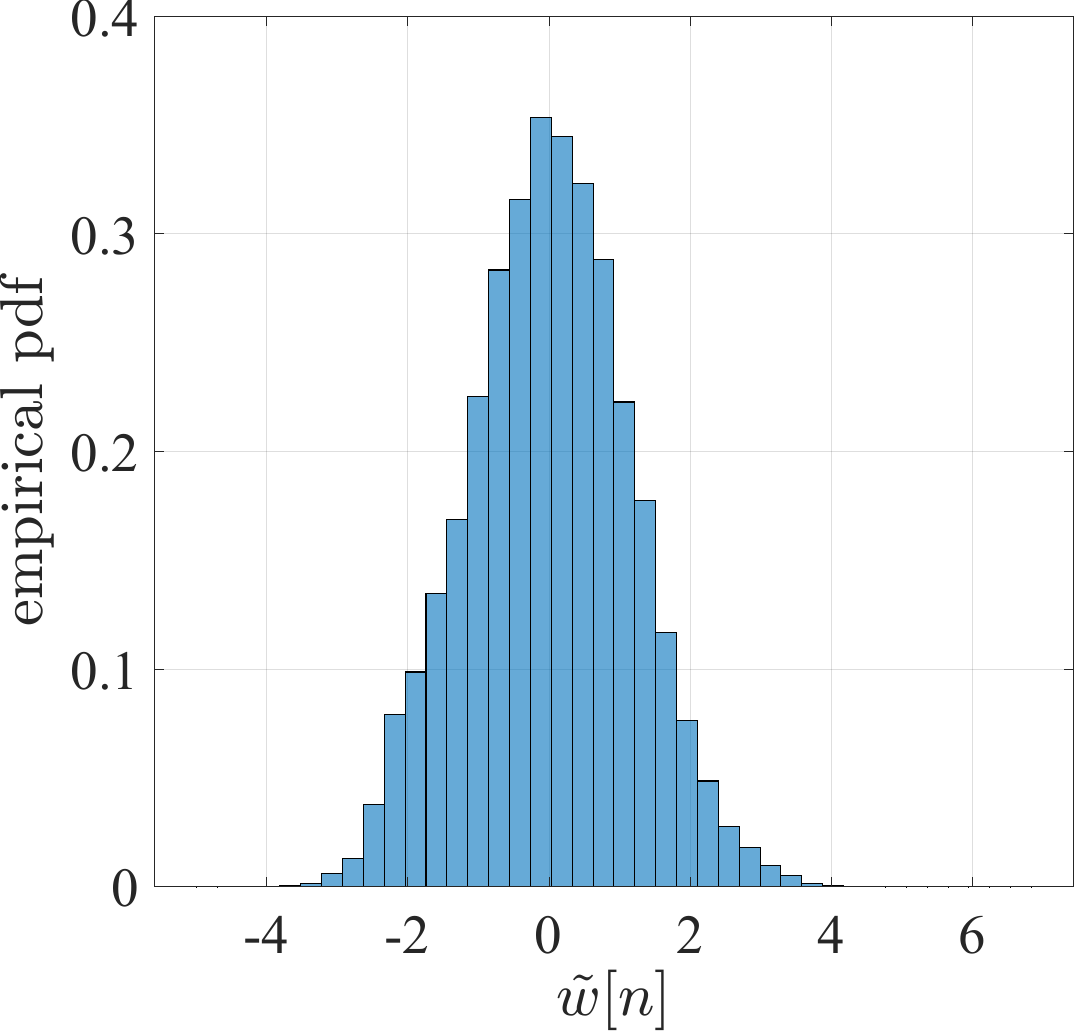}
            \caption{Empirical normalized distribution of the process $w[n]$.}
            \label{fig:empirical_W_dist}
        \end{subfigure}
        \hfill
        \begin{subfigure}[t]{0.48\linewidth}
            \centering
            \includegraphics[height=0.24\textheight,keepaspectratio]{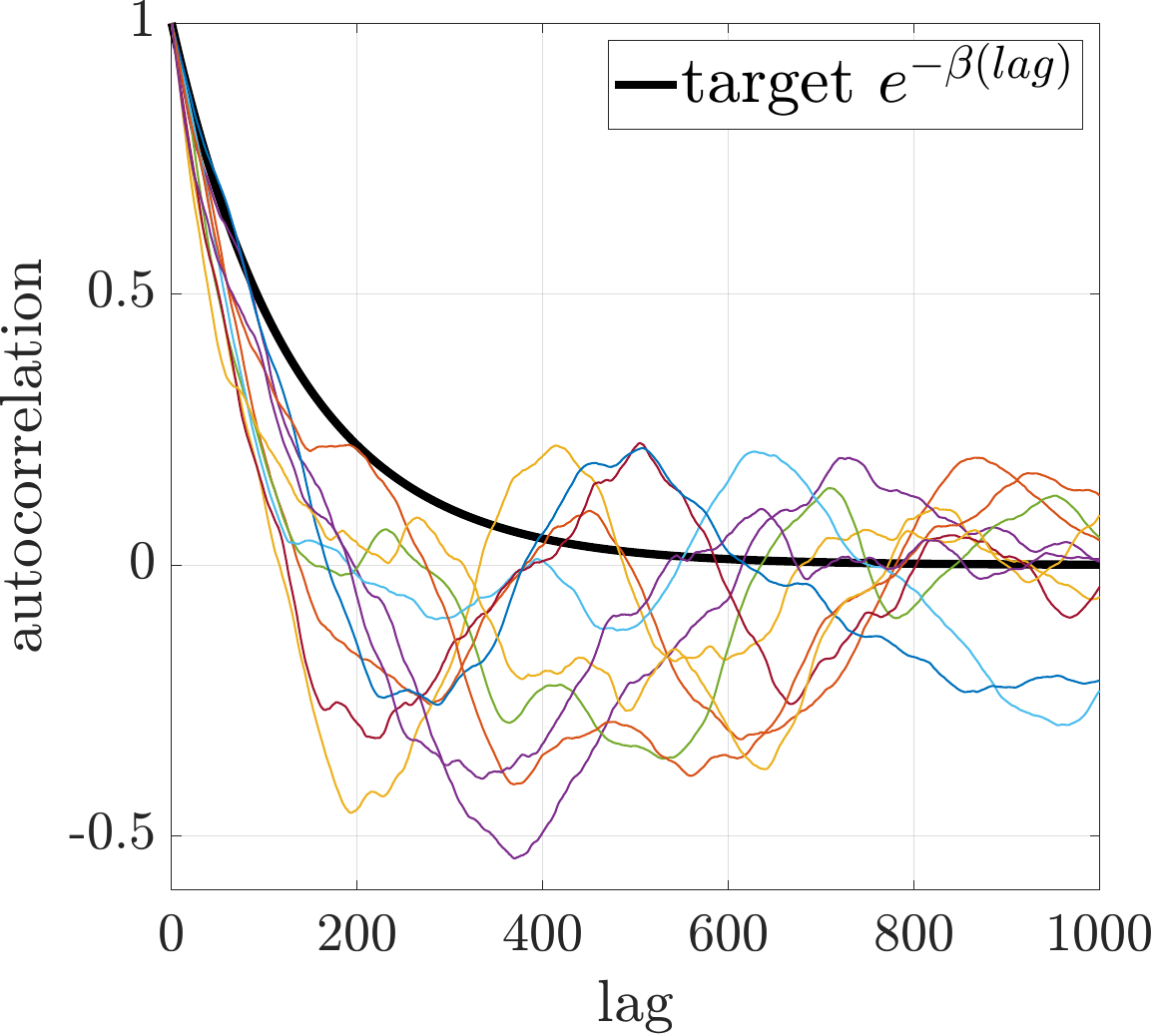}
            \caption{Exponential correlation model fitting.}
            \label{fig:exp_corr_fit}
        \end{subfigure}

        \vspace{0.3cm}

        \begin{subfigure}[t]{0.48\linewidth}
            \centering
            \includegraphics[height=0.24\textheight,keepaspectratio]{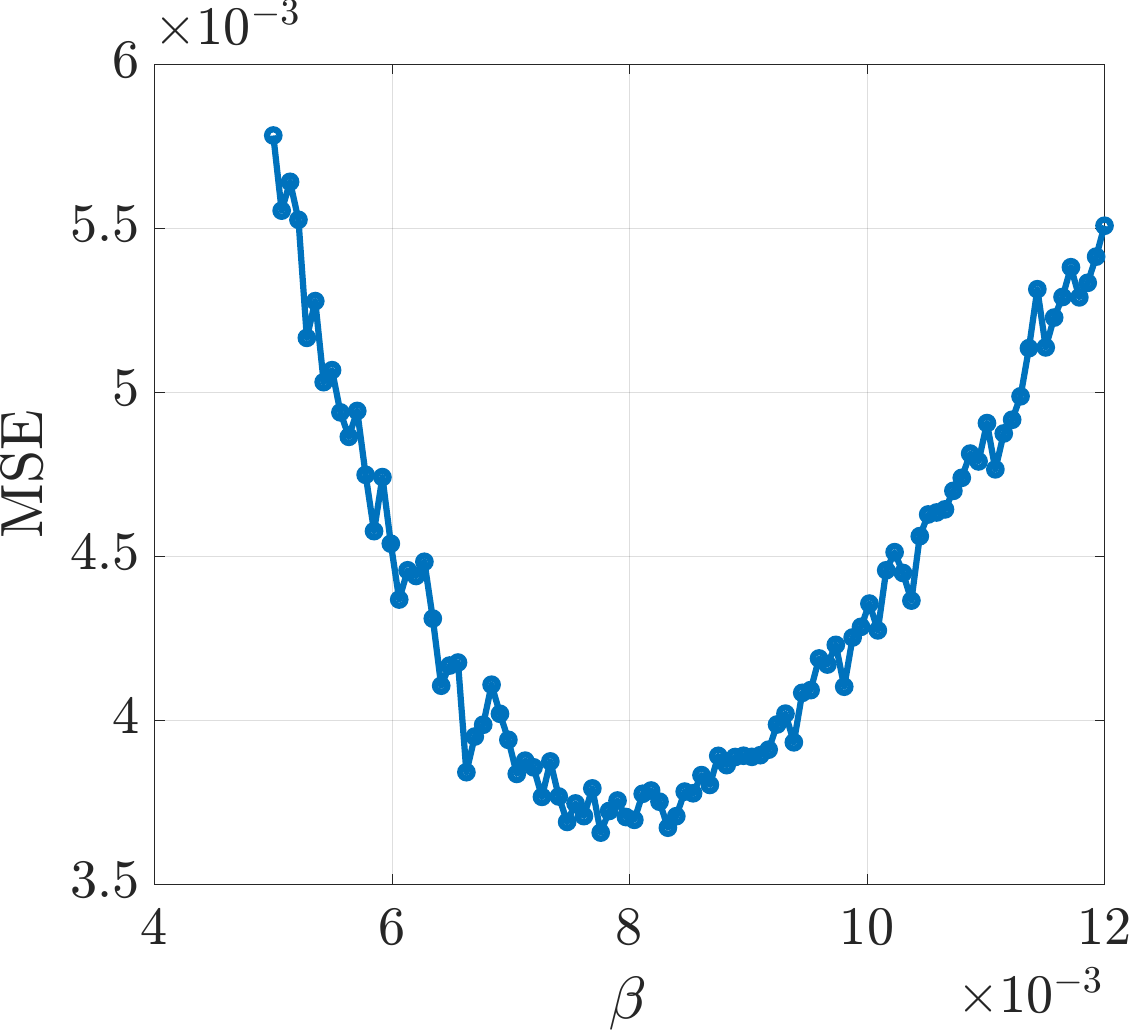}
            \caption{MSE versus the correlation parameter $\beta$.}
            \label{fig:beta_mse}
        \end{subfigure}
        \hfill
        \begin{subfigure}[t]{0.48\linewidth}
            \centering
            \includegraphics[height=0.24\textheight,keepaspectratio]{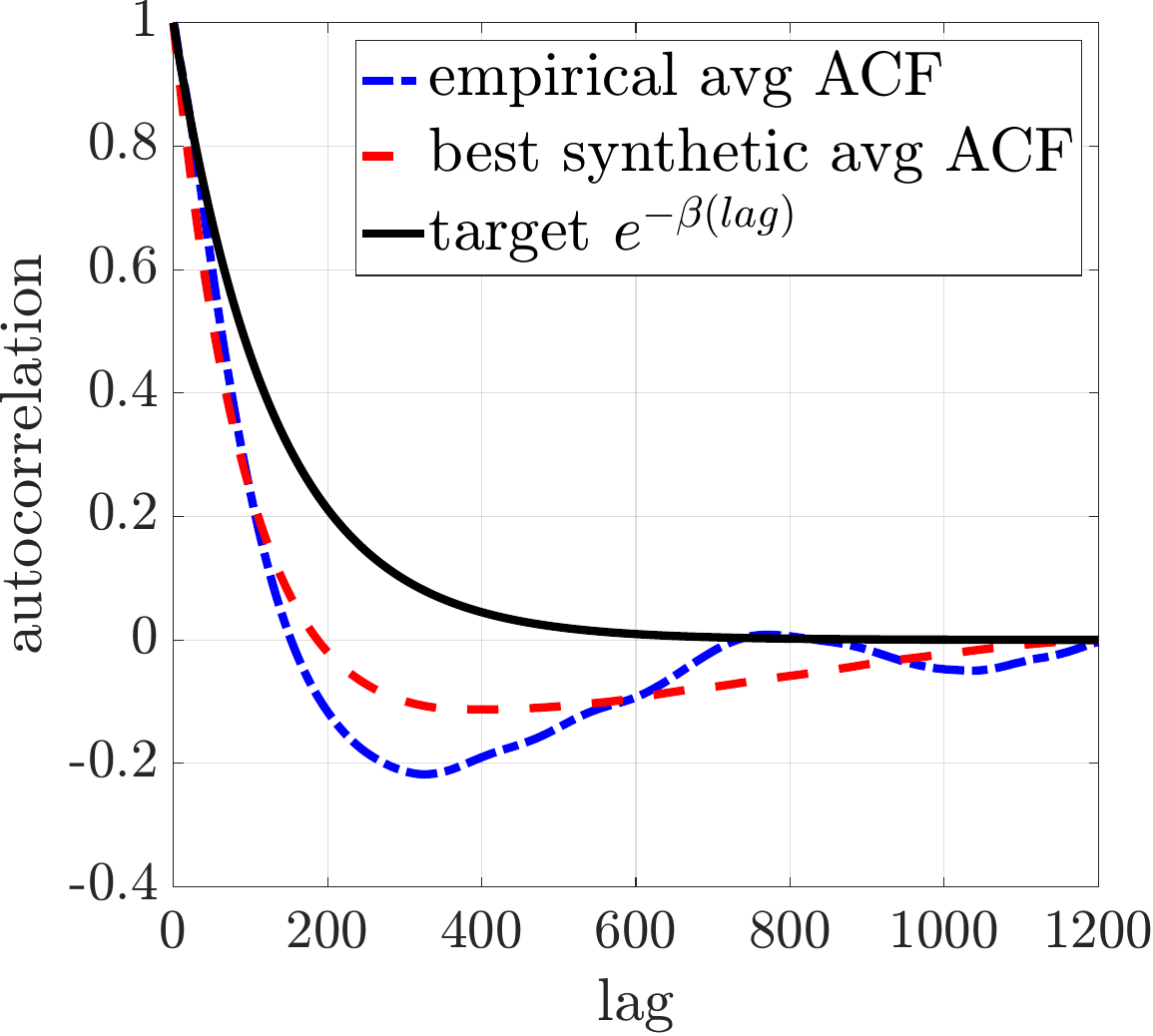}
            \caption{Best-fit correlation parameter selection where $\beta=0.007558$.}
            \label{fig:beta_best_fit}
        \end{subfigure}
    \end{minipage}
    }
    \caption{Statistical analysis and exponential correlation modeling of the normalized velocity process.}
    \label{fig:correlation_analysis}
\end{figure}

For the first analysis of the correlated-disturbance case, we revisit the
two-point randomized scheduler introduced in
Section~\ref{sec:AoI_Correlated_Noise}. We first consider the distribution
shown in Fig.~\ref{fig:iut_distribution_two_point_correlated}. The resulting
running-average LQR cost, obtained by directly using the velocity profiles
from the US-101 dataset, is shown in
Fig.~\ref{fig:Running_Avg_CORR_Two_Point_1_12_0_02}. In this simulation, we
use \(q=1\), \(r=0.001\), and the impulsive controller, consistent with the
theoretical analysis in Section~\ref{sec:AoI_Correlated_Noise}. The result
is striking: increasing the communication rate from \(0.450\) to \(0.821\),
almost doubling the rate, worsens the tracking performance.

\begin{figure}[!t]
    \centering
    \includegraphics[width=0.65\linewidth]{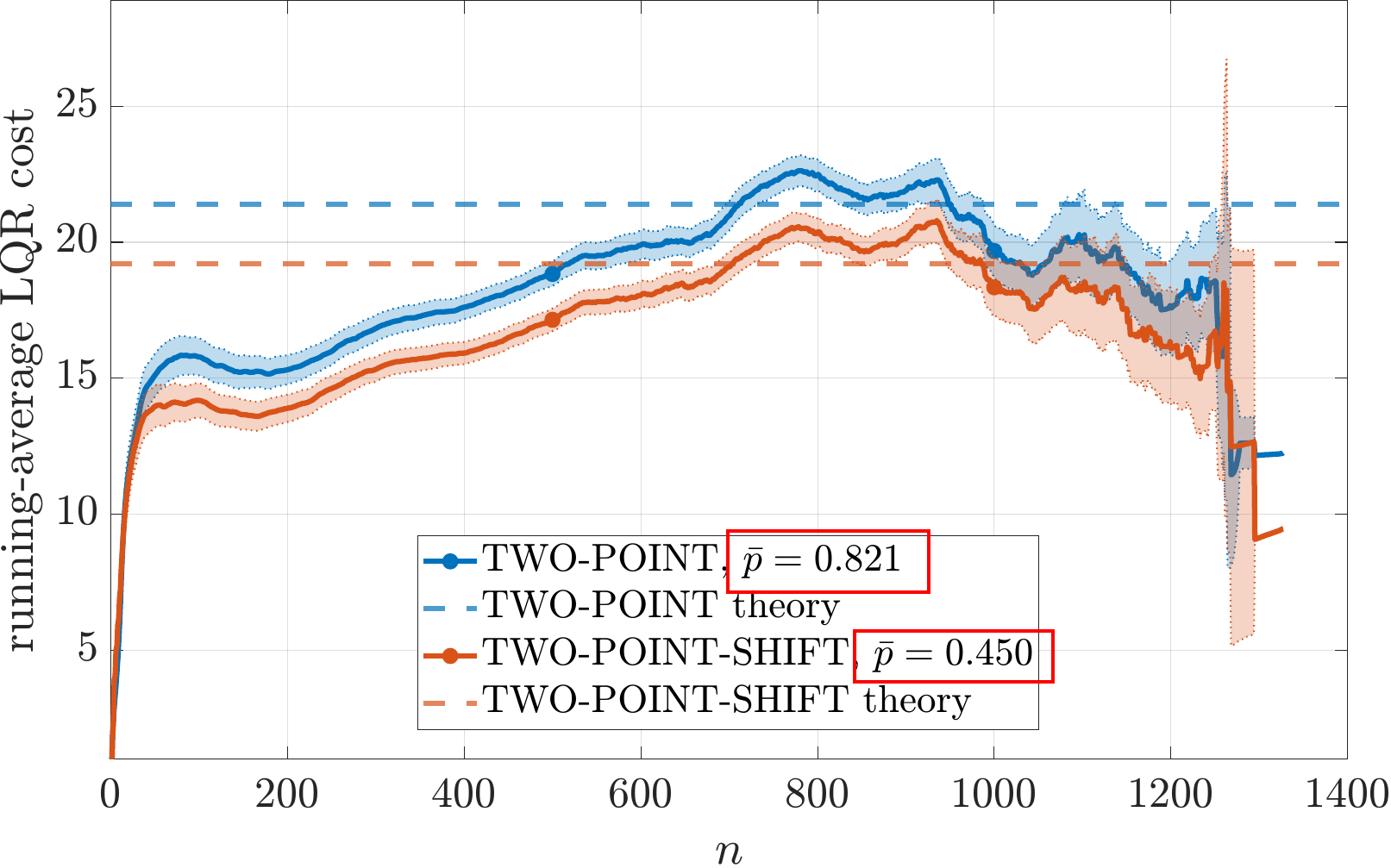}
    \caption{Running-average LQR cost for the original and shifted two-point
    randomized schedulers under correlated US-101 velocity disturbances, with
    \(\delta_s=1\), \(\delta_\ell=12\), \(p_\ell=0.02\), \(s=1\), and
    \(\tau=1\).}
    \label{fig:Running_Avg_CORR_Two_Point_1_12_0_02}
\end{figure}

\begin{figure}[!ht]
	\centering
	\includegraphics[width=0.65\linewidth]{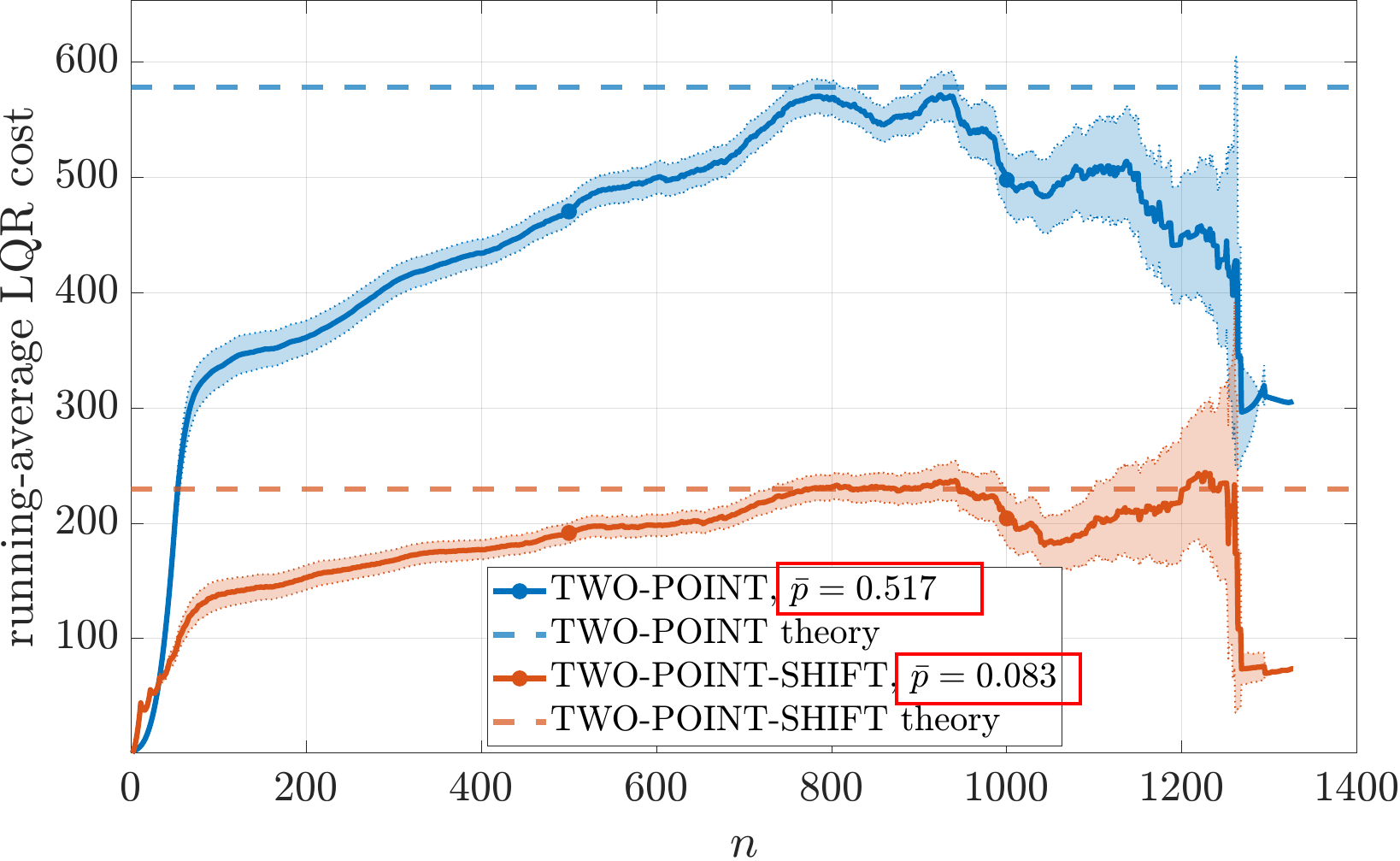}
	\caption{Running-average LQR cost for the original and shifted two-point
		randomized schedulers under correlated US-101 velocity disturbances, with
		\(\delta_s=1\), \(\delta_\ell=50\), \(p_\ell=0.02\), \(s=10\), and
		\(\tau=1\).}
	\label{fig:Running_Avg_CORR_Two_Point_1_50_0_02}
\end{figure}
We next give a more pronounced example, analogous to the one used in the
i.i.d. sampling case. Consider the two-point randomized scheduler shown in
Fig.~\ref{fig:iut_distribution_two_point_1_50}. Using the same simulation
settings as above, the running-average LQR cost is shown in
Fig.~\ref{fig:Running_Avg_CORR_Two_Point_1_50_0_02}. In this case, increasing
the communication rate by roughly a factor of six still degrades the tracking
performance, with a degradation exceeding \(50\%\).

We further examine this example from a safety-oriented perspective. Since
\(r=0.001\), the LQR objective is dominated by the state term and is therefore
approximately proportional to \(\mathbb{E}[X_k^2]\). Recall that the state is
defined as
\[
X_k = e_k - e^*,
\]
where \(e_k\) is the inter-vehicle spacing and \(e^*\) is the desired spacing.
Thus, \(\mathbb{E}[X_k^2]\) measures the average squared deviation of the
actual spacing from the desired spacing. Moreover, since
\[
X_k < -e^*
\quad \Longleftrightarrow \quad
e_k < 0,
\]
the probability
\[
\Pr(X_k < -e^*)
\]
can be interpreted as the probability that the follower vehicle crashes into
the leader vehicle, or more generally as a crash probability.

Figure~\ref{fig:tail_probs} shows the corresponding left-tail probabilities
for the two schedulers in Fig.~\ref{fig:iut_distribution_two_point_1_50}.
The shifted scheduler, despite having a substantially larger mean AoI and a
much lower communication rate, yields a much smaller probability of large
negative spacing errors over most relevant desired spacings. The only exception
occurs for very small desired spacings, below approximately \(7\) m. For
moderate and large desired spacings, however, the shifted scheduler provides a
significantly lower crash probability. This improvement is quantified in
Fig.~\ref{fig:crash_prob_percentage}, which plots the percentage decrease in
the crash probability relative to the original two-point scheduler. Around
\(e^*=30\) m, the shifted scheduler reduces the probability \(\Pr(X_k<-e^*)\) by nearly \(80\%\). Thus, even though the shifted scheduler
communicates roughly six times less frequently, it substantially reduces the
probability that the follower vehicle crashes into the leader vehicle over a
wide range of desired spacings. This provides an even stronger illustration that minimizing mean AoI alone is not sufficient for networked control systems.

This example also has the same channel interpretation discussed earlier. It
can be viewed as a channel that is usually fast, but with probability
\(2\%\) produces an effective update cycle of \(50\) time steps. The practical
US-101 simulation shows that transmitting as frequently as possible is not
necessarily best. In fact, in this example, scheduling almost every time step
can be more than twice as costly as scheduling at roughly one-tenth of that
rate. These results, for both i.i.d. sampling and correlated disturbances,
support two key conclusions. First, the equivalent optimization problems in
Section~\ref{sec:AoI_Correlated_Noise}
correctly capture the relevant scheduling tradeoffs. Second, minimizing mean AoI alone is not sufficient. The
entire AoI distribution must be accounted for through the optimization problems
introduced in Section~\ref{sec:AoI_iid_Noise} and~\ref{sec:AoI_Correlated_Noise}.

\begin{figure}[!t]
	\centering
	\includegraphics[width=0.6\linewidth]{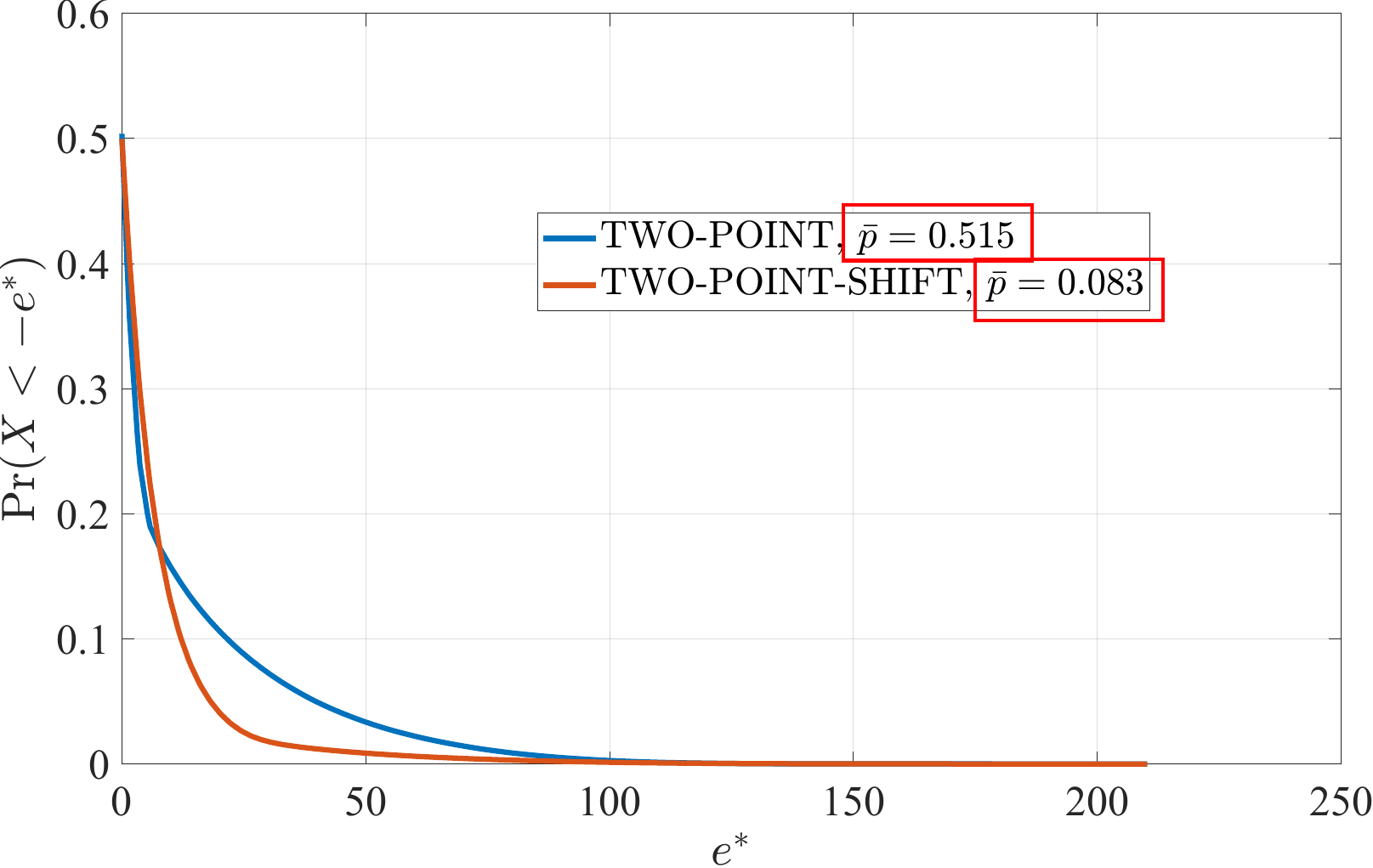}
	\caption{Crash-risk probability \(\Pr(X_k<-e^*)\) for the original and
		shifted two-point schedulers under correlated US-101 velocity
		disturbances. The scheduler parameters are \(\delta_s=1\),
		\(\delta_\ell=50\), \(p_\ell=0.02\), shift \(s=10\), and \(\tau=1\).}
	\label{fig:tail_probs}
\end{figure}

\begin{figure}[!t]
	\centering
	\includegraphics[width=0.6\linewidth]{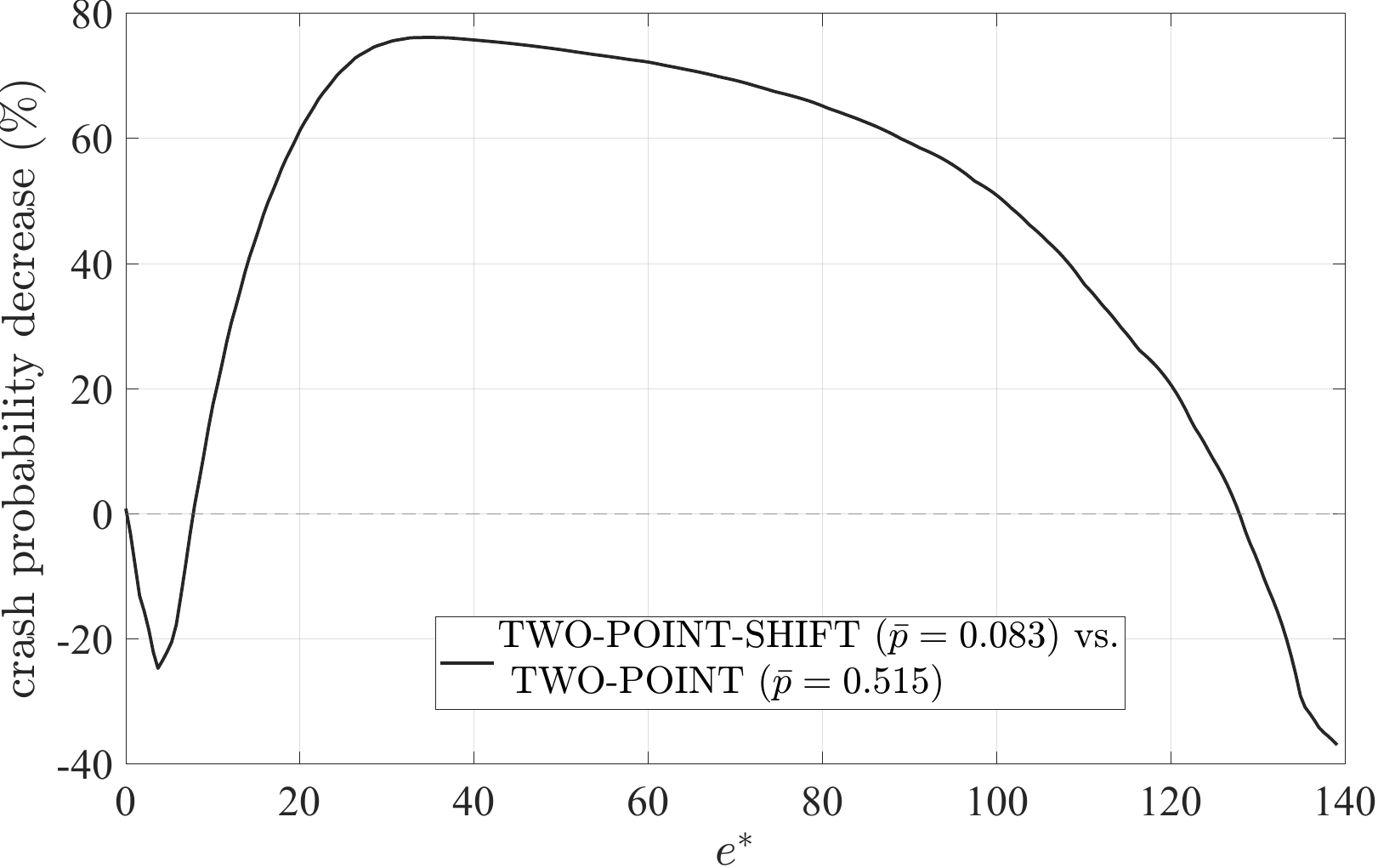}
	\caption{Percentage decrease in crash-risk probability,
		\(100(\Pr_A(X_k<-e^*)-\Pr_B(X_k<-e^*))/\Pr_A(X_k<-e^*)\), for the
		shifted two-point scheduler relative to the original scheduler. Positive
		values indicate a reduction in the probability that the follower vehicle
		crashes into the leader vehicle. The parameters are \(\delta_s=1\),
		\(\delta_\ell=50\), \(p_\ell=0.02\), shift \(s=10\), and \(\tau=1\).}
	\label{fig:crash_prob_percentage}
\end{figure}

In the theoretical analysis for the correlated-disturbance case, we considered
the impulsive controller for notational simplicity and to isolate the dominant
state-tracking contribution. Accordingly, most of the simulations above focus
on the small-\(r\) impulsive-controller regime. However, when we repeat the
same experiment using the ZOH controller and the optimal controller in
\eqref{eqn:optimal_controller}, we observe the same qualitative behavior.
Fig.~\ref{fig:running_avg_corr_two_point_optimal_controller} shows the
running-average LQR cost for the two-point schedulers in
Fig.~\ref{fig:iut_distribution_two_point_1_50} under the optimal controller.
Although the shifted distribution has a lower communication rate, it achieves
better tracking performance. A similar behavior is observed for the ZOH
controller in Fig.~\ref{fig:running_avg_corr_two_point_zoh_controller}, where
the performance gap is even larger. These results show that the phenomenon is
not an artifact of the impulsive-controller approximation used in the
theoretical analysis; it can also occur under other controller structures.

\begin{figure}[!t]
	\centering
	\includegraphics[width=0.65\linewidth]{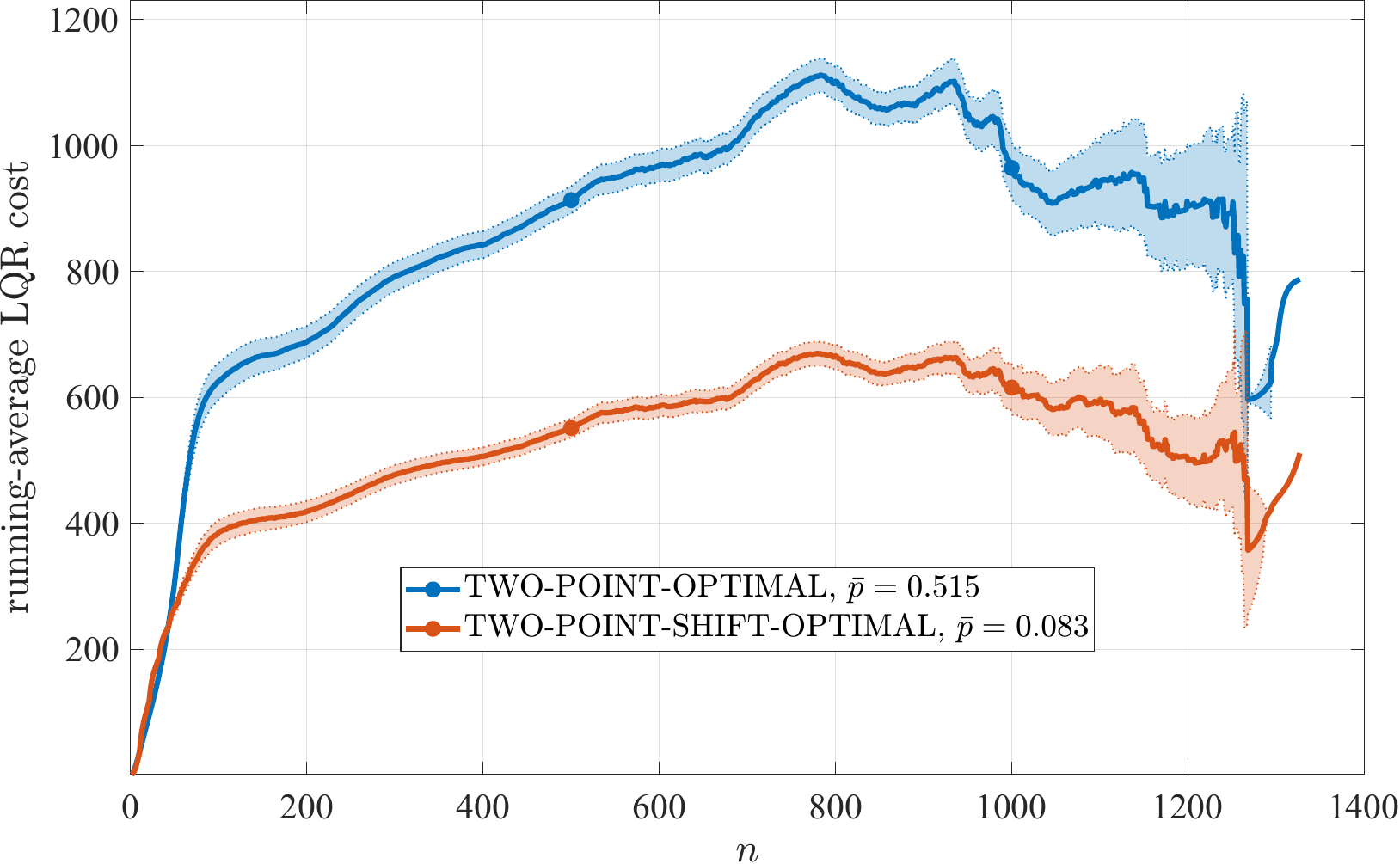}
	\caption{Running-average LQR cost for the original and shifted two-point
		schedulers under correlated US-101 velocity disturbances using the optimal
		feedback controller in~\eqref{eqn:optimal_controller}. The scheduling
		parameters are \(\delta_s=1\), \(\delta_\ell=50\), \(p_\ell=0.02\), shift
		\(s=10\), and \(\tau=1\). The original scheduler operates at achieved rate
		\(\bar p\approx0.515\), whereas the shifted scheduler operates at
		\(\bar p\approx0.083\). Despite its lower communication rate, the shifted
		scheduler achieves a smaller running-average LQR cost.}
	\label{fig:running_avg_corr_two_point_optimal_controller}
\end{figure}
\begin{figure}[!t]
    \centering
    \includegraphics[width=0.65\linewidth]{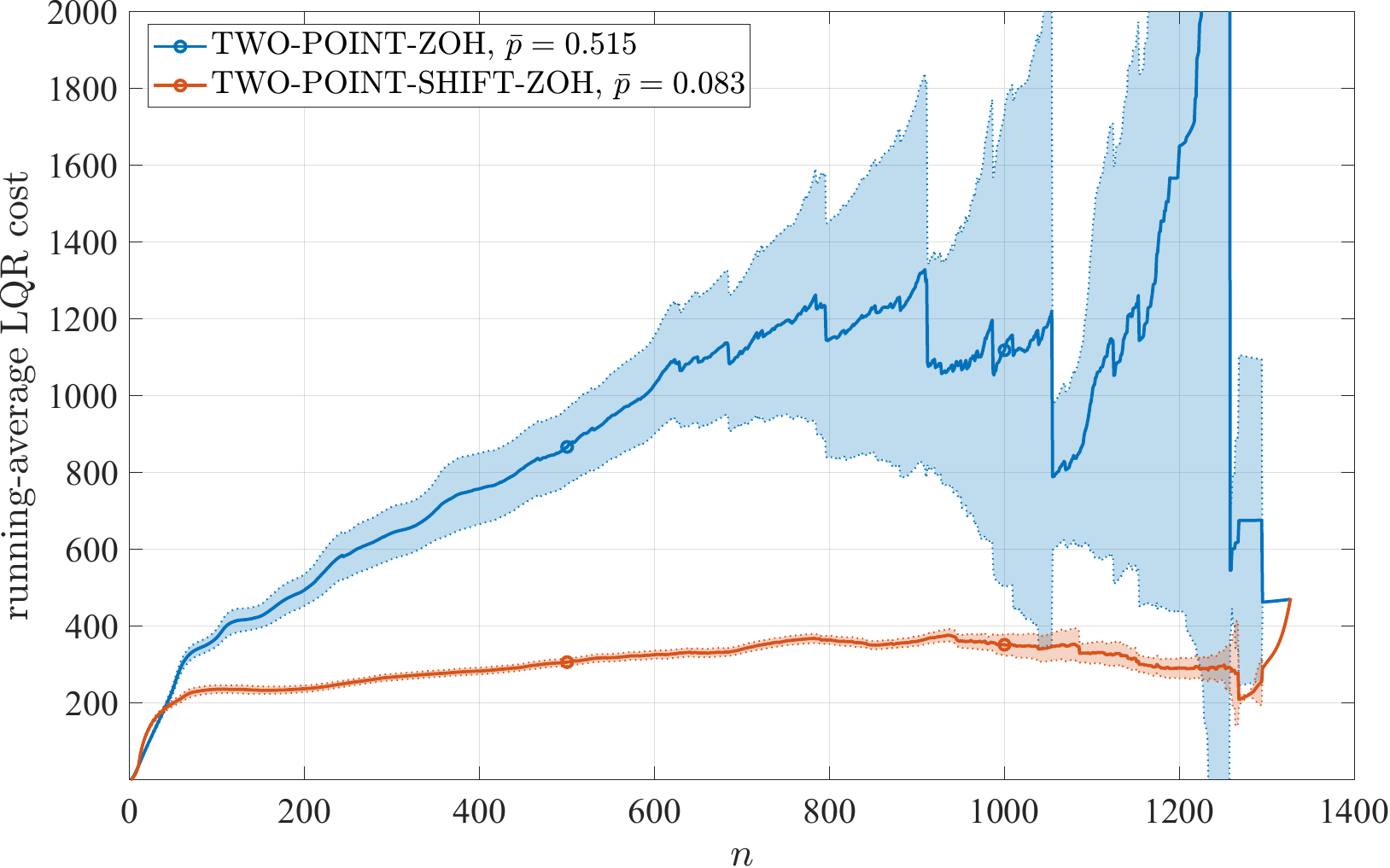}
    \caption{Running-average LQR cost for the original and shifted two-point
    schedulers under correlated US-101 velocity disturbances using the ZOH
    controller. The scheduling parameters are \(\delta_s=1\),
    \(\delta_\ell=50\), \(p_\ell=0.02\), shift \(s=10\), and \(\tau=1\). The
    original scheduler operates at achieved rate \(\bar p\approx0.515\),
    whereas the shifted scheduler operates at \(\bar p\approx0.083\). The
    lower-rate shifted scheduler again achieves a smaller running-average LQR
    cost, with a larger performance gap than in the optimal-controller case.}
    \label{fig:running_avg_corr_two_point_zoh_controller}
\end{figure}

Next, we evaluate how the scheduling policies perform under different mean-normalization methods and feedback delays. Rather than considering only isolated operating points at fixed \(p_{\mathrm{target}}\), we study the steady-state LQR cost as a function of the communication rate \(p_{\mathrm{target}}\). For each value of \(p_{\mathrm{target}}\), we compute the running-average LQR cost for each of the 3900 vehicle trajectories and then average the last 500 samples of these running averages. This produces one steady-state cost value for the corresponding communication rate, controller, scheduling policy, and normalization method. Repeating this procedure over a range of communication rates yields the steady-state LQR cost curves reported below. Throughout this analysis, the LQR weights are set to \(q=r=1\).

\begin{figure}[!t]
    \centering
    \includegraphics[width=0.65\linewidth]{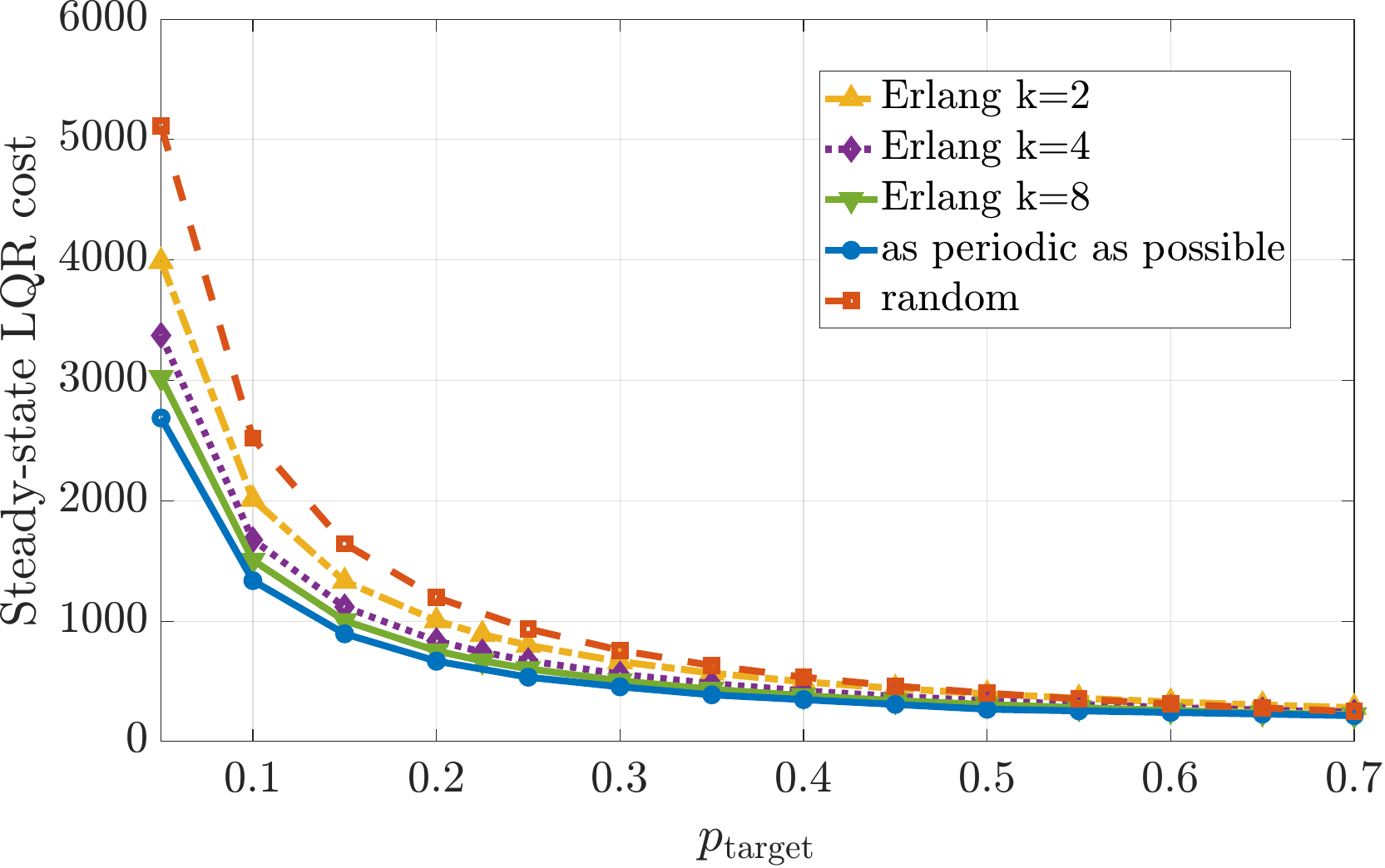}
    \caption{Steady-state LQR cost versus target communication rate $p_{\mathrm{target}}$ for the impulsive (IMP) controller under offline mean normalization with $\tau = 1$.}
    \label{fig:ss_lqr_imp_tau1}
\end{figure}

Figure \ref{fig:ss_lqr_imp_tau1} shows the steady-state LQR cost versus communication rate for the impulsive controller under all state-independent scheduling policies with a fixed delay of $\tau=1$. As expected, the periodic policy achieves the best tracking performance, while the random policy performs the worst. This observation is consistent with the theoretical results in Section \ref{sec:AoI_Correlated_Noise}, where we showed that the equivalent tracking problem depends on minimizing all higher-order moments of the AoI distribution. Since the periodic scheduler minimizes the variability of AoI, it achieves the lowest cost, whereas the random scheduler produces the highest cost.

\begin{figure}[!t]
    \centering
    \includegraphics[width=0.65\linewidth]{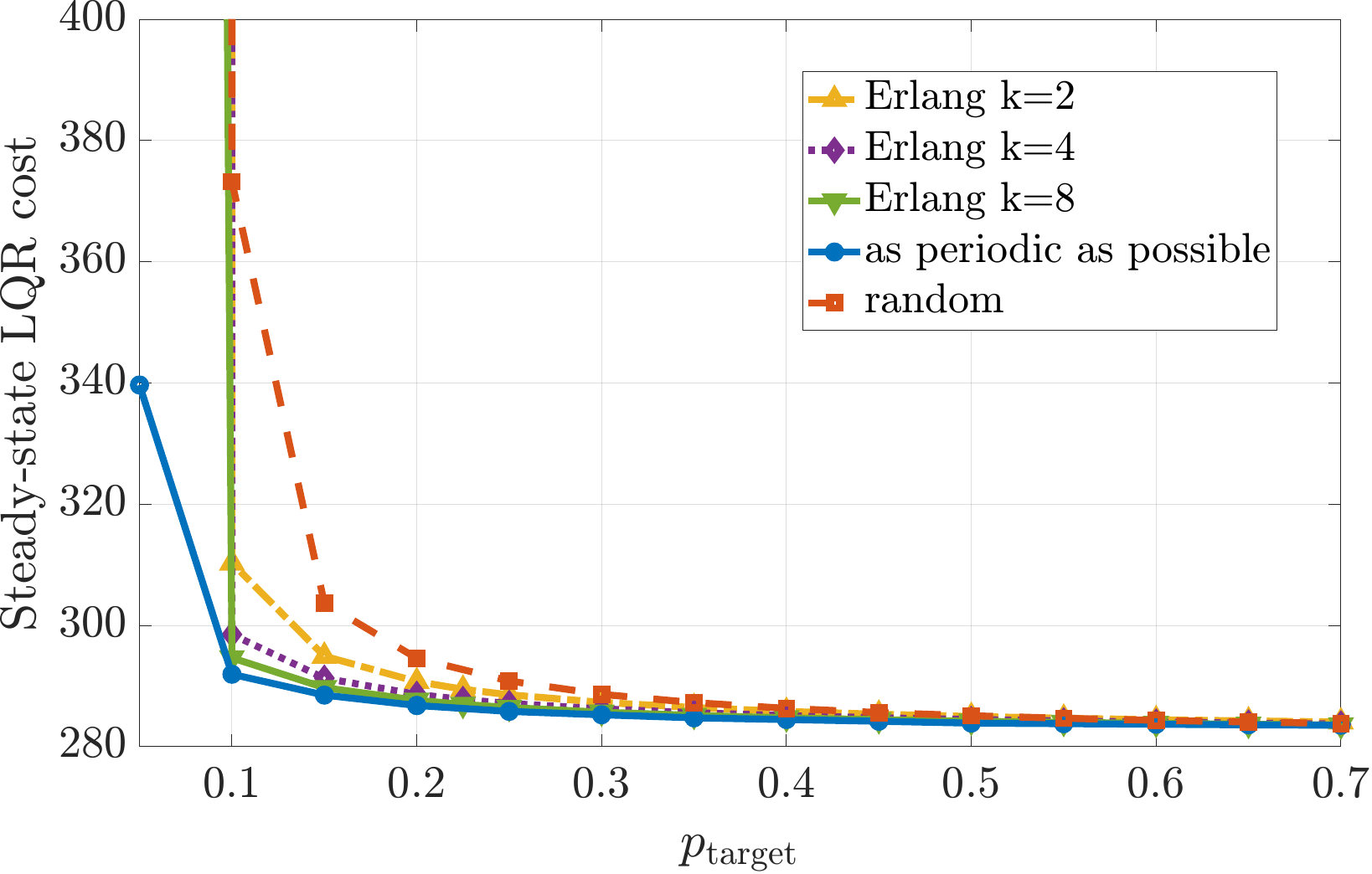}
    \caption{Steady-state LQR cost versus target communication rate $p_{\mathrm{target}}$ for the zero-order-hold (ZOH) controller under offline mean normalization with $\tau = 1$.}
    \label{fig:ss_lqr_zoh_tau1}
\end{figure}

\begin{figure}[!t]
	\centering
	\includegraphics[width=0.65\linewidth]{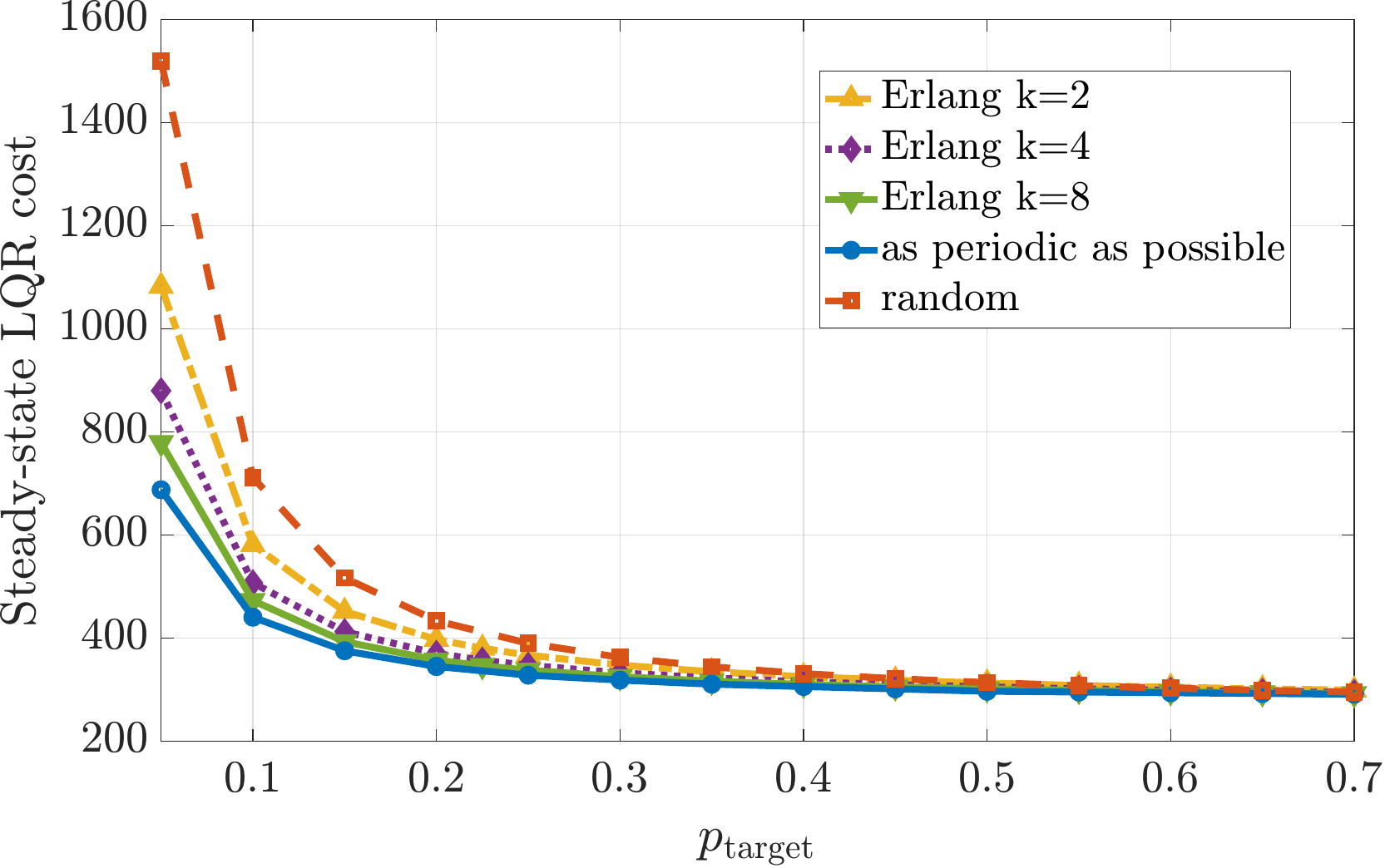}
	\caption{Steady-state LQR cost versus target communication rate $p_{\mathrm{target}}$ for the symmetric controller under offline windowing with $\tau = 1$.}
	\label{fig:ss_lqr_lqr_tau1}
\end{figure}

Figure~\ref{fig:ss_lqr_zoh_tau1} shows the corresponding results for the ZOH controller. Compared with the impulsive-controller results in Fig.~\ref{fig:ss_lqr_imp_tau1}, the ZOH controller generally achieves lower LQR cost. However, beyond a certain communication rate, the closed-loop system becomes unstable.

Figure~\ref{fig:ss_lqr_lqr_tau1} shows the steady-state LQR cost for the symmetric controller. Although its tracking performance is slightly worse than that of the ZOH controller, the symmetric controller remains stable over the entire range of communication rates considered, unlike ZOH. Despite differences in absolute cost across controller types, the ordering of the scheduling policies remains consistent and supports the theoretical analysis in Section~\ref{sec:AoI_Correlated_Noise}: the periodic scheduler consistently achieves the best performance, whereas the random scheduler performs the worst.

Since the ZOH controller becomes unstable for certain communication rates and the results obtained with different controller types are otherwise qualitatively similar, differing primarily in the overall scale of the cost, the remaining plots present only the results for the symmetric controller. In particular, the relative ordering and performance trends of the scheduling policies remain nearly identical across all controller types.

Figure~\ref{fig:lqr_window_comparison} presents the steady-state tracking performance of the symmetric controller under different online mean-normalization methods. In the online normalization framework, the mean estimate in \eqref{eqn:mean_normalization} is computed causally using a windowed estimator over past disturbance observations. Specifically, at time \(n\), the estimated mean is given by
\[
\hat{\mu}_w[n]
=
\frac{\sum_{k=0}^{L-1} \alpha[k]\, w[n-k]}
{\sum_{k=0}^{L-1} \alpha[k]},
\]
where \(\alpha[k]\) denotes the window coefficient sequence and \(L\) is the window length. We consider four windowing methods: rectangular, Hann, Hamming, and exponential.
For the rectangular window,
\[
\alpha[k] = 1,
\qquad 0 \leq k \leq L-1.
\]
For the Hann window,
\[
\alpha[k]
=
\frac{1}{2}
\left(
1 - \cos\left(\frac{2\pi k}{L-1}\right)
\right),
\qquad 0 \leq k \leq L-1.
\]
For the Hamming window,
\[
\alpha[k]
=
0.54 - 0.46 \cos\left(\frac{2\pi k}{L-1}\right),
\qquad 0 \leq k \leq L-1.
\]
Finally, for the exponential window, the weights are generated according to
\begin{equation}
    \alpha[k] = \zeta^{L-1-k},
    \qquad 0 \leq k \leq L-1,
    \label{eqn:exp_window}
\end{equation}
where \(\zeta \in (0,1)\) is the exponential decay parameter. In our simulations, we use \(\zeta=0.9\), and for all methods the window length is fixed to \(L=80\) samples.
\begin{figure}[!ht]
    \centering
    
    \begin{subfigure}{0.48\linewidth}
        \centering
        \includegraphics[width=0.9\linewidth]{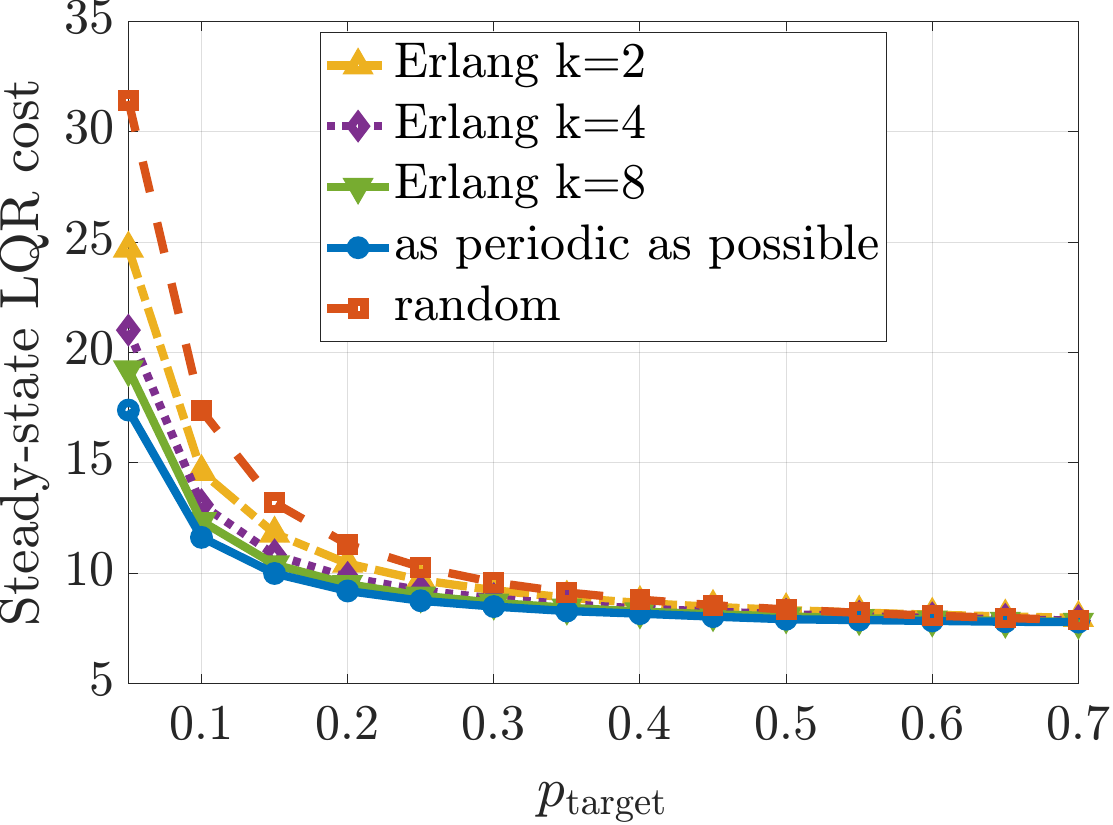}
        \caption{Exponential windowing with $\zeta=0.9$ in \eqref{eqn:exp_window}.}
        \label{fig:lqr_exp}
    \end{subfigure}
    \hfill
    \begin{subfigure}{0.48\linewidth}
        \centering
        \includegraphics[width=0.9\linewidth]{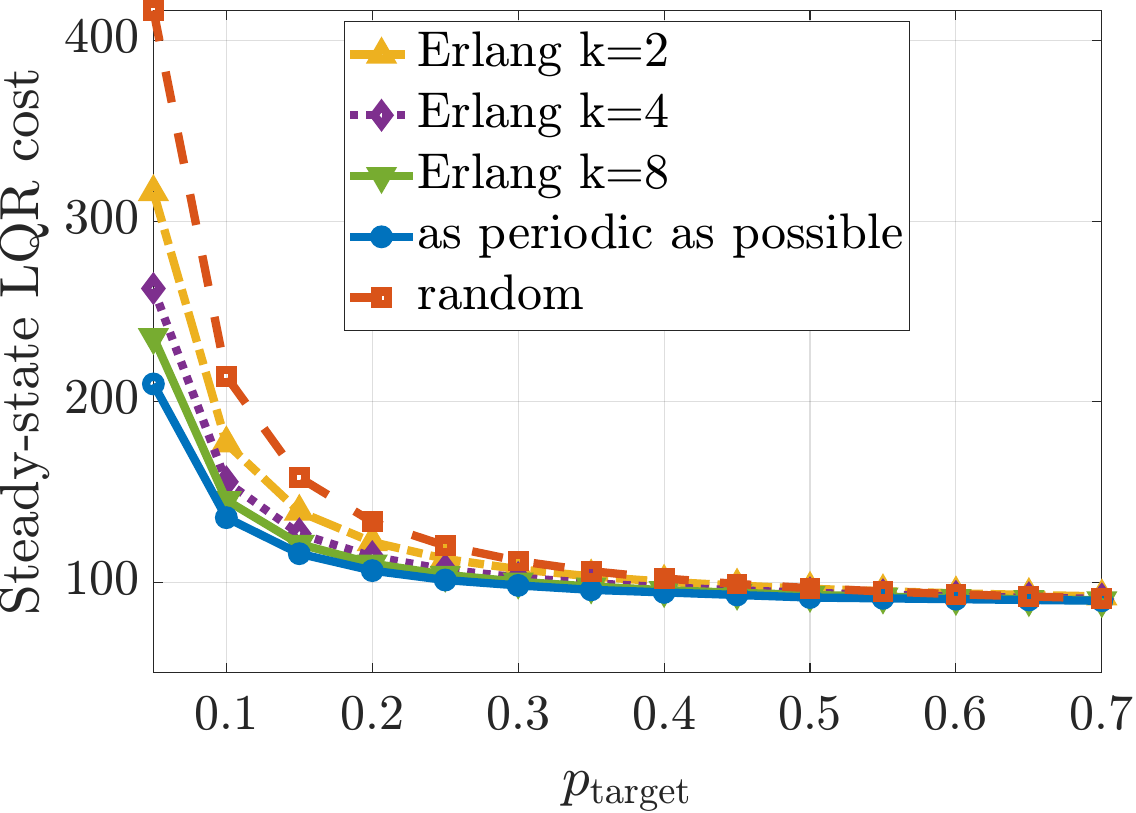}
        \caption{Hann windowing.}
        \label{fig:lqr_hann}
    \end{subfigure}

    \vspace{0.3cm}

    \begin{subfigure}{0.48\linewidth}
        \centering
        \includegraphics[width=0.9\linewidth]{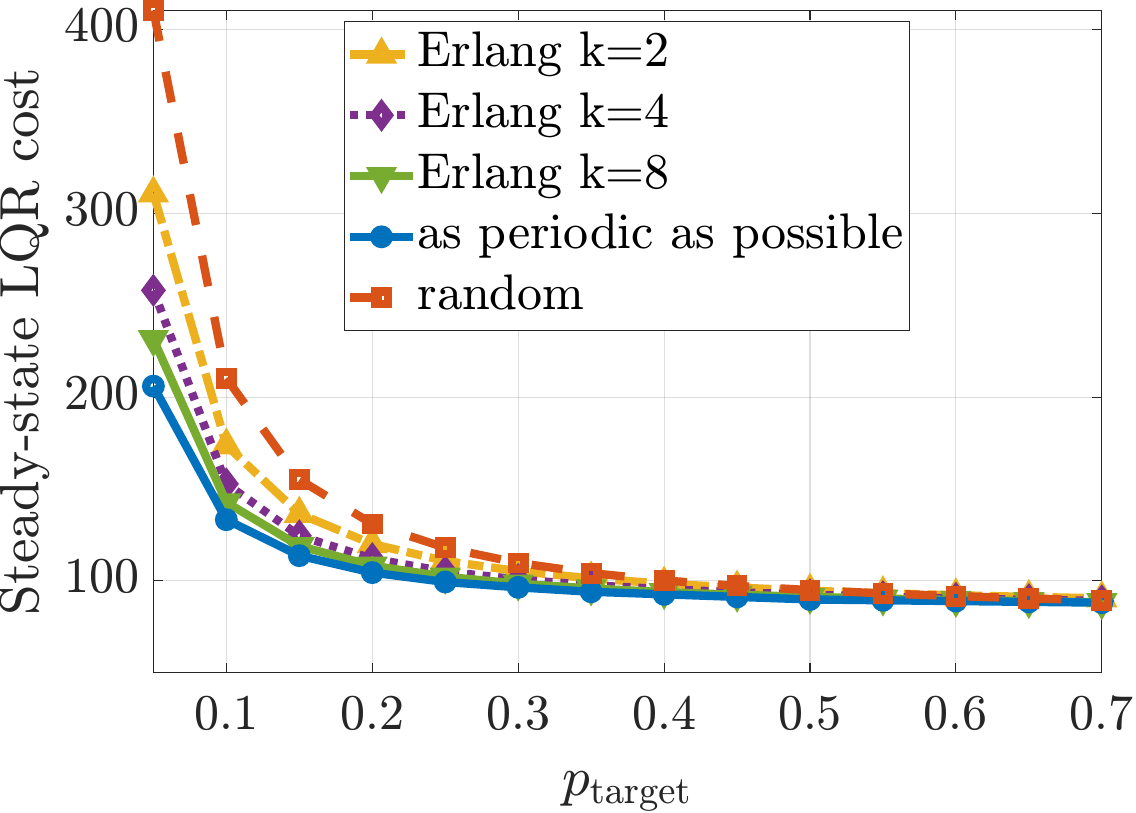}
        \caption{Hamming windowing.}
        \label{fig:lqr_hamming}
    \end{subfigure}
    \hfill
    \begin{subfigure}{0.48\linewidth}
        \centering
        \includegraphics[width=0.9\linewidth]{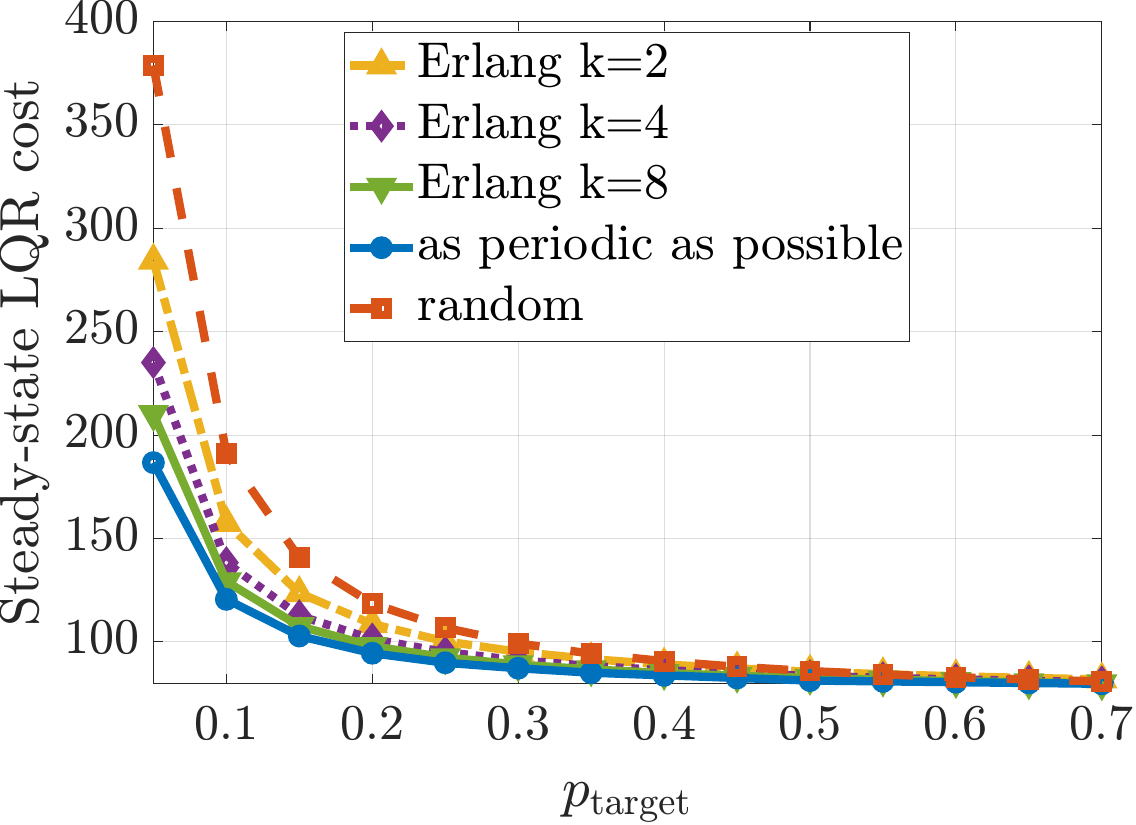}
        \caption{Rectangular windowing.}
        \label{fig:lqr_rect}
    \end{subfigure}

    \caption{Steady-state LQR cost versus target communication rate under different online windowing methods for mean normalization with $\tau = 1$. The window length is fixed to 80 for all methods.}
    \label{fig:lqr_window_comparison}
\end{figure}

As can be seen, the best tracking performance is achieved with exponential windowing. In contrast, the remaining windowing methods produce nearly identical control performance. Moreover, all windowing methods yield qualitatively similar results in the sense that the ordering of the scheduling policies remains unchanged, and the differences in steady-state LQR cost between scheduling policies are also highly similar. Therefore, the choice of windowing method does not significantly alter the overall characteristics of the problem.

To investigate the effect of delay on the control performance, Fig. \ref{fig:lqr_tau5_comparison} presents the results for offline mean normalization and online mean normalization with exponential windowing under a fixed delay of $\tau=5$. Comparing these results with the $\tau=1$ case shown in Fig. \ref{fig:ss_lqr_lqr_tau1} and Fig. \ref{fig:lqr_exp}, we observe that the overall behavior of the curves remains largely unchanged for both offline and online mean normalization methods. The primary difference is a moderate increase in the steady-state LQR cost as the fixed delay increases.

\begin{figure}[!ht]
    \centering
    
    \begin{subfigure}[t]{0.46\linewidth}
        \centering
        \includegraphics[height=0.22\textheight,keepaspectratio]{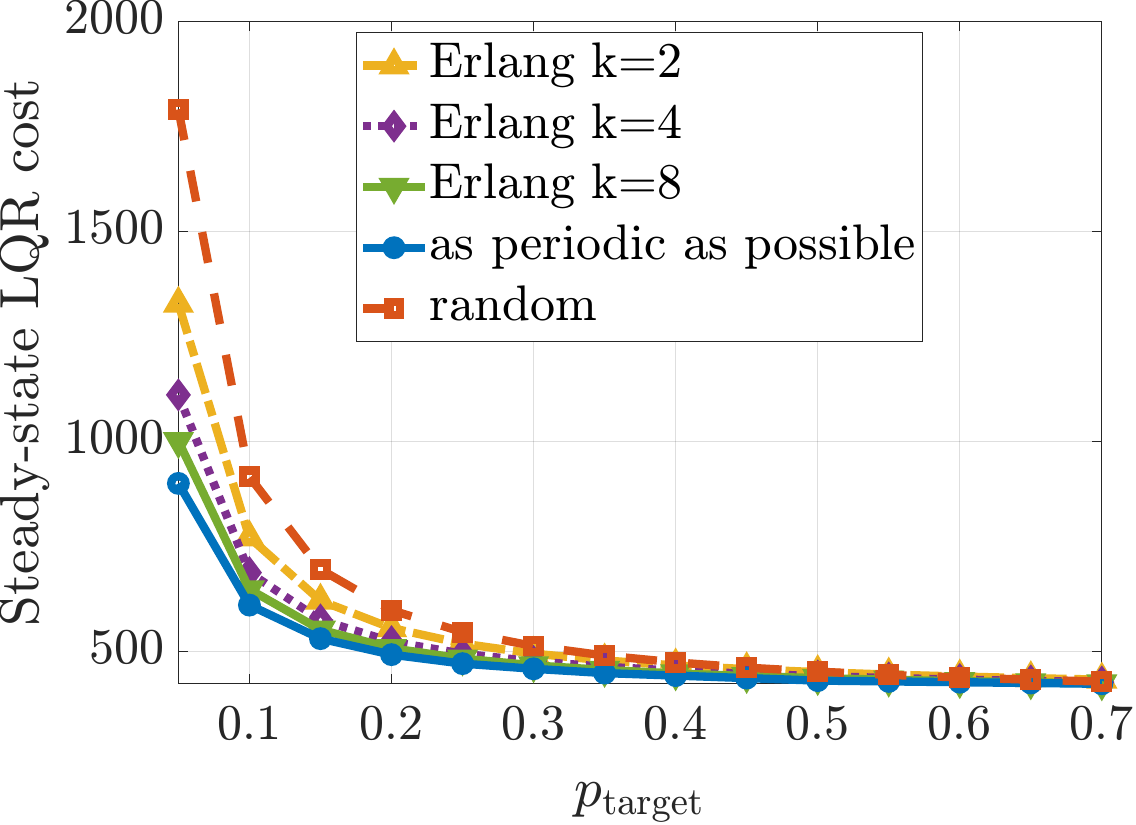}
        \caption{Offline mean normalization.}
        \label{fig:lqr_offline_tau5}
    \end{subfigure}
    \hfill
    \begin{subfigure}[t]{0.46\linewidth}
        \centering
        \includegraphics[height=0.22\textheight,keepaspectratio]{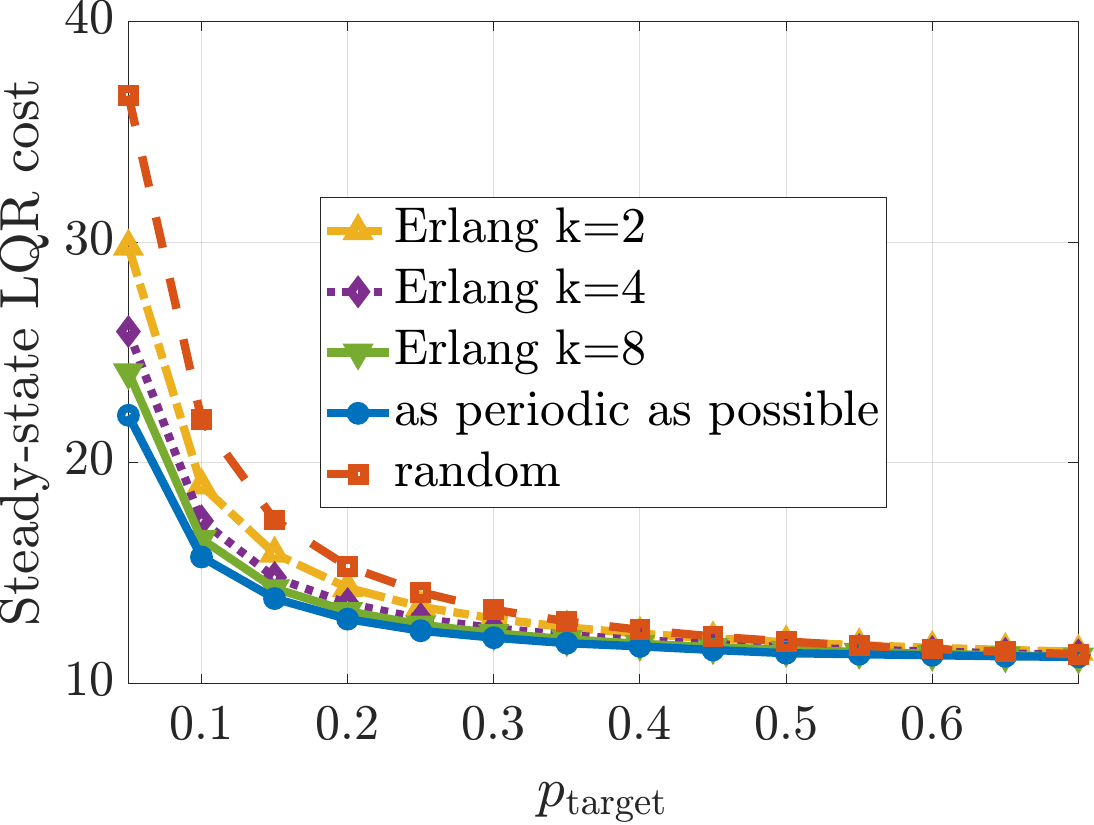}
        \caption{Online mean normalization with exponential windowing.}
        \label{fig:lqr_online_tau5}
    \end{subfigure}

    \caption{Steady-state LQR cost versus target communication rate for $\tau = 5$.}
    \label{fig:lqr_tau5_comparison}
\end{figure}

\begin{figure}[!ht]
    \centering
    \includegraphics[width=0.65\linewidth]{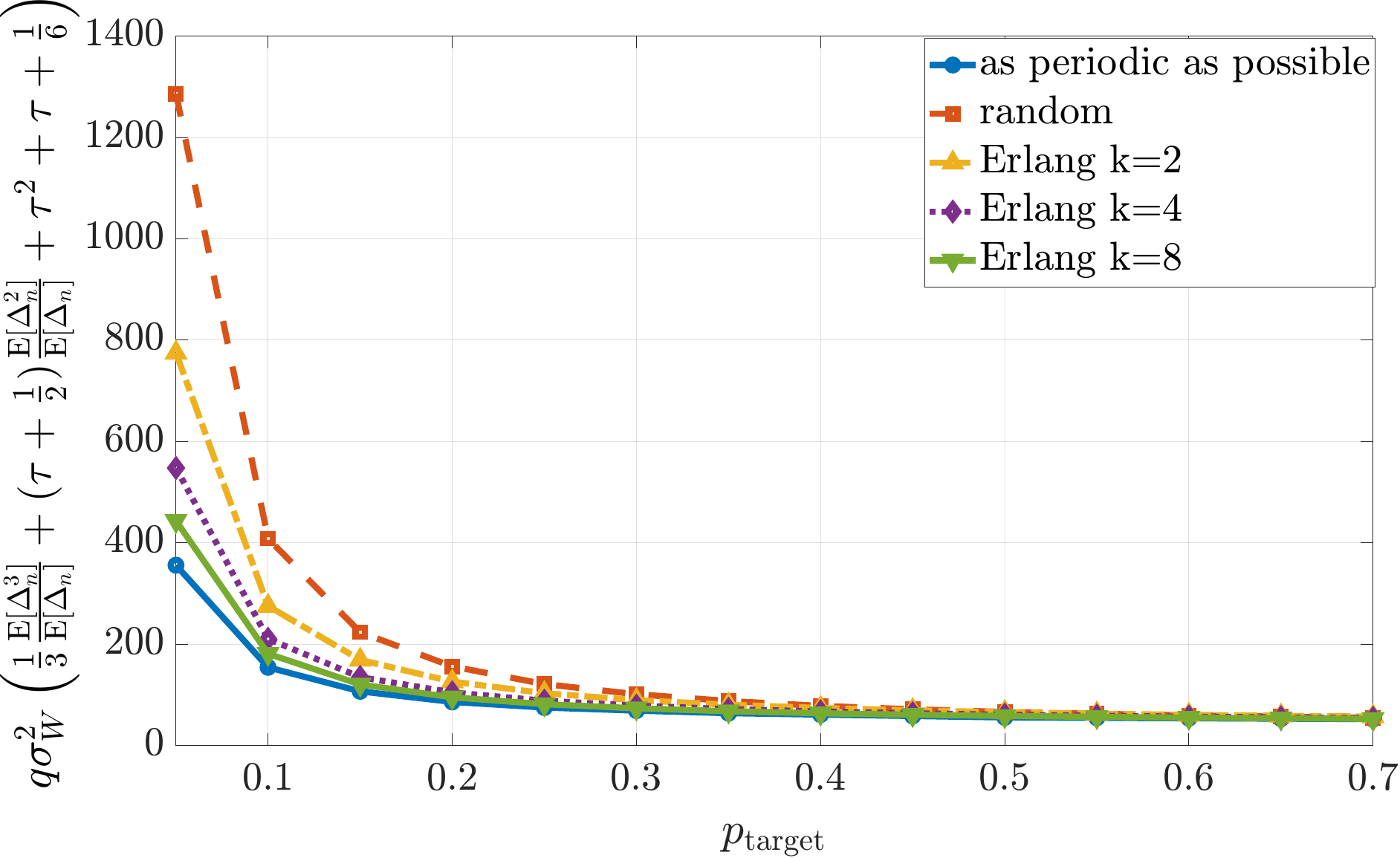}
    \caption{Theoretical steady-state cost expression versus target communication rate based on the LQR equivalent problem $P_{8}$ as a function of the AoI distribution for $\tau = 5$.}
    \label{fig:theoretical_iut_cost}
\end{figure}

We now turn to the analysis of the equivalent problem $P_{8}$, where we assumed $r \to 0$ in the LQR cost together with the impulsive controller. In addition, Fig. \ref{fig:beta_mse} showed that the disturbance process extracted from the US-101 dataset is highly correlated, corresponding to the regime $\beta \to 0$. Therefore, in this section, we compare the equivalent cost expression of problem $P_{8}$, written as a function of the moments of the AoI distribution, with the practical car-following simulation, using the impulsive controller with $q=1$ and $r=0.01$.

\begin{figure}[!t]
    \centering
    \includegraphics[width=0.65\linewidth]{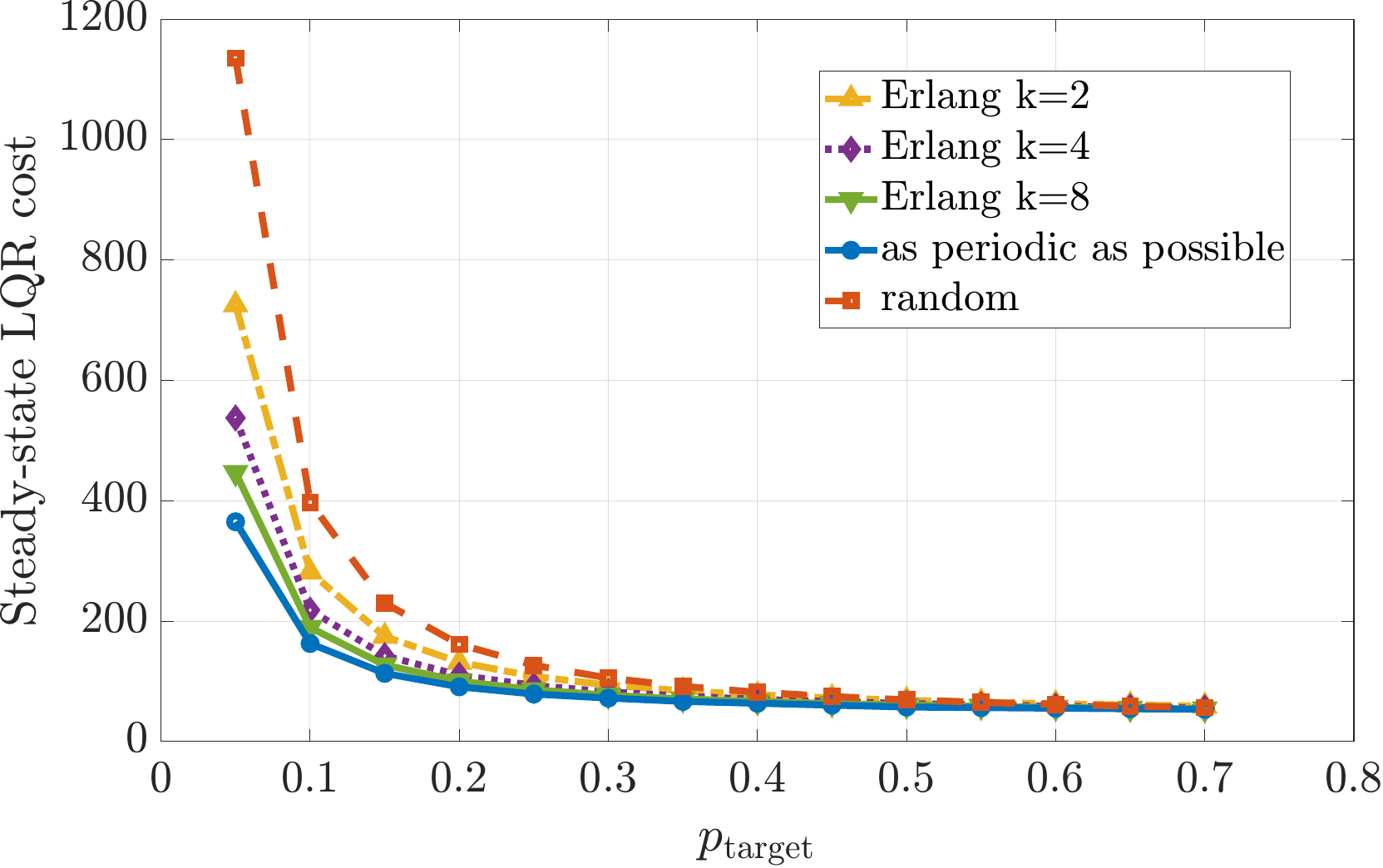}
    \caption{Steady-state LQR cost versus target communication rate $p_{\mathrm{target}}$ for the impulsive (IMP) controller under offline mean normalization with $\tau = 5$ and $r = 0.01$.}
    \label{fig:ss_lqr_imp_tau5_r001}
\end{figure}

\begin{figure}[!t]
    \centering
    \includegraphics[width=0.65\linewidth]{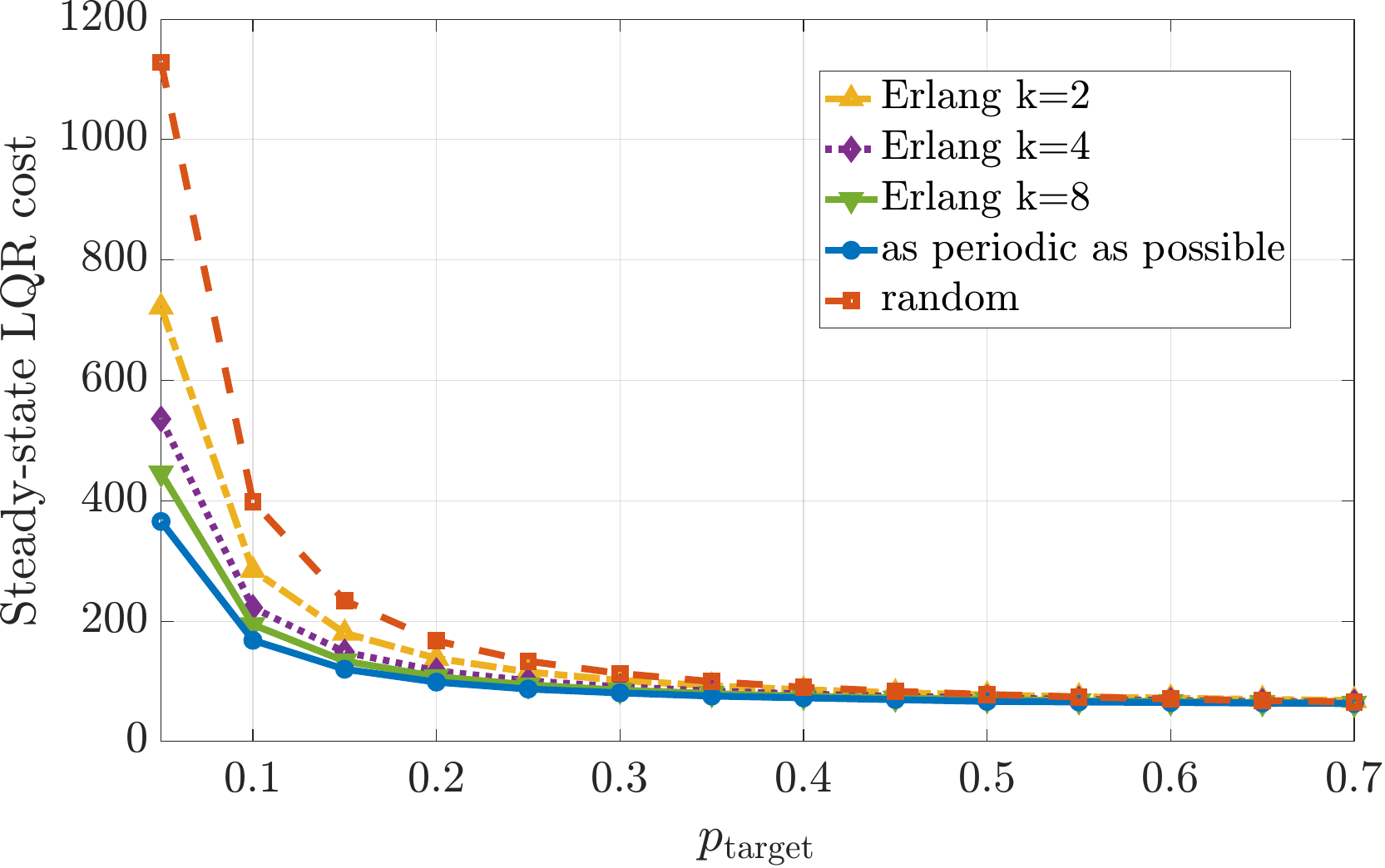}
    \caption{Steady-state LQR cost versus target communication rate $p_{\mathrm{target}}$ for the symmetric controller under offline mean normalization with $\tau = 5$ and $r = 0.01$.}
    \label{fig:ss_lqr_lqr_tau5_r001}
\end{figure}

\begin{figure}[!t]
	\centering
	\includegraphics[width=0.65\linewidth]{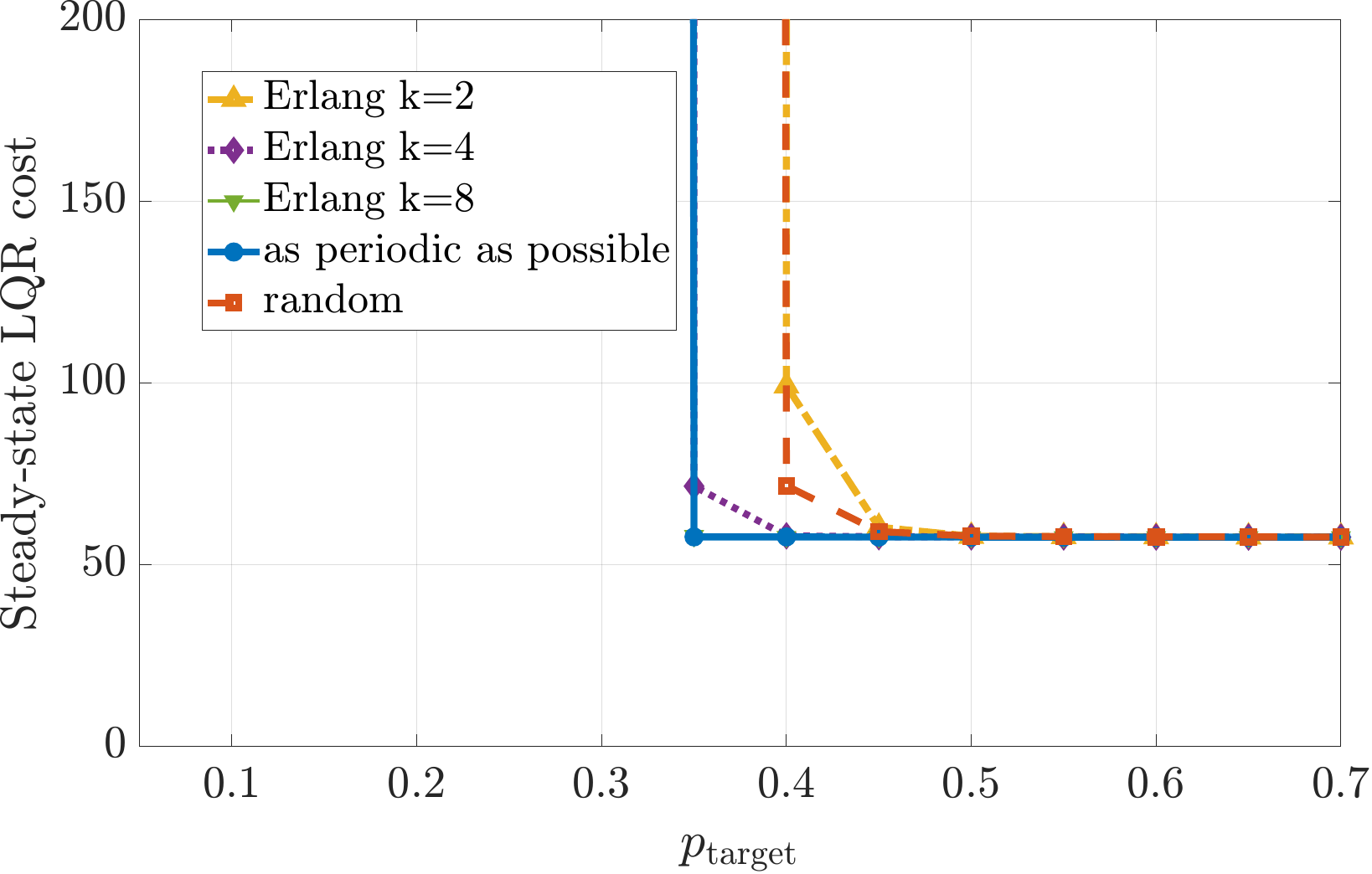}
	\caption{Steady-state LQR cost versus target communication rate $p_{\mathrm{target}}$ for the zero-order-hold (ZOH) controller under offline scheduling with $\tau = 5$ and $r = 0.01$.}
	\label{fig:ss_lqr_zoh_tau5_r001}
\end{figure}
First, Fig. \ref{fig:theoretical_iut_cost} illustrates the behavior of the equivalent cost in problem $P_{8}$ for varying communication rates, different AoI distributions, and a fixed delay of $\tau=5$. Figure \ref{fig:ss_lqr_imp_tau5_r001}, on the other hand, shows the steady-state LQR cost obtained from the practical simulation using the impulsive controller with $r=0.01$. The similarity between these two figures is immediately apparent and closely matches the theoretical predictions. Before comparing these figures in detail, we first discuss how the steady-state LQR cost changes under different controller structures. Figure \ref{fig:ss_lqr_lqr_tau5_r001} shows the steady-state LQR cost versus communication rate under offline mean normalization. Comparing this figure with Fig. \ref{fig:lqr_offline_tau5}, we observe that the costs increase by approximately $50\%$ for all AoI distributions, while the overall characteristics and relative ordering of the curves remain unchanged.

Figure \ref{fig:ss_lqr_zoh_tau5_r001} shows the steady-state LQR cost for the ZOH controller with $r=0.01$ and a fixed delay of $\tau=5$. Similar to the previous observations, the ZOH controller again becomes unstable beyond a certain communication rate. However, in this case, the transition to instability is significantly more abrupt. This behavior indicates that as the delay increases and the control penalty decreases, the stability margin of the system rapidly vanishes.

\begin{figure}[!t]
    \centering
    \includegraphics[width=0.65\linewidth]{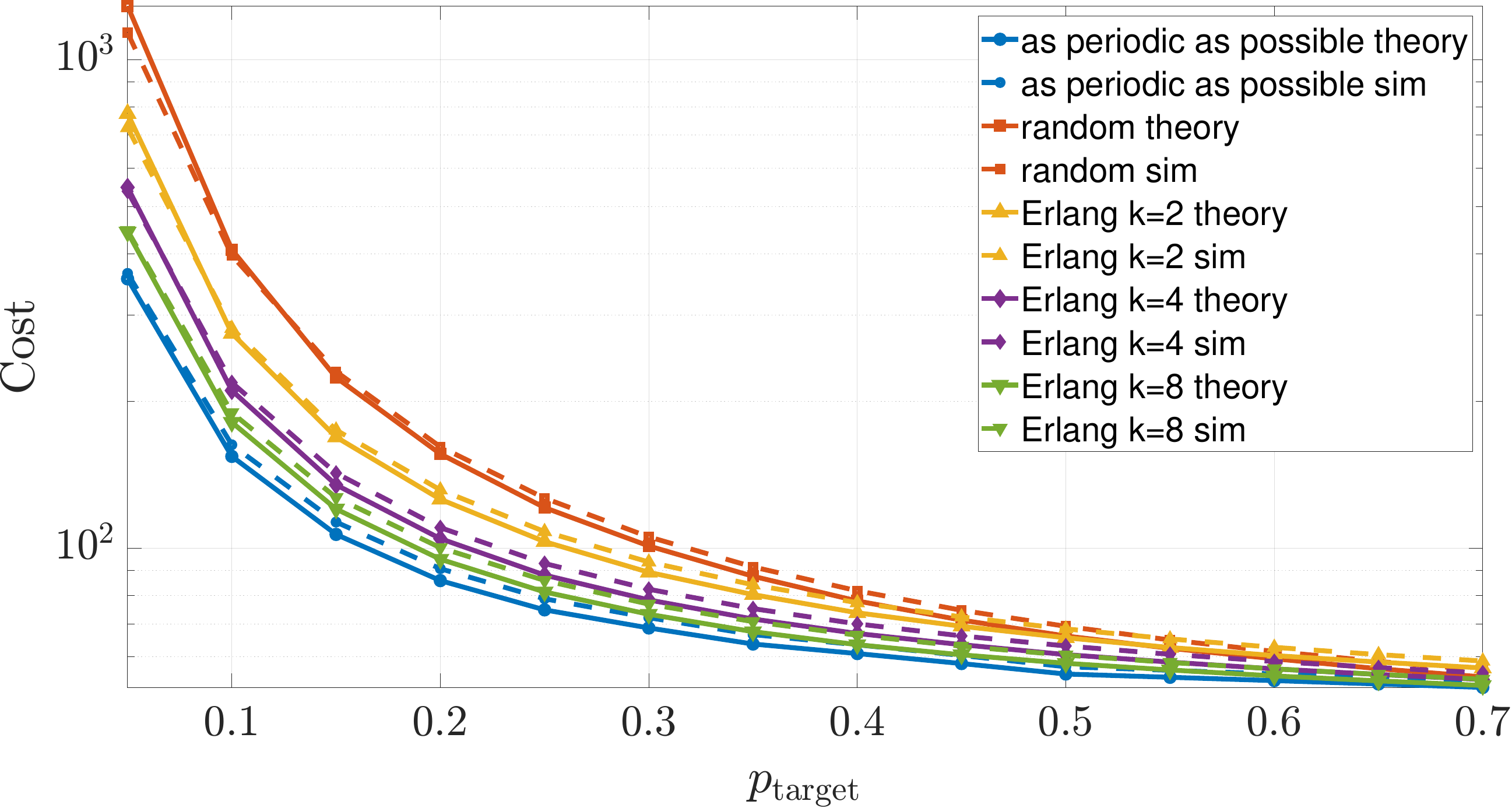}
    \caption{Comparison between the theoretical and simulated steady-state LQR costs for the impulsive (IMP) controller under offline scheduling with $\tau = 5$.}
    \label{fig:theory_vs_sim_imp}
\end{figure}

\begin{figure}[!t]
	\centering
	\includegraphics[width=0.65\linewidth]{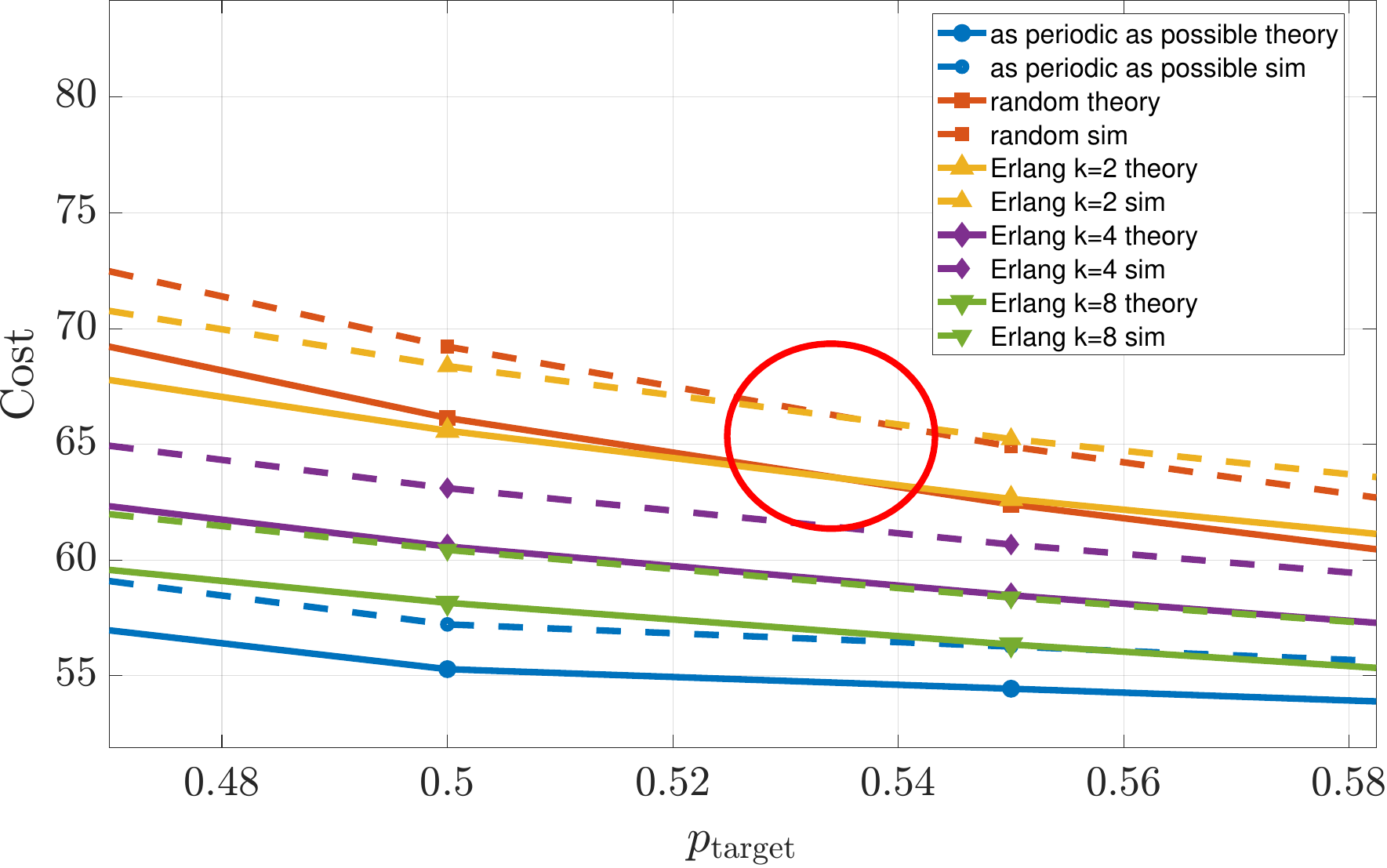}
	\caption{Zoomed comparison of the theoretical steady-state cost expression obtained from the equivalent problem $P_{8}$ and the practical steady-state LQR cost obtained from the car-following simulation under the impulsive controller with offline mean normalization and $\tau=5$. The figure highlights the crossover point around $p \approx 0.53$, where the random scheduling policy begins to outperform the Erlang-$2$ distribution in both the theoretical analysis and the practical simulation.}
	\label{fig:theory_vs_sim_zoomed_0_5}
\end{figure}

We now arrive at the main result of this section, where we compare the steady-state LQR cost obtained from the practical car-following simulation against the equivalent cost expression derived in \eqref{eqn:eqv_cost_corr_noise_a_1_beta_0}. The comparison is shown in Fig. \ref{fig:theory_vs_sim_imp}. To facilitate a more detailed comparison, the $y$-axis is plotted on a logarithmic scale. As can be seen, the theoretical curves (solid lines) closely follow the simulation results (dotted lines) across all communication rates and AoI distributions considered. This agreement strongly validates the theoretical analysis. We further investigate this behavior by zooming into the region around $p=0.5$, as shown in Fig. \ref{fig:theory_vs_sim_zoomed_0_5}. Interestingly, although the random scheduling policy yields the worst performance for communication rates below approximately $0.53$, it outperforms the Erlang-$2$ distribution beyond this point. The same phenomenon appears in both the AoI-equivalent problem ($P_{8}$) and the practical simulation at nearly identical communication rates. Only a small offset between the theoretical and simulation costs is observed. This discrepancy is expected since the theoretical analysis assumes both $r \to 0$ and $\beta \to 0$ as approximations. Overall, these results demonstrate that the approximations used in the analysis are highly accurate and that the resulting AoI-equivalent problem provides an excellent characterization of the practical system.

The agreement between theory and simulation also reveals a more subtle consequence of the distributional nature of the cost. If mean AoI were sufficient, increasing the communication rate would necessarily improve performance across policies. The zoomed comparisons below show that this monotonic intuition can fail once the full AoI distribution is taken into account. To this end, we zoom into the region around $p_{\mathrm{target}}=0.2$ in Fig. \ref{fig:theory_vs_sim_imp}, as shown in Fig. \ref{fig:theory_vs_sim_zoomed_0_2}. As can be seen, increasing $p_{\mathrm{target}}$ does not necessarily reduce the cost in \eqref{eqn:eqv_cost_corr_noise_a_1_beta_0}. In particular, the cost of the periodic scheduling policy with communication rate $p_{\mathrm{target}}=0.2$ is lower than that of the random scheduling policy with $p_{\mathrm{target}}=0.25$. Recall that $p_{\mathrm{target}}$ corresponds to the inverse of the mean AoI. Therefore, these results clearly show that reducing the mean AoI alone is insufficient, since the cost depends on the entire AoI distribution rather than only its mean.

\begin{figure}[!t]
    \centering
    \includegraphics[width=0.65\linewidth]{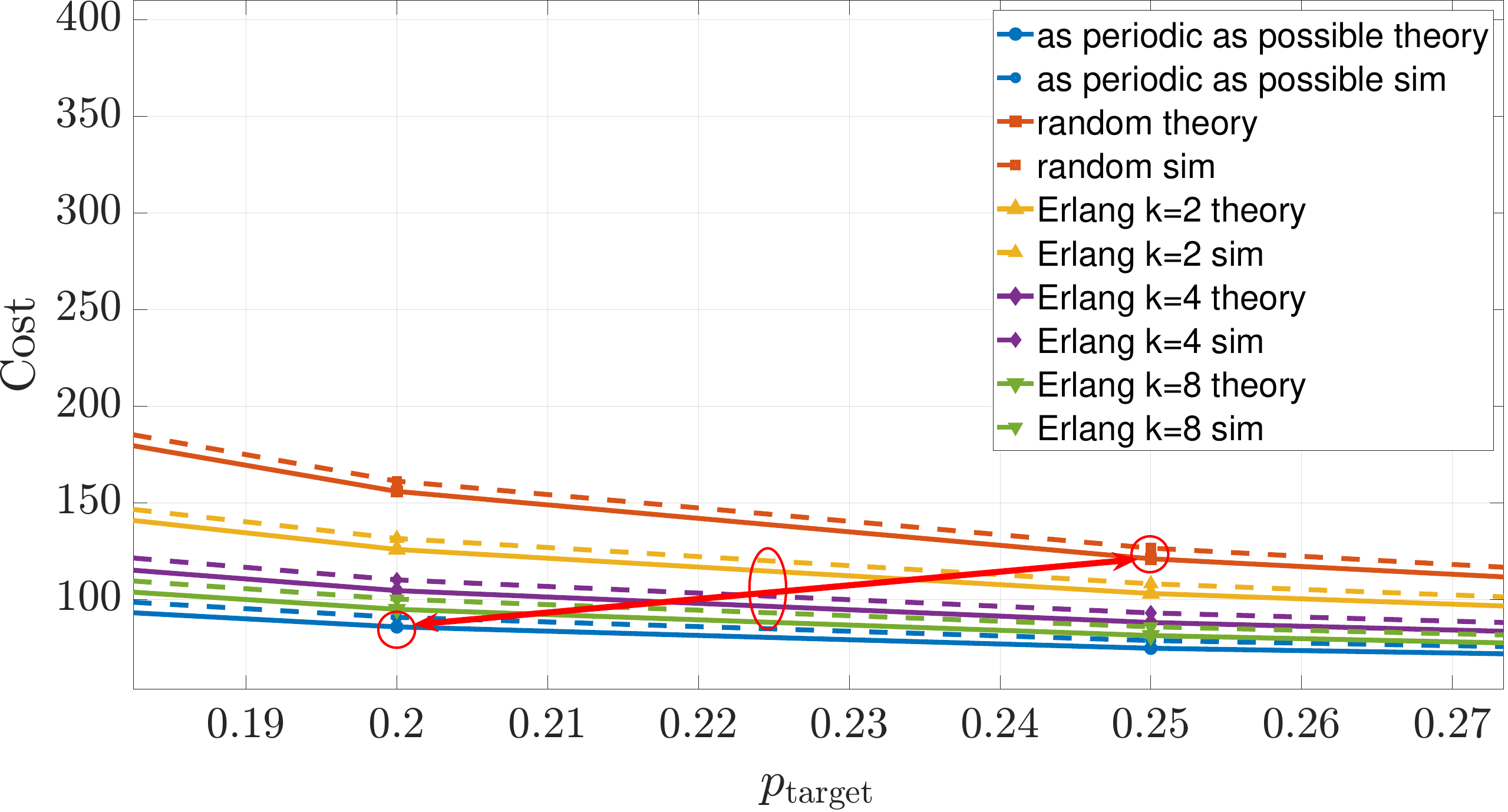}
    \caption{Zoomed comparison of the theoretical and simulated steady-state costs around $p_{\mathrm{target}}=0.2$. Despite operating at a lower communication rate, the periodic scheduling policy achieves a lower cost than the random scheduling policy, highlighting that the control performance depends on the entire AoI distribution rather than only the mean AoI.}
    \label{fig:theory_vs_sim_zoomed_0_2}
\end{figure}

\begin{figure}[!t]
    \centering
    \includegraphics[width=0.65\linewidth]{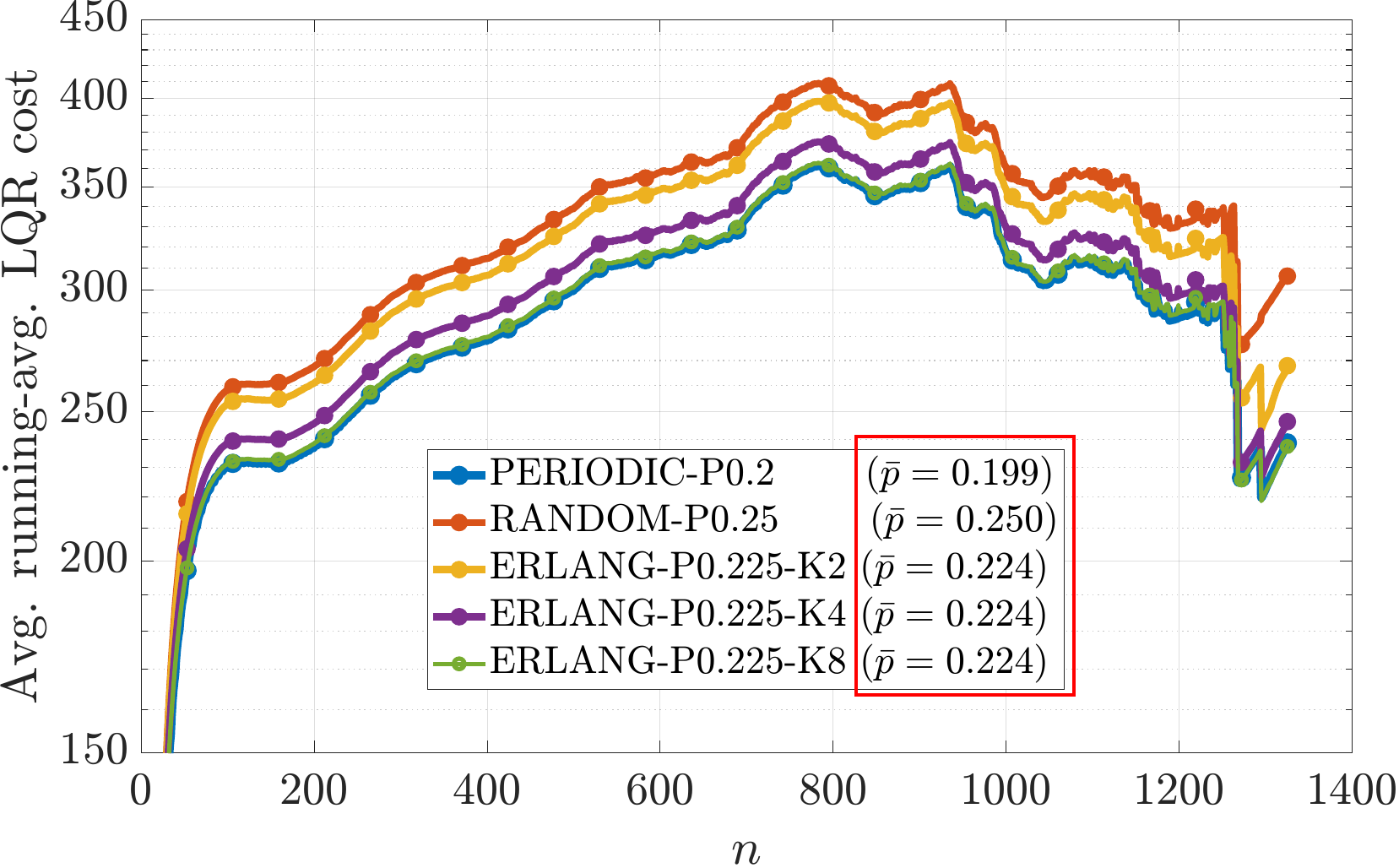}
    \caption{Average running-average LQR cost across all tracks for different scheduling policies operating at different communication rates under offline mean normalization.}
    \label{fig:avg_running_cost_iut_matched}
\end{figure}

Moreover, we observe that the cost in \eqref{eqn:eqv_cost_corr_noise_a_1_beta_0} becomes nearly identical for the Erlang and periodic AoI distributions at communication rates $p_{\mathrm{target}}=0.225$ and $p_{\mathrm{target}}=0.2$, respectively. To further validate this observation, Fig.~\ref{fig:avg_running_cost_iut_matched} shows the average running-average LQR cost at these operating points, where $\bar{p}$ denotes the achieved communication rate of each policy. As expected, the same behavior is observed in the practical simulation results.

These results make the central message clear: mean AoI alone is not a control-relevant objective. What matters is the full AoI distribution. Ignoring that distribution can lead to schedulers that are fresher on average but worse for closed-loop control.

\section{CONCLUSION}
This paper examined whether minimizing mean Age of Information is an optimal design principle for networked control systems. Starting from an infinite-horizon LQR tracking formulation, we showed that under state-independent scheduling the control problem reduces to an optimization over the inter-scheduling interval distribution. The resulting objective is not determined by mean AoI alone. For i.i.d. disturbances, the LQR-induced cost depends on higher-order moments, and under unstable dynamics or temporally correlated disturbances it depends on exponential moments of the interval distribution. Thus, two policies with the same mean AoI can yield substantially different tracking performance.

The analysis also identifies the as-periodic-as-possible scheduler as the optimal state-independent policy under rate-limited channels. This result provides a control-theoretic explanation for why variability in scheduling intervals is harmful: long inter-scheduling gaps are amplified by the plant dynamics and by the temporal structure of the disturbance process. In this sense, the relevant freshness object is the entire AoI distribution, not a scalar average or peak metric.

The theoretical conclusions were validated using data-driven car-following simulations based on real NGSIM US-101 highway trajectories. The empirical velocity profiles exhibit strong temporal correlation, with the fitted exponential correlation parameter \(\beta=0.007558\), placing the system in the regime where the derived exponential-moment terms become significant. The simulated steady-state LQR costs closely followed the theoretical equivalent costs across scheduling policies and communication rates. These observations directly confirm that reducing mean AoI alone is insufficient for control-oriented network design.

\bibliographystyle{IEEEtran}
\bibliography{references_allerton.bib}

\end{document}